\documentclass[sigconf]{acmart}

\renewcommand\footnotetextcopyrightpermission[1]{} % removes footnote with conference information in first column
\usepackage{xspace}
\usepackage{amsfonts}
\usepackage{thm-restate}
\usepackage{amsmath}
\usepackage{ulem}
\usepackage{misc}   %custom for logic

\usepackage{color}
\usepackage{graphicx}
\usepackage{tikz}
\usepackage{tikz-cd}
\usetikzlibrary{calc}

\usepackage{amsthm}
\usepackage[shortcuts]{extdash} %for \=/

\usepackage{hyperref}
\usepackage{algorithm}
\usepackage[noend]{algpseudocode}

\makeatletter 
\DeclareRobustCommand*\cal{\@fontswitch\relax\mathcal}
\makeatother

\AtBeginDocument{%
  \providecommand\BibTeX{{%
    \normalfont B\kern-0.5em{\scshape i\kern-0.25em b}\kern-0.8em\TeX}}}

\setcopyright{iw3c2w3}
\copyrightyear{2018}
\acmYear{2018}
\acmDOI{10.1145/1122445.1122456}

\acmConference{Technical Report}{arXiv version}{2022}
\acmBooktitle{Woodstock '18: ACM Symposium on Neural Gaze Detection,
  June 03--05, 2018, Woodstock, NY}
\acmPrice{15.00}
\acmISBN{978-1-4503-XXXX-X/18/06}

\renewcommand{\dom}{\mn{adom}}
\newtheorem{claim}{Claim}[section]
\newtheorem*{lemma*}{Lemma}
\newcommand{\hideAlgorithm}[1]{}
\renewcommand{\hideAlgorithm}[1]{#1}

\tikzstyle{quantified}=[circle, draw, scale=.6]
\tikzstyle{answer} = [circle, fill, scale=.6]

\newbool{arxive}
\boolfalse{arxive}
\booltrue{arxive}

\begin{document}

\title{Efficiently Enumerating Answers to Ontology-Mediated Queries}


%
%\author{Anonymous}
%
%\author{Anonymous}

\author{Carsten Lutz}
\affiliation{%
 \institution{Institute of Computer Science\\University of Leipzig}
 \streetaddress{1 Th{\o}rv{\"a}ld Circle}
\city{Leipzig}
\country{Germany}}
\email{clu@informatik.uni-leipzig.de}

\author{Marcin Przyby{\l}ko}
\affiliation{%
 \institution{Institute of Computer Science\\University of Leipzig}
 \streetaddress{1 Th{\o}rv{\"a}ld Circle}
 \city{Leipzig}
 \country{Germany}}
\email{przybyl@informatik.uni-leipzig.de}

%\author{\ }
%\affiliation{\ }
% \affiliation{%
%     \institution{Unknown}
%     %  \streetaddress{1 Th{\o}rv{\"a}ld Circle}
% %    \city{Bremen}
% %    \country{Unknown}
% }

% \author{Author 2}
% \affiliation{%
%     \institution{Unknown}
%     %  \streetaddress{1 Th{\o}rv{\"a}ld Circle}
%     %    \city{Bremen}
%  %       \country{Unknown}
% }

\renewcommand{\shortauthors}{Lutz and Przyby{\l}ko}
%\renewcommand{\shortauthors}{Authors}

%\begin{abstract}
%  tbd
%\end{abstract}



\begin{abstract}
  We study the enumeration of answers to ontology-mediated queries
  (OMQs) where the ontology is a set of guarded TGDs or formulated in
  the description logic \ELI and the query is a conjunctive query
  (CQ). In addition to the traditional notion of an answer, we propose and
  study two novel notions of partial answers that can take into
  account nulls generated by existential quantifiers in the
  ontology.  % We also consider the
  % related problems of testing a single answer in linear time and of
  % testing multiple answers in constant time after linear time
  % preprocessing (`all-testing in CD$\circ$Lin').
  Our main result is that enumeration of the traditional complete answers and of both
  kinds of partial answers is possible with linear-time preprocessing
  and constant delay for OMQs that are both acyclic and free-connex
  acyclic. We also provide partially matching lower bounds.  Similar
  results are obtained for the related problems of testing a single
  answer in linear time and of testing multiple answers in constant
  time after linear time preprocessing. In both cases, the border
  between tractability and intractability is characterized by
  similar, but slightly different acyclicity properties.
\end{abstract}

\maketitle

\section{Introduction}
\label{sect:intro}

In knowledge representation, ontologies are an important means for
injecting domain knowledge into an application. In the context of
databases, they give rise to ontology\=/mediated queries (OMQs) which
enrich a traditional database query such as a conjunctive query (CQ)
with an ontology. OMQs aim at querying incomplete data, using the
domain knowledge provided by the ontology to derive additional
answers. In addition, they may enrich the vocabulary available for
query formulation with relation symbols that are not used explicitly
in the data. Popular choices for the ontology language include
(restricted forms of) tuple-generating dependencies (TDGs), also
dubbed existential rules \cite{DBLP:conf/ijcai/BagetMRT11} and
Datalog$^\pm$ \cite{DBLP:journals/ws/CaliGL12}, as well as various
description logics \cite{baader-introduction-to-dl}.

The complexity of evaluating OMQs has been the subject of intense
study, % that has unveiled several interesting aspects such as a tight
% connection to the complexity of constraint satisfaction problems
% (CSPs) with a fixed template
with a focus on \emph{single-testing} as the mode
of query evaluation: given an ontology-mediated query (OMQ)
$Q$, a database~$D$, and a candidate answer $\bar a$, decide
whether $\bar a \in Q(D)$ \cite{AbHV95,barcelo_omq_limits-g,
  bienvenu-answering-omq, bienvenu-ontology-disjunctive-datalog}.
In many applications, however, it is not realistic to
assume that a candidate answer is available. This has led database
theoreticians and practitioners to investigate more relevant modes of
query evaluation such as \emph{enumeration}: given $Q$ and
$D$, generate all answers in $Q(D)$, one after the other and without
repetition.

\newcommand{\cdlin}{\mn{CD}$\circ$\mn{Lin}\xspace}
\newcommand{\dlc}{$\mn{DelayC}_{\mn{lin}}$\xspace} 
The first main aim of this paper is to initiate a study of efficiently
enumerating answers to OMQs. We consider enumeration algorithms that
have a preprocessing phase in which data structures are built that are
used in the subsequent enumeration phase to produce the actual
output. With `efficient enumeration', we mean that preprocessing may
only take time linear in $O(||D||)$ while the delay between two
answers must be constant, that is, independent of $D$. One may or may
not impose the additional requirement that, in the enumeration phase,
the algorithm may consume only a constant amount of memory on top of
the data structures precomputed in the preprocessing phase. We refer
to the resulting enumeration complexity classes as \dlc and \cdlin,
the former admitting unrestricted (polynomial) memory consumption; the use of
these names in the literature is not consistent, we follow
\cite{segoufin-enum,carmeli-enum-ucqs}.
%
% These requirements give rise to
% the enumeration complexity class \dlc, see for example
% \cite{carmeli-enum-ucqs}.  One may impose the additional requirement
% that, in the enumeration phase, the algorithm can use the data
% structures precomputed in the enumeration phase, but otherwise consume
% only a constant amount of memory. The more restricted enumeration
% complexity class obtained in this way is called \cdlin
% \cite{segoufin-enum}.
%
% the
% enumeration complexity class CD$\circ$Lin which requires that
% preprocessing takes linear time $O(||D||)$ while the delay between two answers
% must be constant, that is, independent of $D$.
% Note that there is no restriction on how the running time of the
% preprocessing and how the delay depends on $Q$, which corresponds to
% \emph{data complexity} in single-testing where $Q$ is fixed and thus
% of constant size.
Without ontologies, answer enumeration in \cdlin and in \dlc has received
significant attention
\cite{DBLP:journals/dagstuhl-reports/BorosKPS19,bagan-enum-cdlin,
  berkholz-enum-fpt, carmeli-enum-ucqs, carmeli-enum-rand,
  carmeli-enum-func, segoufin-enum, deep-enum-alg, deep-enum-ranked},
see also the survey~\cite{berkholz-enum-tutorial}.  A landmark
result is that a CQ $q(\bar x)$ admits enumeration in \cdlin if
it is \emph{acyclic} and \emph{free-connex acyclic} where the former
means that $q$ has a join tree and the latter that the extension of
$q$ with an atom $R(\bar x)$ that `guards' the answer variables is
acyclic~\cite{bagan-enum-cdlin}. Partially matching lower bounds 
pertain to self-join free CQs~\cite{bagan-enum-cdlin,BraultBaron}.

The second aim of this paper is to introduce a novel notion of partial
answers to OMQs. In the traditional \emph{certain answers},
$\bar a \in Q(D)$ if and only if $\bar a$ is a tuple of constants from
$D$ such that $\bar a \in Q(I)$ for every model $I$ of $D$ and the
ontology \Omc used in $Q$. In contrast, a \emph{partial answer} may
contain, apart from constants from $D$, also the wildcard symbol
`$\ast$' to indicate a constant that we know must exists, but whose
identity is unknown. Such \emph{labeled nulls} may be
introduced by existential quantifiers in the ontology \Omc. To avoid
redundancy as in the partial answers $(a,\ast)$ and $(a,b)$, we are
interested in \emph{minimal} partial answers that cannot be `improved' by
replacing a wildcard with a constant from $D$ while still remaining a
partial answer. The following simple example illustrates that minimal partial
answers may provide useful information that is not provided by the
traditional answers, from now called \emph{complete
  answers}.
\begin{example}
  \label{ex:1}
  Consider the ontology \Omc that contains % the tuple-generating
  % dependencies
  %
  $$
  \begin{array}{rcl}
    \mn{Researcher}(x) &\rightarrow& \exists y \, \mn{HasOffice}(x,y)
    \\%[1mm]
    \mn{HasOffice}(x,y)&\rightarrow& \mn{Office}(y) \\%[1mm]
    \mn{Office}(x) &\rightarrow& \exists y \, \mn{InBuilding}(x,y), % \\[1mm]
       % \mn{InBuilding}(x,y)&\rightarrow& \mn{Building}(y),
    % \mn{Researcher}(x) &\rightarrow& \exists y \, \mn{HasOffice}(x,y) 
    %                                  \wedge \mn{Office}(y) \\[1mm]
    % \mn{Office}(x) &\rightarrow& \exists y \, \mn{InBuilding}(x,y) 
    %                                  \wedge \mn{Building}(y),
  \end{array}
  $$
  and the CQ
$     q(x_1,x_2,x_3) = %\mn{Researcher}(x_1) \wedge 
                        \mn{HasOffice}(x_1,x_2)  \wedge %\\[1mm]
                        \mn{InBuilding}(x_2,x_3)
                        $
  % $$
  % \begin{array}{rcl}
  %    q(x_1,x_2,x_3) &=& %\mn{Researcher}(x_1) \wedge 
  %                       \mn{HasOffice}(x_1,x_2)  \wedge %\\[1mm]
  %    \mn{InBuilding}(x_2,x_3) 
  % \end{array}
  % $$
  giving rise to the OMQ $Q(x_1,x_2,x_3)$. Take the following database~$D$:
  $$
  \begin{array}{@{}c@{}}
  \begin{array}{lll}
    \mn{Researcher}(\mn{mary}) &
    \mn{Researcher}(\mn{john}) &
    \mn{Researcher}(\mn{mike}) 
  \end{array} \\%[1mm]
  \begin{array}{ll}
    \mn{HasOffice}(\mn{mary},\mn{room1}) & 
    \mn{HasOffice}(\mn{john},\mn{room4}) 
  \end{array} \\%[1mm]
    \mn{InBuilding}(\mn{room1},\mn{main1})
  \end{array}
  $$
  The minimal partial answers to $Q$ on $D$ are 
  $$
  \begin{array}{ccc}
    (\mn{mary},\mn{room1},\mn{main1}) &
    (\mn{john},\mn{room4},\ast) &
    (\mn{mike},\ast,\ast).
  \end{array}
  $$
\end{example}
We also introduce and study \emph{minimal partial answers with multiple
  wildcards} $\ast_1,\ast_2,\dots$. Distinct occurences of the same
wildcard in an answer indicate the same null, while different
wildcards may or may not correspond to different nulls.  Multiple
wildcards may thus be viewed as adding equality on wildcards, but not
inequality. We note that there are certain similarities between minimal
partial answer to OMQs and answers to SPARQL queries with the
`optional' operator
\cite{DBLP:conf/pods/BarceloPS15,DBLP:conf/icdt/KrollPS16}, but also
many dissimilarities.

% It is particularly important for
% this paper that for CQs, enumeration in CD$\circ$Lin has been closely
% tied to certain notions of acyclicity.

The third aim of this paper is to study two problems for OMQs that are
closely related to constant delay enumeration: single-testing in
linear time (in data complexity) and \emph{all-testing} in
\cdlin or \dlc. Note that for Boolean queries, single-testing in linear
time coincides with enumeration in \cdlin and in \dlc.  An all-testing
algorithm has a prepocessing phase followed by a testing phase where
it repeatedly receives candidate answers $\bar a$ and returns `yes' or
'no' depending on whether
$\bar a \in Q(D)$~\cite{berkholz-enum-tutorial}. All-testing in
\dlc grants preprocessing time $O(||D||)$ while
the time spent per test must be independent of $D$, and all-testing
in \cdlin is defined accordingly.

An ontology-mediated query takes the form $Q(\bar x)=(\Omc,\Sbf,q)$
where \Omc is an ontology, \Sbf a schema for the databases on which
$Q$ is evaluated, and $q(\bar x)$ a conjunctive query.  In this paper,
we consider ontologies that are sets of guarded tuple-generating
dependencies (TGDs) or formulated in the description logic
$\mathcal{ELI}$. We remind the reader that a TGD takes the form
$\forall \bar x \forall \bar y \, \big(\phi(\bar x,\bar y) \rightarrow
\exists \bar z \, \psi(\bar x,\bar z)\big) $ where $\phi$ and $\psi$
are CQs, and that it is \emph{guarded} if $\phi$ has an atom that
mentions all variables from $\bar x$ and~$\bar y$. Up to
normalization, an $\mathcal{ELI}$-ontology may be viewed as a finite
set of guarded TGDs of a restricted form, using in particular only
unary and binary relation symbols. Both guarded TGDs and $\ELI$ are
natural and popular choices for the ontology language
\cite{cali-more-expressove-onto,cali-taming-chase,baader-introduction-to-dl}. We
use $(\class{G},\class{CQ})$ to denote the language of all OMQs that
use a set of guarded TGDs as the ontology and a CQ as the actual
query, and likewise for $(\class{ELI},\class{CQ})$ and
$\mathcal{ELI}$-ontologies.

We next summarize our results. In Section~\ref{sect:singletesting}, we
start with showing that in $(\class{G},\class{CQ})$, single-testing
complete answers is in linear time for OMQs that are weakly acyclic. A
CQ is % \emph{acyclic} if it has a join tree and it is
\emph{weakly acyclic} if it is acyclic after replacing the answer
variables with constants and an OMQ is weakly acyclic if the CQ in it
is; in what follows, we lift other properties of CQs to OMQs in the
same way without further notice.  Our proof relies on the construction
of a `query-directed' fragment of the chase and a reduction to the
generation of minimal models of propositional Horn formulas. We also
give a lower bound for OMQs from $(\class{ELI},\class{CQ})$ that are
self-join free: every such OMQ that admits single-testing in linear
time is weakly acyclic unless the triangle conjecture from
fine-grained complexity theory fails.  This generalizes a result for
the case of CQs without ontologies~\cite{BraultBaron}.  We observe
that it is not easily possible to replace $\class{ELI}$ by $\class{G}$
in our lower bound as this would allow us to remove also `self-join
free' while it is open whether this is possible even in the case
without ontologies. We also show that single-testing minimal partial
answers with a single wildcard is in linear time for OMQs from
$(\class{G},\class{CQ})$ that are acyclic and that the same is true
for multiple wildcards and acyclic OMQs from
$(\class{ELI},\class{CQ})$. We also observe that these (stronger) requirements
cannot easily be relaxed.

In Section~\ref{sect:enumallcomplete}, we turn to enumeration and
all-testing of complete answers. We first show that in
$(\class{G},\class{CQ})$, enumerating complete answers is in
\cdlin for OMQs that are acyclic and free\=/connex acyclic while
all-testing complete answers is in \cdlin for OMQs that are
free\=/connex acyclic (but not necessarily acyclic). % Here, a CQ
% $q(\bar x)$ is \emph{free-connex acyclic} if the extension of $q$ with
% an atom $R(\bar x)$ that `guards' the answer variables is acyclic.
The proof again uses the careful chase construction and a reduction to
the case without ontologies. The lower bound for single testing
conditional on the triangle conjecture can be adapted to enumeration,
with `not weakly acyclic' replaced by `not acyclic'. For enumeration,
it thus remains to consider OMQs that are acyclic, but not free-connex
acyclic.  % Again adapting a lower bound
% for CQs \cite{bagan-enum-cdlin,berkholz-enum-tutorial} to the case
% with ontologies, 
We show that for every self-join free OMQ from $(\class{ELI},\class{CQ})$
that is acyclic, connected, and admits enumeration in \cdlin,
the query is free-connex acyclic, unless sparse Boolean matrix multiplication
(BMM) is possible in time linear in the size of the input plus the
size of the ouput; this would imply a considerable advance in
algorithm theory and currently seems to be out of reach. We also show
that it is not possible to drop the requirement that the query is
connected, which is not present in the corresponding lower bound for the
case without ontologies 
\cite{bagan-enum-cdlin,berkholz-enum-tutorial}. We prove a similar
lower bound for all-testing complete answers, subject to a condition
regarding non-sparse BMM. All mentioned lower bounds also apply to
both kinds of partial answers.

In Section~\ref{sect:LPAsingleWildcardUpper}, we then prove that
enumerating minimal partial answers with a single wildcard is in
\dlc for OMQs from $(\class{G},\class{CQ})$ that are acyclic
and free\=/connex acyclic. This is one of the main results of this
paper, based on a non-trivial enumeration algorithm. Here, we only
highlight two of its features. First, the algorithm precomputes
certain data structures that describe `excursions' that a homomorphism
from $q$ into the chase of $D$ with $\Omc$ may make into the parts of
the chase that has been generated by the existential quantifiers in the
ontology. And second, it involves subtle sorting and pruning
techniques to ensure that only \emph{minimal} partial answers are
output. We also observe that all-testing minimal partial answers is less
well-behaved than enumeration as there is an OMQ
$Q \in(\class{ELI},\class{CQ})$ that is acyclic and free\=/connex
acyclic, but for which all-testing is not in \cdlin unless the
triangle conjecture fails.

Finally, Section~\ref{sect:enummulti} extends the upper bound from
Section~\ref{sect:LPAsingleWildcardUpper} to minimal partial answers
with multiple wildcards. We first show that all-testing (not
necessarily minimal!) partial answers with multiple wildcards is in
\dlc for OMQs that are acyclic and free\=/connex acyclic and
then reduce enumeration of minimal partial answers with multiple
wildcards to this, combined with the enumeration algorithm of minimal partial answers
with a single wildcard obtained in the previous section.

\ifbool{arxive}{
  Proof details are deferred to the appendix.
}
{
Most proof details are deferred to the long version {\color{red}cite
  arxive}.
}

% \bigskip 
% \bigskip 
% \bigskip 

% A landmark result is that a CQ
% $q(\bar x)$ admits enumeration in \cdlin if it is \emph{acyclic}
% and \emph{free-connex acyclic} where the former means that $q$ has a
% join tree and the latter that the extension of $q$ with an atom
% $R(\bar x)$ that `guards' the answer variables is acyclic
% \cite{bagan-enum-cdlin}. This is accompanied by two partial lower
% bounds.  The first one states that if a CQ that is self-join free, but
% not acyclic admits enumeration in CD$\circ$Lin, then (a generalization
% of) the triangle conjecture from fine-grained complexity theory fails
% \cite{BraultBaron}. The second lower bound says that if a CQ that is
% self-join free and acyclic, but not free-connex acyclic admits
% enumeration in CD$\circ$Lin, then this implies dramatic improvements
% in algorithm theory that seem to be for out of the reach of the
% current state of the art; more precisely, a sparse version of Boolean
% matrix multiplication would be possible in time linear in the size of
% the input plus the size of the ouput. \cite{bagan-enum-cdlin,
%   berkholz-enum-tutorial}.

\section{Preliminaries}
\label{sect:prelims}

\paragraph{\bf Relational Databases.}
% We consider a class $\Sbf_f$ of
% relation symbols $R$ with associated arity $\mn{ar}(R) \geq 0$ that contains
% countably infinitely many symbols of every arity. A {\em schema} \Sbf
% is a set $\Sbf \subseteq \Sbf_f$.
{%\color{blue}
Fix countably infinite and disjoint  
sets of constants \Cbf and \Nbf. We refer to the constants in \Nbf as
\emph{nulls}.}
A {\em schema} \Sbf
is a set of relation symbols~$R$ with associated arity \mbox{$\mn{ar}(R) \geq 0$}.
% We write $\mn{ar}(\Sbf)$ for
% $\max_{R \in \Sbf} \{\mn{ar}(R)\}$.
% and call $\Sbf_f$ the \emph{full
% schema}.
An {\em \Sbf-fact} is an expression of
the form $R(\bar c)$, where $R \in \Sbf$ and $\bar c$ is an
$\mn{ar}(R)$-tuple of constants from $\Cbf \cup \Nbf$.
An {\em $\Sbf$-instance} is a %(possibly infinite) 
set of \Sbf-facts
and an {\em \Sbf-database} is a finite $\Sbf$-instance
{%\color{blue}
that uses only constants from \Cbf.} We write
$\mn{adom}(I)$ for the set of constants used in instance $I$.
For a set $S \subseteq \Cbf \cup \Nbf$, $I_{|S}$ denotes the
restriction of $I$ to
facts that mention only constants from $S$.
%
%As usual, we write $\size{I}$ for the size of $I$.
%
A {\em homomorphism} from $I$ to an instance $J$ is a function $h :
\mn{adom}(I) \rightarrow \mn{adom}(J)$ such that $R(h(\bar c)) \in J$
for every $R(\bar c) \in I$.
% We write $I \rightarrow J$ for the fact that there is a homomorphism from $I$ to $J$.
A set $S \subseteq \mn{adom}(I)$ is a \emph{guarded set in $I$} if
there is a fact $R(\bar c ) \in I$ such that all constants from $S$
are in $\bar c$. The \emph{Gaifman graph} of a database $D$ is the
 undirected graph with vertices $\mn{adom}(D)$ and an edge $\{c_1,c_2\}$
whenever $c_1,c_2$ co-occur in a fact in $D$.

\paragraph{\bf Conjunctive Queries.}
A \emph{term} is a variable or a constant
{%\color{blue}
from~\Cbf}.
A {\em conjunctive query} (CQ) $q(\bar x)$ over a schema $\Sbf$ takes
the form $q(\bar x) \leftarrow \phi(\bar x, \bar y)$
%\end{equation}
where $\bar x$ and $\bar y$ are tuples of variables, $\phi$ is a
conjunction of \emph{relational atoms} $R_i(\bar t_i)$ with
$R_i \in \Sbf$ and $\bar t_i$ a tuple of terms of length
$\mn{ar}(R_i)$. %  and \emph{equality atoms} $t_1 = t_2$ with $t_1,t_2$ terms.
% We % refer to the variables in $\bar x$ as the {\em answer
%   % variables} of $q$ and
% require that only variables from $\bar x$ appear in
% equality atoms.
We refer to the variables in $\bar x$ as the \emph{answer variables}
of $q$ and to the variables in $\bar y$ as the \emph{quantified
  variables}.  With $\mn{var}(q)$, we denote the set of all variables
in $q$ and with $\mn{con}(q)$ the set of constants. Whenever
convenient, we identify a conjunction of atoms with a set of
atoms. % {\color{blue}When we are not interested in order and multiplicity, we treat
% tuples such as $\bar x$ and $\bar y$ as a set. }
The {\em arity} of $q$
is defined as the number of its answer variables and $q$ is Boolean if
it is of arity~0. When we do not want to make $\phi(\bar x, \bar
y)$ explicit, we may denote $q(\bar x) \leftarrow \phi(\bar x, \bar
y)$ simply with $q(\bar x)$.
%
%If $\bar x$ is empty, then $q$ is a \emph{Boolean CQ}.
%
%We write $\var{q}$ for the set of variables occurring in $q$. By abuse of notation, we may write $\alpha \in q$ to indicate that the atom $\alpha$ occurs in $q$.
%
%The evaluation of CQs is defined in terms of homomorphisms.
We say that $q(\bar x)$ is \emph{self\=/join free} if no relation symbol
occurs in more than one atom in it.
We write $\class{CQ}$ for the class of CQs. 

 Every CQ $q(\bar x)$ can be naturally seen as a database $D_q$, known as the
 {\em canonical database} of~$q$, obtained by viewing variables as
 constants
 {%\color{blue}
from \Cbf}.
 The \emph{Gaifman graph} of $q$ is that of $D_q$.
%We may simply write $q$ instead of $D[q]$.
%
  A \emph{homomorphism}
$h$  from $q$ to an instance $I$ is a homomorphism from $D_q$ to
 $I$ that is the identity on all constants that appear in $q$.
 A tuple $\bar c \in \mn{adom}(I)^{|\bar x|}$ is an {\em answer} to
 $q$ on $I$ if there is a homomorphism $h$ from $q$ to $I$ with
 $h(\bar x) = \bar c$.
The {\em evaluation of $q(\bar x)$ on $I$}, denoted $q(I)$, is the set of all answers to $q$ on~$I$.
%
% The only
% possible answer to a Boolean query is the empty tuple. For a Boolean
% CQ $q$, we may write $I \models q$ if $q(I) = \{()\}$ and
% $I \not\models q$ otherwise.

For a CQ $q$, but also for any other syntactic object $q$, we use
$||q||$ to denote the number of symbols needed to write $q$ as
a word over a suitable alphabet.

\paragraph{\bf Acyclic CQs.} Let $q(\bar x)
\leftarrow \phi(\bar x, \bar y)$ be a CQ. A \emph{join tree} for
$q(\bar x)$ is an undirected tree $T=(V,E)$ where $V$ is the set of
atoms in $\phi$ and for each variable $x \in \mn{var}(q)$, the set $\{
\alpha \in V \mid x \text{ occurs in } \alpha \}$ is a connected
subtree of $T$. % Note that for constant occurring in $q$, there is no
% connectedness condition.
Then $q(\bar x)$ is \emph{acyclic} if it has a join
tree. 
% SHORT
%  This is equivalent to $q(\bar x)$ being of generalized
% hypertree width~one \cite{berkholz-enum-tutorial}.
Note that constants
need not satisfy the connectedness condition imposed on variables. 
 We say that $q(\bar x)$ is \emph{weakly acyclic} if $q$
becomes acyclic after consistently replacing all answer variables with
fresh constants.
% SHORT
%\footnote{With `consistently', we mean that different
%  occurrences of the same answer variable must be replaced by the same
%  constant.}
%
A CQ $q(\bar x)$ is \emph{free\=/connex acyclic} if adding an atom
$R(\bar x)$ that `guards' the answer variables, where $R$ is a
relation symbol of arity $|\bar x|$, results in an acyclic CQ. Note
that other authors have called a CQ $q$ free\=/connex acyclic (or even
just free\=/connex) if $q$ is both acyclic and (in our sense)
free\=/connex acyclic
\cite{berkholz-enum-tutorial}. % We do not
% require the former.
%
Acyclicity and free-connex acyclicity are independent
properties, that is, neither of them implies the other. Each of them
implies weak acyclicity while the converse is
false. Figure~\ref{fig:examples_measures} shows (the Gaifman graphs
of) simple example CQs that illustrate the differences. Hollow nodes
indicate quantified variables, \mn{ac} stands for acyclic,
   	\mn{fc} for free\=/connex acyclic, and \mn{wac} for weakly acyclic.

\newcommand{\soutthick}[1]{%
    \renewcommand{\ULthickness}{1.5pt}%
       \sout{#1}%
    \renewcommand{\ULthickness}{.4pt}% Resetting to ulem default
}

\begin{figure}
    %\begin{wrapfigure}{r}{0.45\textwidth}
    \centering
    \begin{tikzpicture}[scale=.5]

    \node (top) at (0,2) {};
    \node (bottom) at (0,-.7) {};
    
    \node (a1) at (0,0) {};
    \node (a2) at (2,0) {};
    \node (a3) at (1, 1.73) {};
    \node (a4) at (1, -.5) {};
    
    \node (x0) at (0, 0) {};
    \node (y0) at (3, 0) {};
    \node (z0) at (6, 0) {};
    \node (t0) at (9, 0) {};
    \node (v0) at (12, 0) {};
    \node (of) at (-.5,0) {};

    \node[quantified] (x1) at ($(x0) + (a1)$) {};
    \node[quantified] (x2) at ($(x0) + (a2)$) {};
    \node[answer] (x3) at ($(x0) + (a3)$) {};
    \node         (x4) at ($(x0) + (a4)$) {\footnotesize \mn{ac}\ \mn{fc}\ \mn{wac} };

    \node[answer] (y1) at ($(y0) + (a1)$) {};
    \node[quantified] (y2) at ($(y0) + (a2)$) {};
    \node[answer] (y3) at ($(y0) + (a3)$) {}; 
    \node         (x4) at ($(y0) + (a4)$) {\footnotesize \mn{ac}\  \soutthick{\mn{fc}}\  \mn{wac} };
%    \node[circle, draw] (y5) at (4,1) {};

    \node[answer] (z1) at ($(z0) + (a1)$) {};
    \node[quantified] (z2) at ($(z0) + (a2)$) {};
    \node[answer] (z3) at ($(z0) + (a3)$) {}; 
    \node         (z4) at ($(z0) + (a4)$) {\footnotesize \soutthick{\mn{ac}}\ \soutthick{\mn{fc}}\ \mn{wac} };

    \node[answer] (t1) at ($(t0) + (a1)$) {};
    \node[answer] (t2) at ($(t0) + (a2)$) {};
    \node[answer] (t3) at ($(t0) + (a3)$) {}; 
    \node         (t4) at ($(t0) + (a4)$) {\footnotesize \soutthick{\mn{ac}}\ \mn{fc}\ \mn{wac} };
    
    \node[quantified] (v1) at ($(v0) + (a1)$) {};
    \node[quantified] (v2) at ($(v0) + (a2)$) {};
    \node[quantified] (v3) at ($(v0) + (a3)$) {}; 
    \node         (v4) at ($(v0) + (a4)$) {\footnotesize \soutthick{\mn{ac}}\ \soutthick{\mn{fc}}\ \soutthick{\mn{wac}} };
    
    \draw

    (x1) -> (x2) -> (x3)
    (y1) -> (y2) -> (y3)
    (z1) -> (z2) -> (z3) -> (z1)
    (t1) -> (t2) -> (t3) -> (t1)
    (v1) -> (v2) -> (v3) -> (v1);
    
    \draw[dashed]
    ($(y0) + (of) + (top)$) -- ($(y0) + (of) + (bottom)$)
    ($(z0) + (of) + (top)$) -- ($(z0) + (of) + (bottom)$)
    ($(t0) + (of) + (top)$) -- ($(t0) + (of) + (bottom)$)
    ($(v0) + (of) + (top)$) -- ($(v0) + (of) + (bottom)$);

    \end{tikzpicture}
   \caption{Different forms of acyclicity}
    \label{fig:examples_measures}
\end{figure}
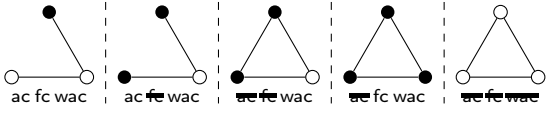
%
% Also note
% that even when the arity of relation symbols is at most binary such as
% in description logics, the fresh relation symbol $R$ may have any
% arity.
% There are several characterizations of when an acyclic CQ is 
% free\=/connex \cite{berkholz-enum-tutorial}. 
% % SHORT
% % Without giving
% % details, we mention that this is the case if and only if
% % it has a generalized hypertree decomposition of width one in which the
% % variables from $\bar x$ form a connected subtree. 
% A characterization
% that we use for the lower bounds in this paper is via bad paths.  A
% \emph{bad path} in a CQ $q$ is a sequence of variables
% $y_1, \dots, y_n$, $n \geq 3$, such that $y_1$ and $y_n$ are distinct
% answer variables, $y_2, \dots, y_{n-1}$ are quantified variables,
% and $\{y_i,y_{i+1}\}$ is an edge in the Gaifman graph of $q$ while
% $\{y_1,y_n\}$ is not. It was shown in
% \cite{bagan-enum-cdlin} that an acyclic CQ is free\=/connex if and
% only if it has no bad path, see also \cite{berkholz-enum-tutorial}. % We
% % will also consider bad paths in the more general context of weakly
% % acyclic CQs. Note that the notion of free\=/connexity makes no sense
% % for weakly acyclic CQs.

%\paragraph{TGDs, Guardedness.}  
\paragraph{\bf TGDs, Guardedness, Chase.}  
A {\em tuple-generating dependency} (TGD) $T$ over $\Sbf$ is a
first-order sentence $ \forall \bar x \forall \bar y \,
\big(\phi(\bar x,\bar y) \rightarrow \exists \bar z \, \psi(\bar
x,\bar z)\big) $ such that $q_\phi(\bar x) \leftarrow %\exists \bar y \,
\phi(\bar x,\bar y)$ and $q_\psi(\bar x) \leftarrow %\exists \bar z \,
\psi(\bar x,\bar z)$ are CQs that do not contain constants.  We call
$\phi$ and $\psi$ the {\em body} and {\em head} of $T$.  The body may
be the empty conjunction, i.e.~logical truth, denoted by \mn{true}.
The variables in $\bar x$ are the 
\emph{frontier variables}. 
For simplicity, we write $T$ as $\phi(\bar x,\bar y) \rightarrow
\exists \bar z \, \psi(\bar x,\bar z)$.  % We call $\phi$ and $\psi$ the
% {\em body} and {\em head} of $T$, denoted $\mn{body}(T)$ and
% $\mn{head}(T)$, respectively. 
An instance $I$ over $\Sbf$
\emph{satisfies}~$T$, denoted $I \models T$, if $q_\phi(I) \subseteq
q_\psi(I)$.  It {\em satisfies} a set of TGDs~$\Omc$, denoted $I
\models \Omc$, if $I \models T$ for each $T \in \Omc$. We then also
say that $I$ is a \emph{model} of $\Omc$.
%We work with finite sets of TGDs. 
% finite
% \emph{sets} of TGDs. {\color{blue} if we drop the closed world
%   stuff, we would introduce ontologies HERE!}
%
A TGD $T$ is {\em guarded} if its body is \mn{true} or contains a
\emph{guard atom} $\alpha$ that contains all variables in the
body~\cite{cali-taming-chase}. % Such an atom~$\alpha$ is the {\em
  % guard} of $T$, denoted $\mn{guard}(T)$.
%In case there are more than one atoms that can serve as the guard of $\sigma$, then we fix the left-most such atom in $\body{\sigma}$ as the guard. 
We write $\class{TGD}$ to denote the class of all TGDs
and $\class{G}$ for the class of guarded TGDs. % A~TGD $T$ is
% \emph{full} if the tuple $\bar z$ of variables is empty. We use
% $\class{FULL}$ to denote the class of full TGDs and shall often refer
% to $\class{G} \cap \class{FULL}$, the class of TGDs that are both
% guarded and full. Note that this class is essentially the class of
% Datalog programs with guarded rule bodies.

The well-known chase procedure makes explicit in an instance the
consequences of a set of TGDs \cite{MaMS79,JoKl84, FKMP05,
  cali-taming-chase}.
% We use the chase throughout this paper, a
% detailed definition can be found in the appendix. Here, we only mention 
% that our chase is oblivious, that is, a TGD is eventually applied
% whenever its body is satisfied, even if also its head is already
% satisfied. As a consequence, every fair chase sequence for an instance
% $I$ with a set of TGDs $\Omc$ produces the same instance, up to
% isomorphism.  We denote this instance with $\mn{ch}_\Omc(I)$.
Let $I$ be an instance and
$\Omc$ be a set of TGDs.  A TGD
\mbox{$T= \phi(\bar x,\bar y) \rightarrow \exists \bar z \, \psi(\bar x,\bar
z) \in \Omc$} is \emph{applicable} to a tuple $(\bar c,\bar c')$ of
constants in $I$ if $\phi(\bar c,\bar c') \subseteq I$. In this case,
the {\em result of applying $T$ in $I$ at $(\bar c,\bar c')$} is the
instance $I \cup \{\psi(\bar c,\bar c'')\}$ where $\bar c''$ is the
tuple obtained from $\bar z$ by simultaneously replacing each variable
$z$ with a fresh distinct null that does not occur in~$I$. We 
refer to such an application as a \emph{chase step}. %  and to the fresh
% constants as \emph{nulls}.  % We
% describe such a single chase step by writing
% $I \xrightarrow{T, \,(\bar c,\bar c')} J$.
A {\em chase sequence for $I$ with $\Omc$} is a sequence of instances
$I_0,I_1,\dots$ such that $I_0=I$ and each $I_{i+1}$ is the result of
applying some TGD from \Omc at some tuple $(\bar c, \bar c')$ of
constants in $I_i$. The \emph{result} of this chase sequence is the
instance $J = \bigcup_{i \geq 0} I_i$. The chase sequence is
\emph{fair} if whenever a TGD $T \in \Omc$ is applicable to a tuple
$(\bar c,\bar c')$ in some $I_i$, then this application is a chase
step in the sequence. Fair chase sequences are oblivious in that a TGD
is eventually applied whenever its body is satisfied, even if also its
head is already satisfied. As a consequence, every fair chase sequence
for $I$ with $\Omc$ leads to the same result, up to isomorphism.
We denote this result with
$\mn{ch}_\Omc(I)$.

\paragraph{\bf Ontology-Mediated Query, Description Logic.}
An \emph{ontology} is a finite set of TGDs.  An
\emph{ontology-mediated query (OMQ)} takes the form $Q= (\Omc,\Sbf,q)$
where $\Omc$ is an ontology, $\Sbf$ is a finite schema called the
\emph{data schema}, and $q$ is a CQ.  Both \Omc and $q$ can use
symbols from~\Sbf, but also additional symbols, and in particular \Omc
 can `introduce' symbols to enrich the vocabulary available
 for querying.
We assume w.l.o.g. that \Sbf contains only relation
symbols that occur in $\Omc$ or~$q$.
% We assume w.l.o.g.\ that all
% relation symbols in $q$ that are not from \Sbf occur also in \Omc.
% This can always be achieved by introducing dummy TGDs $R(\bar x)
% \rightarrow R(\bar x)$.
%
% When \Omc and $q$ only use symbols from \Sbf, then we say that the
 % data schema of $Q$ is \emph{full}. 
The {\em arity} of $Q$ is defined
as the arity of $q$. We write $Q(\bar x)$ to emphasize that the answer
variables of $q$ are $\bar x$
% and for brevity often refer to the
% data schema simply as the schema.
% In fact, use of ontologies
% is to enrich the schema that is available for formulating queries.
and say that $Q$ is \emph{acyclic} if $q$ is and likewise for
\emph{weakly acyclic},
\emph{free\=/connex acyclic}, \emph{self\=/join free}, and so on.

A tuple $\bar c \in \mn{adom}(D)^{|\bar x|}$ is a {\em (certain) answer} to $Q$
on $D$ if $\bar c \in q(I)$ for every model $I$ of \Omc with
$I \supseteq D$.
The {\em evaluation of $Q(\bar x)$} over $D$, denoted $Q(D)$, is the
set of all answers to $Q$ over~$D$.  Importantly, $Q(D) =
q(\mn{ch}_\Omc(D))$ for every OMQ $Q = (\Omc,\Sbf,q)$ and
\Sbf-database $D$. When convenient, we may write $D \cup
  \Omc\models q(\bar c)$ in place of $\bar c \in Q(D)$. We say that
$Q$ is \emph{empty} if $Q(D)=\emptyset$ for all \Sbf-databases $D$.

Let us remark that a CQ $q$ can be semantically acyclic in the sense
that it is equivalent to an acyclic CQ, but not acyclic itself
\cite{DBLP:conf/cp/DalmauKV02,barcelo-semantic-optimization}.
It is
known that this is the case if and only if the homomorphism core of $q$ is
acyclic. An OMQ can be semantically acyclic (in the same sense) even
if the homomorphism core of the CQ in it is not acyclic, that is, the
ontology has an impact on semantic acyclicity; see
\cite{DBLP:conf/lics/BarceloFLP19,barcelo_omq_limits-g} for very
similar effects that pertain to bounded treewidth. Since we are
concerned with data complexity in this article, we can simply
replace an OMQ with any equivalent one and thus w.l.o.g.\ refrain
from considering semantic acyclicity.

We next introduce the widely known description logic
\ELI~\cite{baader-introduction-to-dl}. Traditionally, description
logics come with their own variable-free syntax.  Here, we introduce
\ELI using TGD syntax. A guarded TGD
$\phi(\bar x,\bar y) \rightarrow \exists \bar z \, \psi(\bar x,\bar
z)$ is an \emph{\ELI TGD} if it uses only unary and binary relation
symbols, has only a single frontier variable, contains no reflexive
loops and multi-edges in body or head, and has a head that is acyclic
{%\color{blue}
and connected}.
% SHORT
%is a tree without reflexive loops and multi-edges.
% , that is, the
% undirected graph $(\mn{var}(q),\{\{x,y\} \mid R(x,y) \in \psi\})$ is a
% tree and $\psi$ contains no two binary atoms with the same set of
% variables. 
Note that the original definition of \ELI is more liberal in that it
restricts the body in the same way as the head in our definition, thus
encompassing also unguarded TGDs. However, the restricted form used
here can be attained by syntactic
normalization~\cite{baader-introduction-to-dl}.
{%\color{blue}
  Since the normalization of an ontology inside an OMQ does not affect query
  answers, all results in this paper apply also to the more liberal
  definition of \ELI.
}
% When applied
% to an ontology inside an OMQ, then the normalization preserves all
% (non-)answers as long as the fresh relation symbols that it introduces
% are not included in the data schema. Our more restricted form of \ELI
% TGDs is thus without loss of generality. An extension of \ELI is
% \ELIH. An \emph{\ELIH TGD} is a TGD that is an \ELI TGD or of the form
% $R(x,y) \rightarrow S(x,y)$. 
We use $\class{ELI}$ to denote the set
of all \ELI TGDs. %, and likewise for $\class{ELHI}$.

An \emph{OMQ language} is a class of OMQs. For a class of TGDs
$\class{C}$ and a class of CQs $\class{Q}$, we write
$(\class{C},\class{Q})$ to denote the OMQ language that consists of
all OMQs $(\Omc,\Sbf,q)$ where \Omc is a set of TGDs from $\class{C}$
and $q \in \class{Q}$. For example, we may write
$(\class{G},\class{CQ})$ and $(\class{ELI},\class{CQ})$.

Let $Q_i(\bar x)=(\Omc_i,\Sbf,q_i)$ for $i \in \{1,2\}$.  Then
OMQ $Q_1$ is \emph{contained} in OMQ $Q_2$, written $Q_1 \subseteq Q_2$,
if $Q_1(D) \subseteq Q_2(D)$ for every \Sbf-database $D$. Moreover,
$Q_1$ and $Q_2$ are \emph{equivalent}, written $Q_1 \equiv Q_2$,
if $Q_1 \subseteq Q_2$ and $Q_2 \subseteq Q_1$.

\medskip 

%\paragraph{\bf Machine Model.} 
{\bf Machine Model.} 
As our computational model, we use RAMs
under the uniform cost model \cite{DBLP:journals/jcss/CookR73}, see
\cite{Grandjean-RAM} for a formalization. Such a RAM has a
one-way read-only input tape, a write-only output tape, and
an unbounded number of registers that store non-negative integers of
$O(\log n)$ bits, $n$ the input size; this is called a DRAM %DLINRAM
in \cite{Grandjean-RAM}, used there to define the complexity class
DLINEAR. Adding, subtracting, and comparing the values
of two registers as well as bit shift takes time~$O(1)$. On a DRAM,
sorting is possible in linear time and we can use and access lookup
tables indexed by constants from $\mn{adom}(D)$ or by tuples of
constants of length $O(1)$ \cite{Grandjean-RAM}.  This model is 
standard in the context of constant delay enumeration, see also
\cite{segoufin-enum, bagan-enum-cdlin,carmeli-enum-func,
  berkholz-enum-tutorial} and the
\ifbool{arxive}{
  appendix}
{
  long version
}
for more details.

\medskip 

%\paragraph{\bf Modes of Query Evaluation.}
{\bf Modes of Query Evaluation.}
% The main modes of query
% evaluation studied in this paper are the enumeration of all answers
% and all-testing.
% , that is, after a preprocessing phase we
% repeatedly receive tuples of elements as an answer candidate and then
% have to check whether the candidate is indeed an answer. We refer to
% this as \emph{all-testing}. In contrast,
\emph{Single-testing} means to decide, given an OMQ $Q(\bar
x)=(\Omc,\Sbf,q)$, an \Sbf-database $D$, and an answer candidate $\bar
c \in \mn{adom}(D)^{|\bar x|}$, whether $\bar c \in Q(D)$. We
generally consider data
complexity, where the OMQ $Q$ is fixed and thus of constant size and the
only remaining inputs are $D$ and $\bar c$. % Unless otherwise
% specified,
%When giving complexity results for single-testing we
%generally mean data complexity.

An \emph{enumeration algorithm} for a class of OMQs $\class{C}$ is
given as inputs an OMQ $Q(\bar x) =(\Omc,\Sbf,q)\in \class{C}$ and an \Sbf-database $D$.
% It works in two phases.
In the \emph{preprocessing phase}, it may produce data
structures, but no output. In the subsequent \emph{enumeration phase},
it enumerates all tuples from $Q(D)$, without repetition, followed by
an end of enumeration signal.  An \emph{all-testing algorithm} for
$\class{C}$ is defined similarly.  It takes the same two inputs, and
has the same preprocessing phase, followed by a \emph{testing phase}
where it repeatedly receives tuples
$\bar a \in \mn{adom}(D)^{|\bar x|}$ and returns `yes' or 'no'
depending on whether $\bar a \in Q(D)$.

Let $(\class{L},\class{Q})$ be an OMQ language.  We say that
\emph{answer enumeration for $(\class{L},\class{Q})$ is possible with
  linear preprocessing and constant delay}, or \emph{in \dlc} for
short, if there is an enumeration algorithm for
$(\class{L},\class{Q})$ in which preprocessing takes time
$f(||Q||) \cdot \Omc(||D||)$, $f$ a computable function, while the
delay between the output of two consecutive answers depends only on
$||Q||$, but not on $||D||$. Enumeration in \cdlin is defined
likewise, except that in the enumeration phase, the algorithm may
consume only a constant amount of memory. Accessing the data
structures computed in the preprocessing phase does not count
as memory usage. It is not clear
whether \dlc and \cdlin coincide or not, see e.g.\ \cite{kazana-phd}.
% Note that this corresponds to data
% complexity: if $||Q||$ is bounded by a constant, then the time required
%   for preprocessing is linear and the delay is
%   $O(1)$.
  % Also note that the time requirement for
% the preprocessing phase can be described as 
% fixed-parameter tractability with linear runtime, whence the term FPL.
The definition of \emph{\dlc} and \emph{\cdlin} for all-testing is analogous,
except that the enumeration delay is replaced with the time needed for
testing. 

% Sometimes, CD$\circ$Lin is defined to additionally
% require only a constant number of memory writes during
% the enumeration phase \cite{kazana-phd}. We do not require this condition.
% (and it is indeed not satisfied by our algorithms enumerating least
% partial answers).

%\paragraph{\bf Partial Answers.}
{\bf Partial Answers.}
% The existential quantifiers in the
% heads of TGDs used in ontologies introduce constants that we know must
% exist, but whose exact identity is unknown; they may be viewed as
% \emph{(labeled) nulls} \cite{DBLP:journals/jacm/ImielinskiL84}. In the
% standard notion of certain answers to OMQs as defined above, such null
% constants are never part of an answer. We introduce a notion of
% partial answers in which nulls may be returned in the form of a
% wildcard symbol.
%
% We are then interested in least partial
% answers, meaning that we never want to return a wildcard in a position
% where also a non-null constant could appear.
We first introduce partial answers with a single wildcard symbol
`$\ast$' (that is not
{%\color{blue}
in $\Cbf \cup \Nbf$}).
A~\emph{wildcard tuple} for an instance $I$ takes the form
$(c_1,\dots,c_n) \in (\mn{adom}(I) \cup \{\ast\})^n$, $n \geq 0$.  For
wildcard tuples $\bar c = (c_{1},\dots,c_{n})$ and $\bar c' =
(c'_{1},\dots,c'_{n})$, we write $\bar c \preceq \bar c'$ if $c'_{i}
\in \{ c_i, \ast \}$ for $1 \leq i \leq n$.  Moreover, $\bar c \prec
\bar c'$ if $\bar c \preceq \bar c'$ and $\bar c \neq \bar c'$. For
example, $(a,b) \prec (a,\ast)$ and $(a,\ast) \prec
(\ast,\ast)$. Informally, $\bar c \prec \bar c'$ expresses that tuple
$\bar c$ is preferred over tuple $\bar c'$ as it carries more
information.  A \emph{partial answer} to OMQ
$Q(\bar x)=(\Omc,\Sbf,q)$ on \Sbf-database $D$ is a wildcard tuple
$\bar c$ for $D$ of length $|\bar x|$ such that for each model $I$ of
$\Omc$ with $I \supseteq D$, there is a $\bar c' \in q(I)$ such that
% {\color{blue}$\bar c$ can be obtained from $\bar c'$ by replacing every constant
% from $\mn{adom}(I) \setminus \mn{adom}(D)$ with `$\ast$'.}
$\bar c' \preceq \bar c$. Note that some positions in $\bar c'$ may
 contain constants from $\mn{adom}(I) \setminus \mn{adom}(D)$, and that
 the corresponding position in $\bar c$ must then have a wildcard.
%{\color{red} Are we sure $c$ and $c'$ do not have to agree on domain
%  of $D$?}
%since $\bar c$ is a wildcard tuple for $D$.
%
A
partial answer $\bar c$ to $Q$ on \Sbf-database $D$ is a
\emph{minimal partial answer} if there is no partial answer
$\bar c'$ to $Q$ on $D$ with $\bar c' \prec \bar c$.
The {\em partial evaluation of $Q(\bar x)$} on~$D$, denoted
$Q(D)^{\ast}$, is the set of all minimal partial answers to $Q$ on~$D$.
Note that $Q(D) \subseteq Q(D)^{\ast}$. An illustrating example
is provided
in Section~\ref{sect:intro}.

Minimal partial answers may provide valuable information not captured by
complete answers. However, one might argue that
complete answers are more important than minimal partial answers that
contain a wildcard, and should thus be output first by an enumeration
algorithm. We observe that this is always
possible if we are interested in \dlc (whereas it is not clear whether
an analogous statement for \cdlin holds).
\begin{restatable}{proposition}{propcompletefirst}
\label{prop:completefirst}
  Let $Q \in (\class{TGD},\class{CQ})$. If minimal partial answers to
  $Q$ can be enumerated in \dlc and the same is true for
  complete answers, then there is a \dlc enumeration algorithm
  for minimal partial answers to $Q$ that produces the complete
  answers first.
\end{restatable}
%
%\cMP{just to note: not CDLin as we store in memory.}
We next introduce partial answers with multiple wildcards. Fix a
countably infinite set of \emph{wildcards}
$\Wmc = \{ \ast_1,\ast_2,\dots \}$ (that are not
{%\color{blue}
in $\Cbf \cup \Nbf$}). A \emph{multi-wildcard tuple} for an instance
$I$ is a tuple $(c_1,\dots,c_n) \in (\mn{adom}(I) \cup \Wmc)^n$,
$n \geq 0$, such that if $c_i = \ast_j$ with $j>1$, then there is an
$i' < i$ with $c_{i'}=\ast_{j-1}$. Examples for multi-wildcard tuples
are $(\ast_1,\ast_2)$ and $(a,\ast_1,b,a,\ast_2,\ast_1,\ast_2)$ and a
non-example is $(\ast_2,\ast_1)$. Occurrences of the same wildcard
represent occurrences of the same null while different wildcards
represent nulls that may or may not be
different. % In this sense, multi-wildcard extend
% wildcard-tuples with equality, but not with inequality. We disallow
% tuples such as $(\ast_2,\ast_1)$ to avoid redundancy, as they
% represent the same situation as the (valid) tuple $(\ast_1,\ast_2)$.
%
For multi-wildcard tuples $\bar c = (c_{1},\dots,c_{n})$ and
$\bar c' = (c'_{1},\dots,c'_{n})$, we write $\bar c \preceq \bar c'$ if 
\begin{enumerate}

\item $c_{i}=c'_{i}$ or $\Wmc \not\ni c_i \neq c'_{i} \in  \Wmc$ for $1 \leq i \leq n$ and

\item %$c_{i},c_{j} \in \Wmc$ and
  $c'_{i} = c'_{j}$
  implies $c_{i} = c_{j}$ for \mbox{$1 \leq i,j \leq n$}.

% \item one of the following is true:
%   %
%   \begin{itemize}
%
%   \item $c_{i,1} \notin \Wmc$ and $c_{i,2} \in \Wmc$ for some $i$;
%
%   \item $c_{i,1}=c_{j,1} \in \Wmc$ and $c_{i,1} \neq c_{j,2}$ for some $i,j$.
%
%   \end{itemize}
%   %
\end{enumerate}
Moreover, $\bar c \prec \bar c'$ if $\bar c \preceq \bar c'$ and
$\bar c \neq \bar c'$. For example, $(\ast_1,a) \prec (\ast_1,\ast_2)$
and $(a,\ast_1,\ast_2,\ast_1) \prec
(a,\ast_1,\ast_2,\ast_3)$. \emph{Partial answers with multi-wildcards}
and \emph{minimal partial answers with multi-wildcards} are defined
in exact analogy with (minimal) partial answers, but using
multi-wildcard tuples in place of wildcard tuples. The {\em partial
  evaluation of $Q(\bar x)$ with multi-wildcards} on $D$, denoted
$Q(D)^{\Wmc}$, is the set of all minimal partial answers with
multi-wildcards to $Q$ on~$D$. % Informally, we may view partial answers
% with multi-wildcards as partial answers with a single wildcard, but
% extended with `wildcard equalities'. Inequalities, in contrast, cannot
% be represented by multi-wildcards.
%
\begin{example}
  Reconsider the ontology OMQ $Q=(\Omc,\Sbf,q)$ and database $D$ from
  Example~\ref{ex:1}.  Then $Q(D)^{\Wmc}$
  contains the tuples
  $$
  \begin{array}{ccc}
    (\mn{mary},\mn{room1},\mn{main1}) &
    (\mn{john},\mn{room4},\ast_1) &
    (\mn{mike},\ast_1,\ast_2).
  \end{array}
  $$
  Let the ontology $\Omc'$ be obtained from \Omc by adding
  $$
  \mn{Prof}(x) \wedge \mn{HasOffice}(x,y) \quad \rightarrow \quad \mn{LargeOffice}(y)
  $$
  and $\Sbf'$ from \Sbf by adding \mn{LargeOffice}, consider the CQ
  $$
    \begin{array}{rcl}
    q'(x_1,x_2,x_3,x_4)  &=&
  \mn{HasOffice}(x_1,x_2) \wedge \mn{LargeOffice}(x_2)\,\wedge\\[1mm]
  && \mn{HasOffice}(x_1,x_3) \wedge\mn{InBuilding}(x_3,x_4),
    \end{array}
  $$
  and let $Q'=(\Omc',\Sbf',q')$. Moreover, let $D'$ be $D$ extended with fact
  $$
  \mn{Prof}(\mn{mike}).
  $$
  Then $Q'(D')^{\Wmc}$ contains, among others, the tuple
  $(\mn{mike},\ast_1,\ast_1,\ast_2)$, but not the tuple
  $(\mn{mike},\ast_1,\ast_2,\ast_3)$ which is a partial answer, but not a minimal one.
  
  Finally, let the ontology $\Omc''$ be obtained from \Omc by adding
  $$
  \mn{OfficeMate}(x,y) \quad \rightarrow \quad \exists z \, \mn{HasOffice}(x,z)
  \wedge \mn{HasOffice}(y,z)
  $$
  and $\Sbf''$ from \Sbf by adding $\mn{OfficeMate}$,
  consider the CQ
  $$
  \begin{array}{rcl}
    q''(x_1,x_2,x_3,x_4)  &=& \exists y \, \mn{Hasoffice}(x_1,x_3)
                             \wedge \mn{Hasoffice}(x_2,x_4) \, \wedge
    \\[1mm]
    && \phantom{\exists y \,} \mn{InBuilding}(x_3,y) \wedge \mn{InBuilding}(x_4,y),
  \end{array}
  $$
  and set $Q''(x_1,x_2,x_3,x_4)=(\Omc'',\Sbf'',q'')$.
  Moreover, let $D''$ be $D$ extended with fact
  $$
  \mn{OfficeMate}(\mn{mary},\mn{mike}).
  $$
  $Q''(D'')^{\Wmc}$ contains, among others, the tuple
  $
     (\mn{mary},\mn{mike},\ast_1,\ast_1).
  $
\end{example}
%
% One may also consider a variation of partial answers and minimal partial
% answers where only a single wildcard $\ast$ can be used and where
% repeated occurrences of $\ast$ do not indicate multiple occurrences of
% the same null. For a concrete definition, please see the extended
% abstract \cite{DLpaper} of an unpublished earlier version of this
% paper. Such minimal partial answers are strictly less informative than
% the ones considered in this paper. We stress that all results in this
% paper, lower bounds as well as upper bounds, remain valid in this
% variation.

% Single-testing, enumeration, and all-testing of partial answers is
% defined in the expected way.
% To the best of our knowledge, partial
% answers to OMQs have not been considered before in the literature.
% This might be due to the fact that enumeration of answers to OMQs
% has not been considered either, but instead almost all studies
% concentrate on single-testing. For single-testing, however, partial
% answers are not interesting as when given a concrete candidate
% answer we can always resort to quantifying away the variables
% for which a wildcard is prescribed in the answer.

{%\color{blue}
It should not be surprising that minimal partial answers can
equivalently be defined in terms of the chase.
Let $q(\bar x)$ be a CQ and $I$ an instance, possibly containing
nulls. For an answer $\bar a \in q(I)$, we use $\bar a^\ast_\Nbf$ to
denote the (unique) wildcard tuple for $I$ obtained from $\bar a$ by
replacing all nulls with `$\ast$'. We call such an $\bar a^\ast_\Nbf$
a \emph{partial answer} to $q$ on $I$ and say that it is a
\emph{minimal} partial answer if there is no $\bar b \in q(I)$ with
$\bar b^\ast_\Nbf \prec \bar a^\ast_\Nbf$. We use $q(I)^{\ast}_\Nbf$
to denote the set of minimal partial answers to $q$ on $I$. Similarly,
we use $\bar a^{\Wmc}_\Nbf$ to denote the (unique) multi-wildcard tuple
for $I$ obtained by consistently replacing all nulls with
wildcards from $\Wmc = \{ \ast_1, \ast_2,\dots
\}$. % such that positions that
% have the same variable in $\bar x$ carry the same wildcard.
We then
define \emph{minimal partial answer with multi-wildcards} to $q$ on
$I$, denoted $q(D)^{\Wmc}_\Nbf$, in the expected way.
\begin{restatable}{lemma}{lemtwosemantics}
\label{lem:twosemantics}
Let $Q(\bar x)=(\Omc,\Sbf,q) \in (\class{TGD},\class{CQ})$ and $D$ be an
\Sbf-database. Then $Q(D)^{\ast}=q(\mn{ch}_\Omc(D))^{\ast}_\Nbf$ and
$Q(D)^{\Wmc}=q(\mn{ch}_\Omc(D))^{\Wmc}_\Nbf$.
\end{restatable}
We remark that there is a subtlety here. In contrast to
Lemma~\ref{lem:twosemantics}, the (not necessarily minimal) partial
answers to $Q(\bar x)=(\Omc,\Sbf,q) $ on an \Sbf-database $D$ need not
coincide with the partial answers to $q$ on $\mn{ch}_\Omc(D)$. In
fact, $(\ast,\cdots,\ast)$ is a partial answer to
$Q(\bar x)=(\Omc,\Sbf,q) $ on $D$ if there is a partial answer at all,
but this is not the case for the partial answers to $q$ on
$\mn{ch}_\Omc(D)$, e.g.\ when $Q(x)=(\emptyset,\{A\},A(x))$ and
$D=\{ A(c) \}$.  }

\section{Single-Testing}
\label{sect:singletesting}

We consider the limits of single-testing in linear time for the OMQ
languages $(\class{G},\class{CQ})$ and $(\class{ELI},\class{CQ})$.
{\color{black}For complete answers, we establish a close link to weak
  acyclicity while minimal partial answers with a single wildcard are
  linked (in a more loose way) to acyclicity. The latter is also achieved
  for minimal partial answers with multi-wildcards, but only when the
  ontology is from $\class{ELI}$.
  {%\color{blue}
    To the best of our
    knowledge, these are the first results on linear time
    single-testing for ontology-mediated queries. Existing
    algorithms from the literature
    seem to require at least quadratic time (although authors
    typically
    do not analyse the degree of the polynomial explicitly).}
\begin{restatable}{theorem}{thmsinglelin}
  \label{thm:singlelin}
  Single-testing is in linear time for
  \begin{enumerate}

  \item weakly acyclic OMQs from $(\class{G},\class{CQ})$ in the case
    of complete answers;

  \item acyclic OMQs from $(\class{G},\class{CQ})$ in the case of
    minimal partial answers with single wildcards;

  \item acyclic OMQs from $(\class{ELI},\class{CQ})$ in the case of
    minimal partial answers with multi-wildcards.

  \end{enumerate}
\end{restatable}
}
To prove
Theorem~\ref{thm:singlelin}, we first show that for every OMQ
$Q(\bar x)=(\Omc,\Sbf,q) \in (\class{G},\class{CQ})$ and \Sbf-database
$D$, one can compute in time linear in $||D||$ a (finite!) database
$\mn{ch}^q_\Omc(D)$ that enjoys all properties of the chase
$\mn{ch}_\Omc(D)$ which are important for enumerating answers
to $Q$, both complete and partial. Informally, $\mn{ch}^q_\Omc(D)$
contains only those parts of $\mn{ch}_\Omc(D)$ that are `relevant to $q$'.
%
% In particular,
% $Q(D)=q(\mn{ch}^q_\Omc(D)) \cap \mn{adom}(D)^{|\bar x|}$ and also
% $Q(D)^{\ast}$ and $Q(D)^{\Wmc}$ can be obtained from
% $\mn{ch}^q_\Omc(D)$ in a straightforward way. % Essentially,
% $\mn{ch}^q_\Omc(D)$ is constructed from $D$ by considering a certain
% set of relevant CQs $p(\bar y)$ and all tuples
% $\bar c \in \mn{adom}(D)^{|\bar y|}$ such that
% $D \cup \Omc \models p(\bar c)$ and the constants in $\bar c$
% form a guarded set in $D$, and then adding a copy of $D_p$ that
% uses the constants in $\bar c$ in place of the answer variables
% $\bar y$ of $p$ and only fresh constants otherwise. 
We refer to
$\mn{ch}^q_\Omc(D)$ as the \emph{query-directed chase}, similar
constructions have been used e.g.\ in
\cite{DBLP:conf/ijcai/BienvenuOSX13,barcelo_omq_limits-g}.

{%\color{blue}
%We now define $\mn{ch}^q_\Omc(D)$ in  detail.
Let $\mn{cl}(Q)$  denote the set of CQs that are connected
and use only relation symbols that occur in \Omc, no constants, and
only variables from a fixed set $V$ whose cardinality is the maximum
of $|\mn{var}(q)|$ and the arities of relation symbols in \Omc. Note
that the CQs in $\mn{cl}(Q)$ may have any arity, including zero,
and that the number of CQs in $\mn{cl}(Q)$ is independent of $D$. The
database $\mn{ch}^q_\Omc(D)$ is obtained from $D$ by adding, for every CQ
$p(\bar y) \in \mn{cl}(Q)$ and every
$\bar c \in \mn{adom}(D)^{|\bar y|}$ such that
$D \cup \Omc \models p(\bar c)$ and the constants in $\bar c$
constitute a guarded set in $D$, a copy of $D_p$ that uses the
constants in $\bar c$ in place of the answer variables $\bar y$ of $p$
and only fresh constants otherwise.
\begin{restatable}{lemma}{proprestrictedchaseworks}
  \label{prop:restrictedchaseworks}
  Let $Q(\bar x)=(\Omc,\Sbf,q) \in (\class{G},\class{CQ})$ and $D$
  be an {\Sbf}\=/database. Then $Q(D)=q(\mn{ch}^q_\Omc(D)) \cap
  \mn{adom}(D)^{|\bar x|}$, $Q(D)^{\ast}=
  q(\mn{ch}^q_\Omc(D))^{\ast}_\Nbf$, and $Q(D)^{\Wmc}=
  q(\mn{ch}^q_\Omc(D))^{\Wmc}_\Nbf$.
\end{restatable}
}

As announced, the query-directed chase can be computed in linear
time. $Q$ is not required to be acyclic for this to hold.
\begin{restatable}{proposition}{propchaseinlineartime}
\label{prop:chaseinlineartime}
  Let $Q(\bar x)=(\Omc,\Sbf,q) \in (\class{G},\class{CQ})$ and let $D$
  be an \Sbf-database. Then $\mn{ch}^q_\Omc(D)$ can be computed
  in time linear in $||D||$, more precisely in time $2^{2^{O(||Q||^2)}}{\cdot} ||D||$.
\end{restatable}
To prove Proposition~\ref{prop:chaseinlineartime},
% construct $\mn{ch}^q_\Omc(D)$ in linear time (actually in time
% $2^{2^{O(||Q||^2)}}\cdot ||D||$), 
we derive from $D$ and $Q$ a
satisfiable Horn formula $\theta$, make use of the fact that a minimal
model of $\theta$ can be computed in linear
time~\cite{dowling-gallier-horn}, and then read off
$\mn{ch}^q_\Omc(D)$ from the minimal model.
{%\color{blue}
We are not aware that such
an approach has been used before.}

For Point~(1) of Theorem~\ref{thm:singlelin}, we have to check whether
$\bar c \in Q(D)$ which can now be done in linear time a straightforward way. First compute
$\mn{ch}^q_\Omc(D)$.
% introduce a fresh unary
% relation symbol $P_{\mn{db}}$ to identify in $\mn{ch}^q_\Omc(D)$ the
% constants from $D$, to avoid that nulls are returned as answers. We
% also
Then replace the answer variables in $q$ by the constants
from~$\bar c$, turning the weakly acyclic $q$ into an acyclic CQ.
Finally, use an existing procedure such as Yannakakis' algorithm to
single-test the resulting CQ in linear
time~\cite{yannakakis-algotrithm}.  Points~(2) and~(3) of
Theorem~\ref{thm:singlelin} are proved by a (Turing) reduction to the
case of complete
answers. Details are provided in the
\ifbool{arxive}{
  appendix. }
{
  long version.
}
% , considering multiple OMQs that can be obtained from $Q$ by
% quantifying certain answer variables and (in the case of multiple
% wildcards) possibly identifying some of the freshly quantified
% variables.

We next prove a lower bound that partially matches
Theorem~\ref{thm:singlelin}. As in the case without ontologies, we do
not obtain a full dichotomy as the lower bound only applies to queries
that are self-join free. In addition (and related to this), it only
applies to OMQs where the ontology is formulated in the subclass
$\class{ELI}$ of $\class{G}$.  The lower bound is conditional on the
triangle conjecture, which we formulate next.  \emph{Triangle
  detection} is the problem to decide, given an undirected graph
$G=(V,E)$ as a list of edges, whether $G$ contains a 3-clique. The
triangle conjecture from fine-grained complexity theory
\cite{abboud-triangle} states that triangle detection cannot be solved
in linear time.
% {\color{blue}Well,
%   the actual conjecture is formulated in a more complex way; we might
%   want to properly understand it and maybe even state it}
%
% In the case of single testing, the weak acyclicty
% separates queries that can be tested in linear time 
% from those that cannot.
%
\begin{restatable}{theorem}{lemmasingletestinglowerbound}  
\label{lemma:single-testing-lower-bound}  
Let $Q \in (\class{ELI},\class{CQ})$ be non-empty and self\=/join
free. If $Q$ is not weakly acyclic, single-testing complete answers to
$Q$ is not in linear time unless the triangle conjecture fails. The
same is true for minimal partial answers and minimal partial answers with
multiple wildcards.
\end{restatable}
The proof of Theorem~\ref{lemma:single-testing-lower-bound} is an
adaptation of the construction given in
\cite{BraultBaron,berkholz-enum-tutorial} where no ontologies are
considered. The challenge is to deal with the ontology and the fact
that the ontology may contain relation symbols that are not admitted
in the database. We address this by % using non-emptiness,
% and
modifying
the database construction from
\cite{BraultBaron,berkholz-enum-tutorial} so that every constant $c$
comes with fact $A(c)$ for every unary relation symbol $A\in \Sbf$ and
has an incoming and an outgoing $R$-edge for every binary relation
symbol $R \in \Sbf$. Informally,
this ensures that everything that could possibly be implied by the
ontology is indeed implied. Self-join freeness is important for this
approach to work.

While it would be desirable to replace $\class{ELI}$ with $\class{G}$
in Theorem~\ref{lemma:single-testing-lower-bound}, this seems hard to
achieve as it would also allow us to remove `self-join free' from
that theorem.  Even in the case without ontologies, it is currently
not known whether this is possible. 
\begin{example}
  \label{ex:sjf}
  Let $Q(\bar x)=(\Omc,\Sbf,q) \in (\class{G},\class{CQ})$ and let
  $Q'=(\Omc',\Sbf,q')$ be the OMQ that can be obtained from $Q$
  as follows: consider every atom $R(\bar z)$ in $q$, replace it
  with $R_{\bar z}(\bar z)$ where $R_{\bar z}$ is a fresh relation
  symbol of the same arity as $R$, and add to \Omc the TGDs
  $$
  R(\bar x) \rightarrow R_{\bar z}(\bar x) \text{ and }
  R_{\bar z}(\bar x) \rightarrow R(\bar x)
  $$
  where $\bar x$ is a tuple
  of $\mn{ar}(R)$ distinct variables. Then $Q \equiv Q'$, and $Q'$
  is self-join free. Moreover, $Q'$ is weakly acyclic if and only
  if $Q$ is.
\end{example}
More examples regarding Theorem~\ref{lemma:single-testing-lower-bound}
are given in the
\ifbool{arxive}{
  appendix. }
{
  long version.
}
{%\color{orange}
  We close with
  noting that the prerequisites given in Theorem~\ref{thm:singlelin}
  for the case of minimal partial answers cannot easily be relaxed.
\begin{restatable}{theorem}{anotherpartiallower}
  \label{thm:anotherpartiallower}
 (1) There is a weakly acyclic OMQ $Q \in (\class{ELI},\class{CQ})$ for
which single-testing minimal partial answers is not in linear
time unless the triangle conjecture fails and (2)~an %(2)~There is an
acyclic
  OMQ $Q \in (\class{G},\class{CQ})$ for which single-testing minimal
  partial answers with multi-wildcards is not in linear time
  unless the triangle conjecture fails.
\end{restatable}
%
% The OMQ used to prove the multi-wildcard part of
% Theorem~\ref{thm:anotherpartiallower} is in fact remarkably simple,
% with the CQ in it taking the form $q(x_1,x'_1,\dots,x_3,x'_3) = \bigwedge_{1 \leq i \leq 3} R(x_i,x'_i)$.
% }
}

% In the main body of the paper, we have defined minimal partial answers
% purely semantically. In view of Lemma~\ref{pro:chase}, it should not
% be surprising that they can equivalently be defined in terms of the
% chase.
% %
% Let $q(\bar x)$ be a CQ, $D$ a database, and
% $N \subseteq \mn{adom}(D)$ a set of constants, the \emph{nulls} in
% $D$. For an answer $\bar a \in q(D)$, we use $\bar a^\ast_N$ to denote
% the (unique) wildcard tuple for $D$ obtained by replacing all
% constants from $N$ with `$\ast$'. We call such an $\bar a^\ast_N$ a
% \emph{partial answer} to $q$ on $D$ w.r.t.\ $N$ and say that it is a
% \emph{minimal} partial answer if there is no $\bar b \in q(D)$ with
% $\bar b^\ast_N \prec \bar a^\ast_N$. We use $q(D)^{\ast}_N$ to denote
% the set of minimal partial answers to $q$ on $D$ w.r.t.\
% $N$. Similarly, we use $\bar a^{\Wmc}_N$ to denote the (unique)
% multi-wildcard tuple for $D$ obtained by consistently replacing
% constants from $N$ with wildcards using exactly a prefix of the
% ordered set $\Wmc = \{ \ast_1, \ast_2,\dots \}$ and respecting the
% order of variables in $\bar x$, that is, if the first occurrence of
% $x$ in $\bar x$ is before the first occurrence of $x'$, $x$ is
% replaced with $\ast_i$, and $x'$ is replaced with $\ast_j$, then
% $i < j$. We then define \emph{minimal partial answer with
%   multi-wildcards} to $q$ on $D$ w.r.t.\ $N$, denoted $q(D)^{\Wmc}_N$,
% in the expected way. The following is an immediate consequence of the
% definitions.

\section{Enumeration and All-Testing: Complete Answers}
\label{sect:enumallcomplete}

We consider the limits of enumeration and all-testing of complete
answers with constant delay for the OMQ languages
$(\class{G},\class{CQ})$ and $(\class{ELI},\class{CQ})$. While
enumeration is linked to the combination of acyclicity and free\=/connex acyclicity, we
link all-testing to free\=/connex acyclicity only. In the lower
bounds, we also consider minimal partial answers and minimal partial
answers with multiple wildcards.  We start with the upper bounds.
%\pagebreak
% 
\begin{restatable}{theorem}{thmupperG}
  \label{thm:upperG}
  In $(\class{G},\class{CQ})$,
  \begin{enumerate}

  \item enumerating complete answers is in \cdlin for OMQs that
    are acyclic and free\=/connex acyclic;

  \item all-testing complete answers is in \cdlin for OMQs that
    are free\=/connex acyclic.

%    weakly acyclic and have no bad path.

\end{enumerate}
\end{restatable}
Recall that for a CQ $q$ to be free-connex acyclic, we do \emph{not}
require~$q$ to be acyclic. Thus, the requirement for all-testing in
Theorem~\ref{thm:upperG} is significantly weaker than that for
enumeration and embraces, for example, every OMQ in which the CQ is
full, that is, has no quantified variables. The proof of Point~(1) of
Theorem~\ref{thm:upperG} uses the query-directed
chase also employed in Section~\ref{sect:singletesting} and a~reduction
to the \cdlin enumeration of answers to CQs (without ontologies)
that are acyclic and free\=/connex acyclic
\cite{bagan-enum-cdlin}. Point~(2) can be proved in the same way using
the following observation which, to our knowledge, is novel.
\begin{restatable}{proposition}{propallTestingCompleteUpper}
  \label{prop:allTestingCompleteUpper}
  For CQs (without ontologies) that are free-connex acyclic, 
  all-testing is in \cdlin. 
%    
  % For CQ (without ontologies) $q(\bar{x})$, if $\hat{q}(\bar{x}) \leftarrow q(\bar{x}) \land R(\bar{x})$, 
  % where $R$ is a fresh relation symbol, is acyclic 
  % then all-testing for $q$ is in CD$\circ$Lin. 
\end{restatable}
To prove
Proposition~\ref{prop:allTestingCompleteUpper}, we decompose the given
CQ into CQs that are acyclic and free-connex acyclic, and then use
\cdlin all-testing algorithms for those component CQs in
parallel.   In the
\ifbool{arxive}{
  appendix, }
{
  long version,
}
we give a matching
 lower bound for self-join free CQs.

We next give lower bounds that partially match
Theorem~\ref{thm:upperG}, starting with the
requirement in Point~(1) of
Theorem~\ref{thm:upperG} that OMQs must be
acyclic. The following is a 
consequence of Theorem~\ref{lemma:single-testing-lower-bound}.
% As in Section~\ref{sect:singletesting}, we do
% not obtain a full dichotomy, paralleling the corresponding
% result for CQs without ontologies~\cite{BraultBaron}. Since the
% proofs are almost identical, we also cover minimal partial answers.
%
\begin{restatable}{theorem}{lemmaenumerationlowerbound} 
\label{lemma:enumeration-lower-bound} 
Let $Q \in (\class{ELI},\class{CQ})$ be non-empty, and self\=/join
free. If $Q$ is not acyclic, then enumerating complete answers to $Q$
is not in \dlc unless the triangle conjecture fails. The
same is true for minimal partial answers and for minimal partial answers
with multiple wildcards.
\end{restatable}
In Theorem~\ref{lemma:enumeration-lower-bound}  and all other lower
bounds stated in this section, $\class{ELI}$ cannot easily be replaced
by $\class{G}$, see Example~\ref{ex:sjf}.

Staying with the requirements of Point~(1) of
Theorem~\ref{prop:allTestingCompleteUpper}, we next consider queries
that are acyclic, but not free\=/connex acyclic.  The lower bound that we
establish is conditional on an assumption regarding the problem of
Sparse Boolean matrix multiplication.  A Boolean $n \times n$ matrix
is a function $M:[n]^2 \rightarrow \{0,1\}$ where $[n]$ denotes the
set $\{1,\dots,n\}$.  The \emph{product} of two Boolean $n \times n$
matrices $M_1,M_2$ is the Boolean $n \times n$ matrix
$M_1M_2 := \sum_{c=1}^{n} M_1(a,c)\cdot M_2(c,b)$ where sum and
product are interpreted over the Boolean semiring. In (non-sparse)
\emph{Boolean matrix multiplication (BMM)}, one wants to compute
$M_1M_2$ given $M_1$ and $M_2$ as $n \times n$ arrays.  In
\emph{sparse Boolean matrix multiplication (spBMM)}, input and output
matrices $M$ are represented as lists of pairs $(a,b)$ with
$M(a,b)=1$. Our lower bound is conditional on the assumption that
spBMM is not possible in time $O(|M_1| + |M_2| + |M_1M_2|)$, that is,
in time linear in the size of the input and the output (represented as
lists). While it is not ruled out that such a running time can be
achieved, this would require dramatic progress in algorithm
theory. Informally, the conditioning on spBMM should be read as
`currently out of reach'.
\begin{restatable}{theorem}{thmlowerboundenumeli}
   \label{thm:lower-bound-enum-eli}
   Let $Q = (\Omc, \Sbf, q) \in (\class{ELI},\class{CQ})$ be acyclic,
   non\=/empty, self-join free, and connected. 
 If $Q$ is not
   free-connex acyclic, then enumerating complete answers to $Q$ is not in
   \dlc unless spBMM is possible in time
   $O(|M_1| + |M_2| + |M_1M_2|)$. The same is true for minimal
   partial answers and for minimal partial answers with multiple wildcards.
\end{restatable}
There is a corresponding lower bound for CQs without ontologies, first
proved conditional on the assumption that Boolean $n \times n$
matrices cannot be multiplied in time $O(n^2)$ \cite{bagan-enum-cdlin}
and then improved to the condition on spBMM used in
Theorem~\ref{thm:lower-bound-enum-eli}
in~\cite{berkholz-enum-tutorial}. To prove
Theorem~\ref{thm:lower-bound-enum-eli}, we again have to deal with the
fact that the ontology may contain relation symbols that are not
admitted in the database. Here, this is done by first manipulating the
input matrices $M_1$ and $M_2$ in a suitable way. Note that we
require $Q$ to be connected while this is not a precondition in the
case without ontologies~\cite{berkholz-enum-tutorial}.
%
% because the ontology forces us to use more constants in the
% reduction than in the case without ontologies, and
% without connectedness this may result in too many answers to attain
% the desired running tume for spBMM.
% Connectedness is not required in
% the case without ontologies.
% In our case, however, it cannot easily be
% dropped as shown by
The following proposition shows that we cannot drop connectedness.
\begin{restatable}{proposition}{connectedisneeded}
  \label{prop:connectedisneeded}
  There is an OMQ $Q \in (\class{ELI},\class{CQ})$ that is acyclic,
  non\=/empty, self-join free, but neither free-connex acyclic nor
  connected, such that complete answers to $Q$ can be enumerated
  in \dlc. 
\end{restatable}
We next address the requirement in Point~(2) of
Theorem~\ref{prop:allTestingCompleteUpper} that OMQs must be
free-connex acyclic.
\begin{restatable}{theorem}{thmalltestinglower}
  \label{thm:alltestinglower}
  Let $Q \in (\class{ELI},\class{CQ})$ be non-empty and self\=/join
  free.  If $Q$ is not free-connex acyclic, then all-testing complete
  answers for $Q$ is not in linear time unless the triangle conjecture
  fails or Boolean $n \times n$ matrices can be multiplied
  in time $O(n^2)$. {%\color{red}
    The same is true for minimal partial
    answers and minimal partial answers with multiple wildcards.}
\end{restatable}
Note that Theorem~\ref{thm:alltestinglower} refers to the non-sparse
version of BMM and that spBMM in time $O(|M_1| + |M_2| + |M_1M_2|)$
implies BMM in time $O(n^2)$ while the converse is unknown. 
% In the
% converse, we prove the same result also for CQs that may use relation
% symbols of any arity (without ontologies).

% The following is an immediate consequence of
% Theorem~\ref{lemma:single-testing-lower-bound}.
% %
% \begin{theorem}
%   Let $Q \in (\class{ELI},\class{CQ})$ be non-empty, and self\=/join
%   free.  If $Q$ is not weakly acyclic, then all-testing for $Q$ is not
%   in linear time unless the triangle conjecture fails.
% \end{theorem}
%
% {\color{blue}define/discuss BMM, make sure that ``$n^2$'' is understood}.
%
% \begin{theorem}
%    \label{thm:lower-bound-enum-eli}
%    Let $Q = (\Omc, \Sbf, q) \in (\class{ELI},\class{CQ})$ be weakly
%    acyclic, non\=/empty, self-join free, and connected.  If $Q$ has a
%    bad path, then all-testing of complete answers to $Q$ is not in
%    CD$\circ$Lin unless BMM is possible in time $O(n^2)$, and the same
%    is true for minimal partial answers.
% \end{theorem}

\section{Enumeration with Single Wildcard}
\label{sect:LPAsingleWildcardUpper}

The main aim of this section is to prove that it is possible to
enumerate in \dlc the minimal partial answers with a single wildcard
to OMQs from $(\class{G},\class{CQ})$ that are acyclic and
free\=/connex acyclic. Thus, minimal partial answers are almost as
well-behaved as complete answers, except that for the former it
remains open whether enumeration is also possible in \cdlin. We start,
however, with observing that all-testing of minimal partial answers is
less well-behaved. The following should be contrasted with Point~(2) of
Theorem~\ref{thm:upperG}.
\begin{restatable}{theorem}{thmlowerGpartial}
  \label{thm:lowerGpartial}
  There is an OMQ $Q \in (\class{ELI},\class{CQ})$ that is acyclic and
  free-connex acyclic 
  such that all-testing minimal partial answers to $Q$ is
  not in \dlc unless the triangle conjecture fails.
  The same is true for minimal partial answers with multiple
    wildcards.
\end{restatable}
Intuitively, all-testing of minimal partial answers is difficult
because % less well-behaved
% than all-testing of complete answers since
a single positive test for an answer that contains wildcards may imply
a negative test for polynomially many complete answers. This is not a
problem in enumeration where the `problematic' wildcard answers will
be output late and thus cannot be tested in
linear time.
% nor
% use them to implement triangle detection in linear time.

We now turn to the main result of this section. % enumeration of minimal partial answers (with a single wildcard). The following
% is a main result of this paper.
%
\begin{theorem}
  \label{thm:upperGpartial}
  Enumerating minimal partial answers is in \dlc for OMQs from
  $(\class{G},\class{CQ})$ that are acyclic and free\=/connex acyclic.
\end{theorem}
In the rest of this section, we prove Theorem~\ref{thm:upperGpartial}
by developing an enumeration algorithm. {%\color{red}
  We provide
  an example that  illustrates important aspects of our algorithm in
  Appendix~\ref{sect:example}.}  Fix an OMQ
$Q(\bar x)=(\Omc,\Sbf,q_0) \in (\class{G},\text{CQ})$ with $q_0$
acyclic and free\=/connex acyclic, and let an \Sbf-database $D$ be
given as input. 

\paragraph{\bf Preprocessing phase.}
Recall from Section~\ref{sect:singletesting} that the query-directed
chase $\mn{ch}^{q_0}_\Omc(D)$ can be constructed in time linear in
$||D||$. This is the first step of the preprocessing phase.  By
Lemmas~\ref{lem:twosemantics} and~\ref{prop:restrictedchaseworks}, we
may enumerate $q_0(\mn{ch}^{q_0}_\Omc(D))^{\ast}_\Nbf$ in place of
$Q(D)^\ast$. For brevity, set
$D_0 :=\mn{ch}^{q_0}_\Omc(D)$.

% To
% describe the remaining steps, we introduce some notation.

% Let $D'$ be
% a database and assume that a set $N \subseteq \mn{adom}(D')$ of
% constants is to be viewed as nulls in $D'$. Further, let $q$ be a CQ.  For
% an answer $\bar a \in q(D')$, we use $\bar a^\ast_N$ to denote the
% wildcard tuple for $D'$ obtained by replacing all constants from $N$
% with~`$\ast$'. We call such an $\bar a^\ast_N$ a \emph{partial answer
% to $q$ on $D'$ w.r.t.~$N$} and say that it is a \emph{minimal} partial
% answer if there is no $\bar b \in q(D')$ with
% $\bar b^\ast_N \prec \bar a^\ast_N$. With $q(D')^{\ast}_N$, we denote
% the set of minimal partial answers to $q$ on $D'$ w.r.t.~$N$.

% In $\mn{ch}^{q_0}_\Omc(D)$, the set of nulls is
% $N=\mn{adom}(\mn{ch}^{q_0}_\Omc(D)) \setminus \mn{adom}(D)$ while we
% speak of the constants in $\mn{adom}(D)$ as the \emph{database
%   constants}.
% We observe in the appendix
% that
% $Q(D)^{\ast}=q_0(\mn{ch}^{q_0}_\Omc(D))^{\ast}_N$. It thus suffices to
% enumerate $q_0(\mn{ch}^{q_0}_\Omc(D))^{\ast}_N$, which is what we
% concentrate on in the
% following.  

We
\ifbool{arxive}{
  argue in the appendix that we}
can assume w.l.o.g.\ that the tuple
$\bar x$ has no repeated variables and that $q_0$ contains no
constants and is connected. \ifbool{arxive}{}{For connectedness, see
  Appendix~\ref{app:connected}, for the other properties see the long version.
}
As part of the preprocessing phase, we
preprocess $q_0$ and $D_0$ in a way that resembles the first phase of
the Yannakakis algorithm in which a join tree is traversed in a bottom-up fashion, computing a semi-join in each
step~\cite{yannakakis-algotrithm}.
%
%While doing this, we also make
% $q_0$ self\=/join free, remove all quantified variables, and achieve
% some additional normalisation.
The result is a CQ
$q_1(\bar x)$ and database $D_1$ that satisfy the following
conditions:
{
  \renewcommand{\theenumi}{\roman{enumi}}%
\begin{enumerate}

\item $q_1$ is self\=/join free, connected
  (since $q_0$ is), acyclic, and has no quantified variables
  (thus is free\=/connex
    acyclic); it therefore has a join tree $T_1=(V_1,E_1)$; we choose
  a root in $T_1$ %thus imposing a direction on $T_1$ and
  allowing us to speak about predecessors and successors in $T_1$;

% \item the atoms of $q_1$ have no repeated variables; moreover,
%   if $v =R(\bar y) \in V_1$ has predecessor $v'$ in $T_1$, then
%   the variables shared between $v$ and $v'$ occur in $\bar y$
%   before all other variables; {\color{blue}STILL NEEDED? WHERE?}

\item $\mn{adom}(D_1) \subseteq \mn{adom}(D_0)$ and for every fact
  $R(\bar a) \in D_1$, there is a fact $S(\bar b) \in D_0$ such that
  $\bar a$ and $\bar b$ contain exactly the same (database and null)
  constants;
% \footnote{It is neither guaranteed nor important that all
%     constants from $N$ occur in $D_1$.}

\item $q_0(D_0)=q_1(D_1)$, and thus
  $q_0(D_0)^{\ast}_\Nbf=q_1(D_1)^{\ast}_\Nbf$;
  
\item \label{progress}
  for all $v=R(\bar y) \in V_1$, facts $R(\bar a) \in D_1$, and
  successors $v'=S(\bar z)$ of $v$ in $T_1$, $D_1$ contains a fact
  $S(\bar b)$ such that if position $i$ of $\bar y$ has the same
  variable as position $j$ of $\bar z$, then position $i$ of $\bar a$
  has the same constant as position $j$ of $\bar b$.

\end{enumerate}
}
We refer to Condition~(\ref{progress}) as the \emph{progress
  condition}. Informally, it makes sure that an enumeration algorithm
that traverses $T_1$ in a pre-order tree walk never gets `stuck' in
the sense that it can always extend the partial answer produced so far
to a full answer.  The construction of $q_1$ and $D_1$ is possible in
time linear in $||D_0||$. It has been used many times in the context
of enumerating answers to conjunctive queries (without ontologies)
with constant delay.  We give an outline in the
\ifbool{arxive}{
  appendix}
{
  long version
}
and refer to
\cite{berkholz-enum-tutorial} for a very clear exposition of the full
details.
{%\color{purple}
  The construction of $q_1$ and $D_1$ also tells
  us whether $q_0(D_0)=\emptyset$. If this is the case, we stop
  without entering the enumeration phase.}

\smallskip

We also use the preprocessing phase to compute data structures that
are used in the enumeration phase. We start with some preliminaries.
With a \emph{predecessor variable} in an atom $v \in V_1$, we mean a
variable that $v$ shares with its predecessor in~$T_1$.
{%\color{blue}
By definition,
the root of $T_1$ does not have any predecessor variables.} A
CQ $q$ is a \emph{subtree} of $q_1$ if there is a subset
$V_q \subseteq V_1$ such that the subgraph
$T_q=(V_q,E_1|_{V_q \times V_q})$ of $T_1$ induced by $V_q$ is
connected. Note that $q$ must be connected since $q_1$ is and that
$T_q$ is a join tree for $q$. We assume that $T_q$ inherits the
direction imposed on $T_1$ and thus, for instance, may speak about its
root.
% A \emph{fringe variable} of a subtree $q$ of $q_2$ is a variable
% $x$ that occurs in a leaf $u$ of $T_q$ that has a successor $u'$ in
% $T_2$ in which $x$ also occurs.

% {\color{blue}Let $S$ be a non-empty \emph{no-null guarded set} in $D_2$, that is, a guarded
% set in $D_2$ such that $S \cap N = \emptyset$.}
% \footnote{We remark that
%   if $S$ is a no-null guarded set, then by definition of `chase-like'
%   $D_2$ must contain a fact that contains all constants in $S$, but no
%   null.}
A \emph{progress tree} is a pair $(q,g)$ with $q$ a subtree of $q_1$
and
$g:\mn{var}(q) \rightarrow (\mn{adom}(D_1) \setminus N) \cup \{ \ast
\}$ a map such that the following conditions are satisfied:
\begin{enumerate}

\item $g(x) \neq \ast$ for every predecessor variable $x$ in the root
  of $T_q$; 

% \item if $v \in V_q$ and $v' \notin V_q$ is a successor of $v$ in $T_2$,
%   then $g(x)\neq\ast$ for all predecessor variables $x$ in $v'$;

% \item if $v,v' \in V_q$ with $v'$ a successor of $v$ in $T_2$, then
%   $g(x) =\ast$ for some predecessor variable $x$ in $v'$;

   \item if $v \in V_q$ and $v'$ is a successor of $v$ in $T_1$, then
   $v' \in V_q$ if and only if $g(x) = \ast$ for some predecessor
   variable $x$ in $v'$;

\item there is a homomorphism $h$ from $q$ to $D_1$ such that for all
  $x \in \mn{var}(q)$, $h(x) \in N$ if $g(x) = \ast$ and
  $h(x)=g(x)$ otherwise;

\item the constants in the range of $g$ form a guarded set in $D_1$.
% {\color{red}I suspect that Condition~(4) of 
%   progress 
%   trees is actually implied by the other conditions?}
  
\end{enumerate}
%
% Note that Condition~(2) implies that $g(x) \neq \ast$ for all fringe
% variables $x$ of a progress tree $(q,g)$. 
To explain the intuition of
progress trees, consider a homomorphism $h$ from $q_1$ to $D_1$ and an
atom $v =R(\bar y)\in V_1$ with
% predecessor $v'$ and
predecessor variables $\bar z$.  If $h(\bar y) \cap N =\emptyset$,
then $(v,g)$ is a (single atom) progress tree, $g$ the restriction of
$h$ to the variables in~$\bar y$. More interesting is the case where
$h(\bar z) \cap N =\emptyset$, but $h(\bar y) \cap N \neq
\emptyset$. Informally, under homomorphism $h$ such an atom $v$
`crosses the boundary' between the `database part' of $D_1$ and the
`null part' of $D_1$.  Let $V_q \subseteq V_1$ be the smallest set
that contains $v$ and such that if $u \in V_q$ and $u'$ is a successor
of $u$ in $T_1$ such that $h(x) \in N$ for at least one predecessor
variable in $u'$, then $u' \in V_q$. This defines a subtree $q$ of
$q_1$ and $(q,g)$ is then a progress tree, where $g$ is the
restriction of $h$ to the variables in $q$ with constants from $N$
replaced by $\ast$.  Informally, $(q,g)$ thus describes an `excursion'
of the part $q$ of $q_1$ into the `null part' of $D_1$ and it turns
out that properly dealing with such excursions is key to enumerating
minimal partial answers.  Note that the constants in the range of $g$
form a guarded set in~$D_1$, as required. This relies on $q_1$ being
connected as otherwise, it would be possible to cross the boundary to
the null part of $D_1$ at some guarded set, but return to the database
part at a different guarded set.

% {\color{violet} do we care about order on constants? I mean we will be sorting lists,
% so we technically need one, but it does not matter so we may not want to complicate...}
% I agree, let's leave this implicit

Consider an atom $v$ in $q_1$ with predecessor variables~$\bar z$. A
\emph{predecessor map} for $v$ is a function
$h:\bar z \rightarrow \mn{adom}(D_1) \setminus N$ 
{%\color{purple}
that extends to a homomorphism from $v$ to $D_1$.
% whose range is a
% guarded set in~$D_1$.
We call such $v$ and $h$ \emph{relevant}. For all relevant $v$ and
$h$,} we compute a linked list $\mn{trees}(v,h)$ of all progress trees
$(q,g)$ with root $v$ such that $g(\bar z)=h(\bar z)$.
% for some fact $R(\bar c) \in D_1$,
% the following holds:
% %
% \begin{enumerate}
%
% %\item $(q,g)$ has root $v$;
%
% \item[(a)] $g(\bar z)=h(\bar z)$;
%   % $R(\bar c)$ unifies with
% %  $R(h(\bar z)\bar u)$;
%
% \item[(b)] all database constants in the range of $g$ are from
%   $\bar c \setminus N$;
%
% \item[(c)] $g(\bar z)=h(\bar z)$.
%
% \end{enumerate}
%
% Note that if $D_1$ contains a fact $R(\bar c)$
% that satisfies Condition~(a) and contains no nulls, then 
% $(v, \bar y \mapsto \bar c)$ is a (single node) candidate tree in 
% $\mn{trees}(v,h)$.
%
  We sort the list $\mn{trees}(v,h)$ so that it is in
  \emph{database\=/preferring order}. This means that
  progress tree $(q,g)$ is before progress tree $(q',g')$
  whenever $(q,g) \prec_{\mn{db}} (q',g')$, which is the case
  if
%
%   it is in
% \emph{database\=/preferring order} `$\prec_{\mn{db}}$', meaning that
% $(q,g) \prec_{\mn{db}} (q',g')$ if
$q$ and $q'$ have the same root and
$V_q \subsetneq V_{q'}$, or the following conditions are satisfied
for all $x \in \mn{var}(q)$: 
\begin{itemize}

\item[(a)]  $V_q=V_{q'}$;
  
\item[(b)] $g(x) = \ast$ implies $g'(x) = \ast$;
  
\item[(c)] $g(x) \neq \ast$ implies $g'(x) \in \{g(x),\ast\}$;
   
% \item for all $x \in \mn{var}(q)$:
%   %
%   \begin{itemize}

%   \item   $g(x)\neq \ast$ implies $g'(x) \in
%     \{ g(x), \ast\}$ and

%   \item $g(x)=\ast$ implies $g'(x)=\ast$;
  
%   \end{itemize}

\item[(d)]  for some $x \in \mn{var}(q)$, $g'(x)=\ast$ while 
  $g(x) \neq \ast$.
  
\end{itemize}
The algorithm uses these lists as a global data structure that is both
accessed and modified.
We show in the
\ifbool{arxive}{
  appendix}
{
  long version
}
that the lists $\mn{trees}(v,h)$ can indeed be computed
in linear time on a RAM.
\begin{restatable}{lemma}{lemprecomputetreeslistslinear}
\label{lem:precomputetreeslistslinear}
  The lists $\mn{trees}(v,h)$, for all relevant $v$ and $h$, can be
  computed in overall time linear in $||D_1||$. {%\color{purple}
  Moreover, all these lists are non-empty.}
\end{restatable}

\medskip

% Note that Condition~(3) of $S$-candidate trees can be checked in
% linear time: first modify $q$ by replacing each variable $x$ with
% $g(x)$ if $g(x)$ is a constant and adding $P_N(x)$ if $g(x)=\ast$,
% $P_N$ a fresh unary relation symbol; extend $D_1$ by adding $P_N(c)$
% for all $c \in N$. It then remains to check the existence of a
% homomorphism from the modified $q$ viewed as a Boolean CQ to the
% modified $D_1$, which is possible in time linear in $||D_1||$ since
% $q$ is acyclic \cite{Yannakakis}. Also note that the number of
% $S$-candidate trees in each list $\mn{trees}(v,S)$ is bounded by a
% constant. While $S$-candidate trees $(q,g)$ use constants from $D_1$
% in the $g$-part, they only use constants from the guarded set $S$
% whose cardinality is bounded by the maximum arity of relation symbols
% and thus by a constant since $Q$ is fixed. We can thus put these lists
% onto the right order and fetch first and next elements from them, all
% in time independent of $||D_1||$ {\color{blue}first elements: how
%   exactly?} and even iterate over an entire such list in time
% independent of $||D_1||$.
%This finishes the description of the
%preprocessing phase. 

%
%As a special case, the empty set is always in $\mn{trees}(x,a)$.
%
Let $v_0,\dots,v_k$ be the ordering of the atoms in $V_1$ generated by
a pre-order traversal of $T_1$. For $v_i \in \{v_0,\dots,v_k\}$ and a
partial map
$h:\mn{var}(q_1) \rightarrow (\mn{adom}(D_1) \setminus N) \cup \{ \ast
\}$, we use $\mn{nextat}_h(v_i)$ to denote $v_j$ with $j>i$ smallest
such that $h(x)$ is undefined for some variable $x$ in $v_j$, if such
$j$ exists, and the special symbol $\mn{eoa}$ (\emph{end of atoms})
otherwise. Clearly, computing $\mn{nextat}$ is independent of
$||D_1||$ and can thus be done in constant time.

\paragraph{\bf Enumeration Phase.}
The enumeration phase of the algorithm is presented in
Figure~\ref{fig:enumalg}. In the {\bf forall} loop in Line~10, we
follow the database-preferring order imposed on the \mn{trees} lists.
{%\color{purple}
  It is straightforward to show the invariant that when
  a call $\mn{enum}(v,h)$ is made,  then $v,h|_{\bar z}$ used
  in Line~12 is relevant. The following is an important
  observation.
  \begin{restatable}{lemma}{lemneveremptylem}
    \label{lem:neveremptylem}
  None of the lists $\mn{trees}(v,h)$, with $v,h$ relevant, ever
  becomes empty. 
\end{restatable}
}
Lemma~\ref{lem:neveremptylem} is important to achieve constant delay
because it implies that that in each call $\mn{enum}(v,h)$,
%
% (i)~each list
% $\mn{trees}(v,h|_{\bar z})$ contains at least one progress tree (due
% to the progress condition) and (ii)~{\color{orange}the pruning of progress trees from
% such lists that takes place as part of the enumeration phase can never
% remove all trees from a list.}
% % \footnote{This is because every list
% %   contains
% %   a (single-node) progress tree without wildcards, and such trees are
% %   never removed.}
% %
% %
%   Thus,
the {\bf forall} loop in Line~10 %in \mn{enum}
makes at least one iteration and thus at least one recursive
call in Line~12. Consequently, while traversing $q_1$ we never backtrack without
producing an output. Note that given $v$ and $h|_{\bar z}$, we need to
find the (first element of the) list
$\mn{trees}(v,h|_{\bar z})$ in constant time. On a RAM, this can be
achieved by a straightforward lookup table.
% It is
% easy to verify that whenever we access some
% $\mn{trees}(v,h|_{\bar z})$, then the range of $h|_{\bar z}$ is a
% guarded set, as required.

 In the \mn{prune} subprocedure, there are only constantly
  many progress trees $(q,g)$ with
  $(q,g) \succ_{\mn{db}}(q, h|_{\mn{var}(q)})$ and these can be found
  in constant time by starting with $g=h|_{\mn{var}(q)}$ and then
  choosing one or more variables $x \in \mn{var}(q)$ with
  $g(x) \neq \ast$ and setting $g(x)=\ast$. 
    Note that the
    pair $(q, h|_{\mn{var}(q)})$ is neither required nor guaranteed to be a
    progress tree. % It is nevertheless well-defined to
    % use`$\prec_{\mn{db}}$' on such pairs.
    % We further remark that if
    % $(q,g) \succ_{\mn{db}}(q, h|_{\mn{var}(q)})$, then it can be
    % verified that $(q,g)$ is a progress tree except that the `only if'
    % direction of Condition~(2) of progress trees need not be
    % satisfied. 
    To remove $(q',g')$ from
  $\mn{trees}(v,h|_{\bar z})$, it is not possible to iterate over all
  progress trees in $\mn{trees}(v,h|_{\bar z})$ in search of $(q',g')$
  as there may be linearly many trees in the list. This problem is
  also solved by a lookup table. When generating the \mn{trees} lists
  in the preprocessing phase, we also generate a lookup table that takes as
  argument a progress tree and yields the memory location (register)
  where that tree is stored as part of a list $\mn{trees}(v,h)$. Note
  that every progress tree occurs in at most one such list. If the
  list is bidirectionally linked, it is then easy to locate and remove
  the tree in constant time.
%%<
%% In the two consecutive {\bf forall}
%% loops in the \mn{enum} function, we implicitly use the two access
%% lists computed in the preprocessing phase.

%\hideAlgorithm{
%\begin{figure}[t!]
% \centering
%\begin{algorithm}[H]
% \SetAlgoLined
% \SetKwProg{Fn}{Function}{}{end}
% % $h_0=\{x_0 \mapsto c_0\}$; \hfill \emph{\% $x_0, c_0$ as 
% %   in construction of $q_2$, $D_2$}\\
% $h_0 = \emptyset$;\\
% % \mn{enum}$(\mn{nextat}_{h_0}(P_0(x_0)),h_0)$\;
% % $~$\\[-2mm]
% \mn{enum}$(\mn{nextat}_{h_0}(v_0),h_0)$; \hfill
% \emph{\% $v_0$ root of $T_1$} \\[2mm]
% \Fn{\mn{enum}$(v,h)$}
% {
%  \If{$v=\mn{eoa}$}{
%    output $h(\bar x)$; \hfill \emph{\% $\bar x$ the variables in $q_1$}\\
%    \mn{prune}$(h)$\;
%    {\bf return} \\[1mm]
%  }
%  let $v=R(\bar y)$ with predecessor variables $\bar z$\;
%    \ForAll{$(q,g) \in \mn{trees}(v,h|_{\bar z})$}{
%    $h' = h \cup g$ \;
%      \mn{enum}$(\mn{nextat}_{h'}(v),h')$\; 
%  }
%{\bf return} 
%}
%~\\[2mm]
%
%\Fn{\mn{prune}$(h)$}{ \ForAll{subtrees $q$ of $q_1$} {
%      let $v$ be the root of $q$ with predecessor variables $\bar z$\;
%        \ForAll{progress trees $(q,g) \succ_{\mn{db}} (q, h|_{\mn{var}(q)})$ 
%     }{
%     remove $(q,g)$ from $\mn{trees}(v,h|_{\bar z})$\;
%   }
% }
%{\bf return} 
%}
%\end{algorithm}
% \caption{Enumeration of minimal partial answers}
% \label{fig:enumalg}
%\end{figure}
%}%hideAlgorithm

\begin{algorithm}[t]
    \caption{Enumeration of minimal partial answers.}
    \label{fig:enumalg}
    \begin{algorithmic}[5]
        \State $h_0 = \emptyset$;
        \State \mn{enum}$(\mn{nextat}_{h_0}(v_0),h_0)$;
        \\
        \Function {\mn{enum}}{$v,h$}
            \If{$v=\mn{eoa}$}
                \State output $h(\bar x)$; \hfill \emph{\% $\bar x$ the variables in $q_1$}
                \State \mn{prune}$(h)$;
                \State {\bf return}
            \EndIf
            \State let $v=R(\bar y)$ with predecessor variables $\bar z$;
            \ForAll{$(q,g) \in \mn{trees}(v,h|_{\bar z})$}
                \State  $h' = h \cup g$;
                \State \mn{enum}$(\mn{nextat}_{h'}(v),h')$;
            \EndFor
            \State {\bf return}
        \EndFunction
        \\
        \Function {\mn{prune}} {$h$}
            \ForAll{subtrees $q$ of $q_1$}
                \State let $v$ be the root of $q$ with predecessor variables $\bar z$;
                \ForAll{progress trees $(q,g) \succ_{\mn{db}} (q, h|_{\mn{var}(q)})$ }
                \State remove $(q,g)$ from $\mn{trees}(v,h|_{\bar z})$\;
                \EndFor
            \EndFor
            \State {\bf return}
        \EndFunction
    \end{algorithmic}
\end{algorithm}

% It follows from what was said above that preprocessing only takes time
% linear in $||D||$ while the enumeration delay is constant, that is, it
% only depends on $Q$, but not on $D$.
 %{\color{blue}we need to argue that the right \mn{trees} list
%   can be found in constant time on a RAM!} 
In the
\ifbool{arxive}{
  appendix, }
{
  long version,
}
we prove that the algorithm achieves its goal.
\begin{restatable}{proposition}{proppartialalgocorrect}
  \label{prop:partialalgocorrect}
  The algorithm outputs exactly the minimal partial answers to  $q_1$
  on~$D_1$, without repetition.
\end{restatable}

\section{Enumeration With Multi-Wildcards}
\label{sect:enummulti}

\newcommand{\SPA}{{\textit{Ans}}}

We show that Theorem~\ref{thm:upperGpartial} lifts from the case of a
single wildcard to the case of multi\=/wildcards. 
\begin{theorem}
    \label{thm:upperGspartial}
    Enumerating minimal partial answers with multi- wildcards is in
    \dlc for OMQs from $(\class{G},\class{CQ})$ that are
    acyclic and free\=/connex acyclic.
\end{theorem}
Fix an OMQ $Q(\bar x)=(\Omc,\Sbf,q_0) \in (\class{G},\text{CQ})$ with
$q_0$ acyclic and free\=/connex acyclic and let an \Sbf-database $D$
be given as input. By Lemmas~\ref{lem:twosemantics}
and~\ref{prop:restrictedchaseworks}, we may enumerate
$q_0(\mn{ch}^{q_0}_\Omc(D))^{\Wmc}_\Nbf$ in place of~$Q(D)^\Wmc$. For brevity, we from now on use $D$ to
denote $\mn{ch}^{q_0}_\Omc(D)$ (and we will never refer back to the
original $D$).

Our general approach to enumerating $Q^\Wmc(D)$ is to combine the
enumeration algorithm from Theorem~\ref{thm:upperGpartial}, here
called $A_1$, with a \dlc algorithm for all-testing (not necessarily
minimal) partial answers with multi-wildcards. In fact, we develop
such an algorithm $A_2$ in \ifbool{arxive}{ the appendix, }
{Appendix~\ref{app:alltestingmulti}, } which is
non-trivial.
{%\color{blue}
  The algorithm involves a multi-wildcard
  version of progress trees and running in parallel \dlc all-testing
  algorithms for complete answers to any subquery $q'$ of $q$, that
  is, to any CQ $q'$ that can be obtained from $q$ by dropping atoms.}

With
that algorithm in place, a first implementation of the general
approach could then be as follows.  Use $A_1$ to enumerate
$Q^*(D)$. For each obtained answer~$\bar a^\ast$, construct the
\emph{multi-wildcard ball of~$\bar a^\ast$}, that is, the set
$B^{\Wmc}(\bar{a}^{\ast})$ of multi-wildcard tuples $\bar a^\Wmc$ such
that replacing all occurrences of wildcards from \Wmc in $\bar a^\Wmc$
by the single-wildcard `$\ast$' results in~$\bar a^\ast$. Notice that
if the length of $\bar a^\ast$ is bounded by a constant, then so is
the cardinality of the multi-wildcard ball of~$\bar{a}^\ast$.  Discard
from $B^{\Wmc}(\bar{a}^{\ast})$ those tuples that are not partial
answers using $A_2$, and then output those among the remaining tuples
that are minimal w.r.t.~`$\prec$'. This first implementation is
incomplete.
\begin{example}
\label{ex:multiwildproblem}
  Let $Q=(\Omc,\Sbf,q_0)$ where
  $$
    \Omc = \{ A(x) \rightarrow \exists y_1 \exists y_2 \,
        R(x,y_1) \wedge T(x,y_1) \wedge S(x,y_2) \},
  $$
  \Sbf contains all relation symbols in \Qmc, and
  $$
     q_0(x_0,x_1,x_2,x_3)=R(x_0,x_1) \wedge S(x_0,x_2) \wedge T(x_0,x_3).
  $$
  Further let 
  $ D = \{ A(c), R(c,c')\}$.  Then $Q^*(D) = \{ (c,c',\ast,\ast) \}$
  and $Q^\Wmc(D)=\{(c,c',\ast_1,\ast_2),
  (c,\ast_1,\ast_2,\ast_1)\}$. But we
never consider (and thus do not output) the multi-wildcard tuple
  $(c,\ast_1,\ast_2,\ast_1)$.
\end{example}
The solution % for the problem illustrated in
% Example~\ref{ex:multiwildproblem}
involves replacing the multi-wildcard ball $B^{\Wmc}(\bar{a}^{\ast})$
with the \emph{multi-wildcard cone}
$$
\mn{cone}^\Wmc(\bar{a}^{\ast}) = \bigcup_{\bar{b}^{\ast} : \bar{a}^{\ast} \preceq \bar{b}^{\ast}} B^{\Wmc}(\bar{b}^{\ast}).
$$
Clearly, also the 
cardinality of $\mn{cone}^\Wmc(\bar{a}^{\ast})$ is bounded by a 
constant 
if the length of $\bar a^\ast$ is. 
Regarding Example~\ref{ex:multiwildproblem}, note that
$(c,\ast_1,\ast_2,\ast_1) \notin B^\Wmc(c,c',\ast,\ast)$, but
$(c,\ast_1,\ast_2,\ast_1) \in \mn{cone}^\Wmc(c,c',\ast,\ast)$.
However, the cones of different tuples $\bar a^\ast,\bar b^\ast \in
Q^*(D)$ might overlap and thus for some
$\bar a^\ast \in Q^*(D)$, there might be no tuple in
$\mn{cone}^\Wmc(\bar{a}^{\ast})$ that we haven't yet output, 
compromising constant delay. We address these issues by using a careful combination of balls, cones, and pruning.

% \subsection{Lower Bounds}
% \label{sect:enumlower}

% We prove lower bounds for the enumeration of (complete and least
% partial) answers to OMQs that are not acyclic or acyclic but not
% free-connex acyclic. We also prove that all-testing is less
% well-behaved than enumeration in the case of least partial answers in
% the sense that there is an OMQ from $(\class{ELI},\class{CQ})$ that is
% acyclic and free-connex acyclic, but for which all-testing is not in
% CD$\circ$Lin. As in Section~\ref{sect:singlelower}, we obtain our
% lower bounds only for ontologies formulates in $\class{ELI}$ and
% observe that an extension to $\class{G}$ is probably quite difficult.
% Also as in Section~\ref{sect:singlelower} and as in the case without
% ontologies \cite{berkholz-enum-tutorial}, we further have to assume
% that OMQs are self-join free. Our lower bounds are conditional on
% assumptions whose failure would imply a remarkable advance in
% algorithm theory.

%\input{lower-bounds-enum-eli-pods}

%\section{Functional Roles}

We now describe our algorithm in full detail.
The preprocessing phase consists of running the preprocessing phases
of $A_1$ and $A_2$. The enumeration phase is shown
in Figure~\ref{fig:enum-alg-strong}. 
\begin{algorithm}[t]
    \caption{Enumeration of minimal partial answers with multi-wildcards.}
    \label{fig:enum-alg-strong}
\begin{algorithmic}[3]
    \State $L = [];$
    \ForAll {$\bar{a}^{\ast} \in q(D)^*_\Nbf$}
        \ForAll {$\bar{a}^{\Wmc} \in \mn{cone}^\Wmc(\bar a^\ast) \cap q(D)_\Nbf^{\Wmc,\not\prec}$ with $F(\bar{a}^{\Wmc}) = 0$}
                    \State $F(\bar{a}^{\Wmc}) = 1$;
                    \State append $\bar{a}^{\Wmc}$ to $L$;
                    \State \mn{prune}($\bar{a}^{\Wmc}$) 
        \EndFor
        \State choose  $\bar{a}^{\Wmc} \in \min^{\prec}(B^{\Wmc}(\bar{a}^{\ast}) \cap q(D)_\Nbf^{\Wmc,\not\prec})$;
        \State output $\bar{a}^{\Wmc}$;
        \State remove $\bar{a}^{\Wmc}$ from $L$;
    \EndFor
    \State output all tuples in $L$;
    \State {\bf return}
        \\
    \Function {\mn{prune}}{$\bar{a}^\Wmc$}
        \ForAll{multi-wildcard tuples $\bar{b}^{\Wmc}$ such that $\bar{a}^{\Wmc} \prec \bar{b}^{\Wmc}$}
            \State $F(\bar{b}^{\Wmc}) =1$;
            \State remove $\bar{b}^{\Wmc}$ from $L$;
        \EndFor
        \State {\bf return}
    \EndFunction
\end{algorithmic}
\end{algorithm}
With $L$, we denote a bidirectionally linked
list in which we store multi-wildcard tuples and that is initialized
as the empty list.
In the {\bf forall} loop in Line~2, we use algorithm $A_1$ to iterate over all minimal partial answers in
$q(D)_\Nbf^\ast$. With $q(D)^{\Wmc,\not\prec}_\Nbf$, we denote the set
of (not necessarily minimal) partial answers with multi-wildcards to CQ
$q$ on database $D$.
 The intersections with $q(D)_\Nbf^{\Wmc,\not\prec}$
in Line~3 and~7 can be computed in
constant time using algorithm~$A_2$. $F$ is a lookup table that stores
a Boolean value for every multi-wildcard tuple of length $|\bar x|$,
initialized with~0; this is done implicitly as all memory is initialized with~0 in
our machine model. Informally, $F(\bar a^\Wmc)$ is set to~1 if
$\bar a^\Wmc$ has already been added to the list $L$ or is not in
$q(D)_\Nbf^\Wmc$ (and thus does not need to be added to $L$).
For a set of  multi-wildcard tuples $S$, we use 
$\mn{min}^{\!\prec}(S)$ 
to denote the tuples in $S$ that are minimal  
w.r.t.\ `$\prec$'. 
To remove
multi-wildcard tuples from $L$ in constant time, we use another lookup
table that stores, for every multi-wildcard tuple $\bar a^\Wmc$ that
we have added to $L$, the memory location of the list node
representing $\bar a^\Wmc$ on $L$. Since $L$ is bidirectionally
linked, this allows us to remove $\bar a^\Wmc$ from $L$ in constant
time. Since the arity of relation symbols is (implicitly) bounded by a
constant, so is the number of iterations of the {\bf forall} loop in
Line~14.  From what was said above, it follows that
the preprocessing phase runs in linear time while the enumeration
phase has only constant delay. Correctness is proved in the
\ifbool{arxive}{
  appendix.}
{
  long version.
}
\begin{restatable}{lemma}{lemmultiwildcorr}
\label{lem:multiwildcorr}
  The algorithm outputs exactly the minimal partial answers with
  multi-wildcards to $q$ on $D$, without repetition.
\end{restatable}
%
%
% \begin{lemma}
%     \label{lemma:LSPA-enum-in-cdlin}
%     The algorithm \emph{LSPA-enum} is in CD$\circ$Lin. 
% \end{lemma}

\section{Conclusions}

As future work, it would be interesting to consider as the ontology
language also description logics with functional roles such as
$\mathcal{ELIF}$; there should be a close connection to enumeration of
answers to CQs in the presence of functional dependencies
\cite{carmeli-enum-func}. A much more daring extension would be to
$(\class{G},\class{UCQ})$ or even to $(\class{FG},\class{(U)CQ})$
where $\class{UCQ}$ denotes unions of CQs and $\class{FG}$ denotes
frontier-guarded TGDs. Note, however, that enumeration in \cdlin
of answers to UCQs is not fully understood even in the case without
ontologies
\cite{carmeli-enum-ucqs}. % Also note that $(\class{FG},\class{CQ})$ is no
% less expressive than $(\class{FG},\class{UCQ})$ and thus admitting
% frontier-guarded TGDs implies admitting UCQs.
{%\color{blue}
Another interesting
question is whether the enumeration problems placed in \dlc in
the current paper actually fall within \cdlin, that is, whether the
use of a polynomial amount of memory in the enumeration phase can be
avoided.
}
\medskip 

\noindent
{\bf Acknowledgement.} We acknowledge support by the
DFG project LU 1417/3-1 `QTEC'.

%\newpage
%\cleardoublepage

\bibliographystyle{plainurl}
\bibliography{enum}

\ifbool{arxive}{
\cleardoublepage
}{
  \newpage
}

\appendix   

\ifbool{arxive}{
\section{Additional Preliminaries}

\subsection{The RAM model}

%  \section{The RAM model}

As our computational model, we assume
RAMs under the uniform cost model \cite{DBLP:journals/jcss/CookR73},
following \cite{Grandjean-RAM} in the concrete formalization. Such a
RAM has a one-way read-only input tape and a write-only output tape,
as well as an unbounded number of registers that store non-negative
integers of $O(\log n)$ bits, $n$ the input size; this is called a
DRAM %DLINRAM
in \cite{Grandjean-RAM}, used there to define the complexity class
DLINEAR. Adding, subtracting, and comparing the values
of two registers as well as bit shift takes time~$O(1)$. 
  This model is a standard assumption in the context of
  enumerating the answers to queries
  \cite{segoufin-enum, bagan-enum-cdlin,carmeli-enum-func,
    berkholz-enum-tutorial}, although sometimes smaller registers
  are assumed that can only hold integers up to $O(n / \log(n))$.
  An input database $D$ is given as a word on the input tape.  Since
  we are interested in data complexity, there are only $O(1)$ many
  relation symbols whose arity is $O(1)$.  We assume that constants in
  $\mn{adom}(D)$ are represented in binary, which requires at most
  $\mn{log}(k)$ bits, $k=|\mn{adom}(D)|$. We can thus store a constant in a
  single register, and the same is true for facts in $D$.
  
  We shall often be interested in lists of constants from
  $\mn{adom}(D)$ that are of length $O(1)$, let us call this a
  \emph{short list}. With every short list, we can associate a unique
  memory address (register) that can be computed from the short list
  in $O(1)$ time using bit shifting and addition. This means that we can implement
  lookup tables indexed by such lists that can be accessed and updated
  in $O(1)$ time. % , essentially assuming perfect
  % hashing without collisions.

  Another crucial property of this model is that sorting is possible
  in linear time \cite{Grandjean-RAM}. In fact, we shall be interested
  in sets of short lists equipped with a strict weak order. We
  summarize
  the approach from \cite{Grandjean-RAM}.
  % any list of short
  % lists can be sorted lexicographically in linear time,  treating the
  % constants in $\mn{adom}(D)$ as an ordered set of alphabet
  % symbols. The latter is justified by the fact that each such constant
  % can be stored in a single register and thus comparisons are possible
  % in $O(1)$ time. %  (despite the fact that the number of `alphabet
  % symbols' is not $O(1)$).
  To sort a list $L = l_1;l_2; \dots; l_m$ of $m$ short lists
  one first observes the following:
  \begin{itemize}
      \item we can sort a list of short words, by which we mean words of length $O(\log ||L||)$,
      in time $O(||L||)$ using \emph{counting sort}; indeed, there is no more than 
      $O(||L||)$ words of such length;
      \item we can use standard sorting algorithms to sort a list of long words, i.e.~words of length $\Omega(\log ||L||)$, using $O(n \log n)$ operations, where $n \in O(\frac{||L||}{\log ||L||})$ is the number of elements of length $\Omega(\log ||L||)$.
  \end{itemize}

  To sort a list, it is thus enough to divide it into two disjoint lists: a list of short words and a list of long words.
  Then we sort those lists independently and join them into a single sorted list.
  Since the division and the join can be easily done in linear time,
  sorting can be performed in time $O(||L||) + O(n \log n)$.
  Since $n \in O(\frac{||L||}{\log ||L||})$, this gives the overall running time $O(||L||)$.
%  We can use standard sorting algorithms to
%  sort a list of $m$ short lists with $O(m \log m)$ operations,
%  thus in time $O(m \log m)$ due to the uniform measure.
  % , $m$ the number of short lists in the
  % list.
%  Note that in our case $m \in k^{O(1)}$ with $k=|\mn{adom}(D)|$. Since the
%  representation of $D$ as a word on the input tape has length at
%  least $k \log k$ (with the $\log k$ factor contributed by the
%  representation of constants), this yields linear time
%  sorting.
  We also recall that sorting on a RAM in linear time
  is possible even under the less liberal \emph{logarithmic cost measure} and
  when registers can only hold integers up to $O(n /
  \log(n))$ \cite{Grandjean-RAM}.
  %
  % Some remarks on sorting on RAMs:
  % - under the uniform cost model, it seems to be trivial (see above)
  % - under the logarithmic cost model, it is not, no matter whether
  % one allows polynomial memory or linear memory
  % - 

  \subsection{More on the Chase}
\label{sect:thechase}

We provide some observations regarding chase procedure.
The following is the central property of the chase.
\begin{lemma}\label{pro:chase}
  % ~\\[-4mm]
  % \begin{enumerate}
%
  % \item 
  Let $\Omc$ be a finite set of TGDs and $I$ an instance. Then
  for every model $J$ of $\Omc$
  with $I \subseteq J$, there is a homomorphism $h$ from
  $\mn{ch}_\Omc(I)$ to $J$ that is the identity on $\mn{adom}(I)$. 
%
% \item $Q(D) = q(\mn{ch}_S(D))$ for every OMQ
%%   $Q = (\Omc,\Sbf,q) \in (\class{TGD},\class{UCQ})$.
  %
%   \end{enumerate}
\end{lemma}
We next establish a technical lemma about the chase that may be viewed
as a locality property. Let $I$
be an instance and $\Omc$ a set of TGDs.  With every guarded set $S$
of~$I$, we associate a subinstance
$\mn{ch}_\Omc(I)|^\downarrow_S \subseteq \mn{ch}_\Omc(I)$ such that
every fact in $\mn{ch}_\Omc(I)$ that contains at least one null is
contained in exactly one such subinstance. We first identify with
every fact $R(\bar c) \in \mn{ch}_\Omc(I)$ that contains at least one
null a unique `source' fact
$\mn{source}(R(\bar c)) \in I$. % Recall
% that $\mn{ch}_S(I)$ is constructed by a sequence of chase steps.
% For all facts \mbox{$R(\bar c) \in I$}, $\mn{guard}(R(\bar c)) = R(\bar
% c)$.
Assume that $R(\bar c)$ was introduced by a chase step that applies
a TGD $T$ at a tuple $(\bar d, \bar d')$, and let $R'$ be the relation symbol
in the guard atom in $\mn{body}(T)$. Then we set
$\mn{source}(R(\bar c))=R'(\bar d, \bar d')$ if
$\bar d \cup \bar d' \subseteq \mn{adom}(I)$ and
$\mn{source}(R(\bar c))=\mn{source}(R'(\bar d, \bar d'))$ otherwise. For
any guarded set $S$ of $I$, we now define
$\mn{ch}_\Omc(I)|^\downarrow_{S}$ 
to contain those facts $R(\bar c) \in \mn{ch}_\Omc(I)$ such that
\begin{enumerate}

\item $\bar c \subseteq S$ or
 
\item $\bar c$ contains at least one null and
  $\mn{source}(R(\bar c))=S$.

\end{enumerate}
We shall actually consider such subinterpretations not only of the
final result $\mn{ch}_\Omc(I)$ of the chase, but also of the instances
constructed as part of a chase sequence $I_0,I_1,\dots$ for $I$ with
$\Omc$. In fact, we can define $I_i|^\downarrow_S$ in exact analogy with
$\mn{ch}_\Omc(I)|^\downarrow_{S}$, for all $i \geq 0$.  We next observe
that all facts in a subinstance $I_i|^\downarrow_{S}$ of $I_i$ can be
obtained by starting from the (very small) subinstance $I_i|_{S}$ and
then chasing with $\Omc$. % It is straightforward to prove the
% following by induction on $i$. 
%
\begin{lemma}
  \label{lem:treesarenice}
  Let $I$ be an instance, \Omc a set of guarded TGDs, $I_0,I_1,\dots$ a chase
  sequence of $I$ with \Omc, $i \geq 0$, and $S$ a guarded set in $I$. Then there is a homomorphism from $I_i|^\downarrow_{S}$
  to $\mn{ch}_\Omc(I_i|_{S})$ that is the identity on all constants in
  $S$. 
  % Let $I$ be an instance, \Smc a set of TGDs, $I_0,I_1,\dots$ a chase
  % sequence of $I$ with \Smc, and $S$ a guarded set in $I$. Then for
  % every $i \geq 0$ there is a homomorphism from $I_i|^\downarrow_{S}$
  % to $\mn{ch}_S(I_i|_{S})$ that is the identity on all constants in
  % $S$. Moreover, $h_i(c)=h_{i+1}(c)$ for all $c$ with $h_i(c)$
  % defined.
\end{lemma}
\begin{proof} %[Proof (sketch).]
  The proof is by induction on $i$. The induction start holds as
  $I_0|^\downarrow_S=I_0|_S \subseteq \mn{ch}_\Omc(I|_S)$. For the
  induction step, assume that $I_{i+1}$ was obtained from $I_i$ by
  applying a TGD
  $T= \phi(\bar x,\bar y) \rightarrow \exists \bar z \, \psi(\bar
  x,\bar z)$  at a tuple $(\bar c, \bar c')$. Let
  $R$ be the relation symbol used in the guard atom of
  $\phi$.  Then $R(\bar c,\bar c') \in
  I_i$. Let $S$ be a guarded set in~$I$. If   $I_{i+1}|^\downarrow_{S}
  =I_{i}|^\downarrow_{S}$, it suffices to use the induction
  hypothesis. Thus assume   $I_{i+1}|^\downarrow_{S}\neq I_{i}|^\downarrow_{S}$.
  
  First assume that $S \neq \bar c \cup \bar c'$. Then all facts in
  $I_{i+1}|^\downarrow_{S}\setminus I_{i}|^\downarrow_{S}$ contain
  only constants in $\bar c$, but no nulls.  By induction
  hypothesis, there is a homomorphism $h_i$ from $I_i|^\downarrow_{S}$
  to $\mn{ch}_\Omc(I_i|_{S})$ that is the identity on all constants in
  $S$. Clearly, $h_i$ is also a homomorphism from
  $I_{i+1}|^\downarrow_{S}$ to $\mn{ch}_\Omc(I_{i+1}|_{S})$.

  Now assume that $S=\bar c \cup \bar c'$. Then
  $I_{i+1}|^\downarrow_{S}\neq I_{i}|^\downarrow_{S} = \psi(\bar c,
  \bar c'')$ where $\bar c''$ consists of constants that do not occur
  in $I_i$.  By induction hypothesis, there is a homomorphism $h_i$
  from $I_i|^\downarrow_{S}$ to $\mn{ch}_\Omc(I_i|_{S})$ that is the
  identity on all constants in $S$. Applicability of $T$ at
  $(\bar c, \bar c')$ implies $\phi(\bar c,\bar c') \subseteq I_i$ and
  thus
  $\phi(h_i(\bar c),h_i(\bar c')) \subseteq \mn{ch}_\Omc(I_i|_{S})$.
  It follows $T$ has been applied at $(h(\bar c),h(\bar c'))$ in (any
  fair chase sequence that produces) $\mn{ch}_\Omc(I_i|_{S})$.  As a
  consequence, there are constants $\bar d$ such that $\psi(h(\bar
  c),\bar d) \subseteq \mn{ch}_\Omc(I_i|_{S})$. We extend $h_i$ to $h_{i+1}$
  so that $h_{i+1}(\bar c'')=\bar d$. Clearly, $h_{i+1}$ is a homomorphism from
  $I_{i+1}|^\downarrow_{S}$ to $\mn{ch}_\Omc(I_{i+1}|_{S})$.
\end{proof}

\subsection{\ELI and Simulations}

We introduce some preliminaries that are specific to \ELI.
Recall that, in \ELI, relation symbols can only have arity~1
or~2. A CQ $q$ over such a schema gives rise to an   undirected graph \label{prelim:tree-shape}
$$G_q^{\mn{var}}=(\mn{var}(q),\{\{x,y\} \mid R(x,y) \in q \text{ with } x,y \in
\mn{var}(q) \text{ and } x \neq y \}).$$
Note that in contrast to variables, constants in $q$ do not serve as
nodes in $G_q^{\mn{var}}$.  It is easy to see that $q$ is acyclic if
$G_q^{\mn{var}}$ is a disjoint union of trees. Of course, this admits
  reflexive loops and multi-edges in $q$.

  We now introduce the notion of a simulation, which is closely linked
  to the expressive power of \ELI.  Let \Sbf be a schema that only
  contains relations of arity one and two, and let $I$ and $J$ be
  \Sbf-instances. A \emph{simulation} from $I$ to $J$ is a relation
  $S \subseteq \mn{adom}(I) \times \mn{adom}(J)$ such that
\begin{enumerate}

\item $A(c) \in I$ and $(c,c') \in S$ implies $A(c') \in J$,

\item $R(c_1,c_2) \in I$  and $(c_1,c'_1) \in S$ implies that
  there is a $c'_2 \in \mn{adom}(J)$ such that $R(c'_1,c'_2) \in J$
  and $(c_2,c'_2) \in S$, and

\item $R(c_2,c_1) \in I$  and $(c_1,c'_1) \in S$ implies that
there is a $c'_2 \in \mn{adom}(J)$ such that $R(c'_2,c'_1) \in J$
and $(c_2,c'_2) \in S$.
  
\end{enumerate}
If there is a simulation for $I$ to $J$ such that $(c,c') \in S$, then
we write $(I,c) \preceq (J,c')$.

A unary CQ $q(x)$ is an \emph{ELIQ} (which stands for \emph{\ELI
  query}) if it contains no constants and the undirected graph
$G^{\mn{var}}_q$ is a disjoint unions of trees and $q$ contains no
self-loops and multi-edges, where the latter means that for any
$x,y \in \mn{var}(q)$, $q$ contains at most a single atom that
mentions both $x$ and $y$.\footnote{In the literature, an ELIQ is
  often defined as a single tree, rather than a disjoint union
  thereof. We work with the more general definition as this turns our
  to be more convenient for our purposes.} We use $\class{ELIQ}$ to denote the class of
all ELIQs.
\begin{lemma}
  \label{lem:sim}
  Let $Q(x)=(\Omc,\Sbf,q)\in(\class{ELI},\class{ELIQ})$,
  $D_1,D_2$ \Sbf-data\-bases, and $c_i \in \mn{adom}(D_i)$ for
  $i \in \{1,2\}$. Then $(D_1,c_1) \preceq (D_2,c_2)$ and
  $c_1 \in Q(D_1)$ implies $c_2 \in Q(D_2)$.
\end{lemma}
We start with recalling the following well-known fact, proved e.g.\ as
Theorem~10 in \cite{Lutz-Wolter-JSC-10}.
\begin{lemma}
    \label{lem:sim-cq}
    Let $q(x)$ be an ELIQ, $D_1$, $D_2$ $\Sbf$\=/databases, and $c_i \in
    \mn{adom}(D_i)$ for $i \in \{1,2\}$. If $(D_1,c_1) \preceq
    (D_2,c_2)$ and $c_1 \in q(D_1)$ then $c_2 \in q(D_2)$.
\end{lemma}
%
% \begin{proof}
%   Since $c_1 \in q(D_1)$ then there is a homomorphism $h_1$ from $q$
%   to $D_1$ such that $h_1(x) = c_1$. All we need is to define a
%   homomorphism $h_2$ from $q$ to $D_2$ such that $h_2(x) = c_2$.
    
%   Since $q$ is an ELIQ, $q$ is connected and we can use the simulation
%   relation $S$ to define $h_2(y)$ for every $y \in \text{var}(q)$.
%   The idea is that we follow the simulation relation in both $D_1$ and
%   $D_2$ starting from $(c_1,c_2)$ to define the homomorphism $h_2$.
%   The precise construction is as follows.

%     We first define a child order on variables of $q$. We can do that since $q$ is an ELIQ and, thus, it has a~tree\=/like structure.
%     We define this order by stating that $x$ is the root of the tree.
    
%     Now we define the mapping inductively. Let $y$ be a variable for which $h_2(y)$ is not yet defined, but $h_2(z)$ is defined
%     for their unique parent variable $z$. Let $R(z,y)$ (resp. $R(y,z)$) be the unique atom in $q$ containing $y,z$.
%     We define $h_2(y) = c$ for some $c \in \dom(D_2)$ such that $R(h_1(z), c)$ (resp. $R(c, h_1(z))$) and $(h_1(y),c) \in S$.
%     The constant $c$ exists by Point 2. of the definition of a simulation.
    
%     It is straightforward to check that this inductive procedure terminates and defines a homomorphism $h_2$ from $q$ to $D_2$
%     with $h_2(x) = c_2$.
% \end{proof}
%
\begin{proof}[Proof of Lemma~\ref{lem:sim}]
  Let $D \supseteq D_2$ be a model of \Omc. We have to show that
  $c_2 \in q(D)$. The following claim shall be essential.
  \\[2mm]
  \emph{Claim.} $(\mn{ch}_\Omc(D_1),c_1) \preceq (\mn{ch}_\Omc(D_2),c_2)$.
  \\[2mm]
  The chase constructs a sequence
  $D_1 = I_0 \subseteq I_1 \subseteq \cdots$ such that
  $\mn{ch}_\Omc(D_1)=\bigcup_{i \geq 1} I_i$.  We construct a sequence
  of relations $S = S_0 \subseteq S_1 \subseteq \cdots \subseteq S_k$
  such that
  \begin{itemize}
      \item [$(\dagger)$] %for $0 \leq i \leq k$, $S_i \subseteq
        % \dom(I_i) \times \dom(J_i)$
        $S_i$
        is a simulation from $I_i$ to $\mn{ch}_\Omc(D_2)$.
  \end{itemize}
  Relation $S_0$ is already defined.  For the inductive step, let us
  assume that we have already defined $S_i$.
    
  Assume that $I_{i+1}$ was obtained from $I_i$ by applying the TGD
  $\varphi(x,\bar y) \rightarrow \psi(x,\bar z) \in \Omc$ in $I_i$ at
  $(c,\bar c')$. Since $\Omc$ is formulated in \ELI, $\varphi$ and
  $\psi$ are ELIQ, $c \in \varphi(I_i)$, and $I_{i+1}$ was obtained
  from $I_i$ by adding a copy $D'_\psi$ of $D_\psi$ using the
  constant $c$ in place of the answer variable of ELIQ $\psi$ and the
  fresh constants from $\bar c'$ in place of the quantified variables.
  Let us assume that for each quantified variable $z$ in $\psi$,
  the corresponding constant in $\bar c'$ is $c_z$.

  To construct $S_{i+1}$, start with setting $S_{i+1} = S_i$. Then
  consider all $(c,d) \in S_i$. Since $\varphi$ is an ELIQ,
  $c \in \varphi(I_i)$ and Lemma~\ref{lem:sim-cq} yield
  $d \in \varphi(\mn{ch}_\Omc(D_2))$. It follows that the TGD
  $\varphi(x,\bar y) \rightarrow \psi(x,\bar z)$ is applicable in
  $\mn{ch}_\Omc(D_2)$ at $d$ and was indeed applied during the
  construction of $\mn{ch}_\Omc(D_2)$. Consequently,
  $d \in \psi(\mn{ch}_\Omc(D_2))$ and thus we find a homomorphism
  $h$ from $\psi(x)$ to $\mn{ch}_\Omc(D_2)$ with $h(x)=d$.
  Extend $S_{i+1}$ with $(c_z,h(z))$ for each quantified variable $z$
  in $\psi$. It is easy to verify that $S_{i+1}$ satisfies
  $(\dagger)$, as required. This finishes the proof of the claim.

  \medskip
  
  Now back to the proof of Lemma~\ref{lem:sim}.  Since
  $c_1 \in Q(D_1)$, there exists a homomorphism $h_1$ from $D_q(x)$ to
  $\mn{ch}_{\Omc}(D_1)$ such that $h(x)=c_1$. Clearly, a homomorphism
  is also a simulation. Composing this simulation with a simulation
  $S$ from $\mn{ch}_\Omc(D_1)$ to $\mn{ch}_\Omc(D_2)$ with
  $(c_1,c_2) \in S$, whose existence is guaranteed by the claim, we
  obtain a simulation $S'$ from $D_q$ to $\mn{ch}_\Omc(D_2)$
  with $(x,c_2) \in S'$. Since $x \in q(D_q)$, Lemma~\ref{lem:sim-cq}
  yields $c_2 \in q(\mn{ch}_\Omc(D_2))$ and thus $c_2 \in Q(D_2)$,
  as required.
\end{proof}

\section{Proofs for Section~\ref{sect:prelims}}

\propcompletefirst*
\begin{proof}
  Let $A_{\mn{mpa}}$ be the algorithm for enumerating minimal partial
  answers to $Q$ in \dlc and let $A_\csf$ be the algorithm for
  enumerating complete answers to $Q$ in \dlc. To obtain the
  desired enumeration algorithm that produces the complete answers
  first, we run $A_\csf$ and $A_{\mn{mpa}}$ in parallel, starting with
  the preprocessing phase of both algorithms. In the enumeration
  phase, $A_{\mn{mpa}}$ clearly outputs at least as many answers as
  $A_{\csf}$. We run the enumeration phase of both algorithms in
  parallel. As long as $A_\csf$ still produces answers, we let both
  $A_\csf$ and $A_{\mn{mpa}}$ produce their next answer, but output
  only the answer of $A_\csf$. The answer of $A_{\mn{mpa}}$ is
  diregarded if it is complete, and stored in a linked list otherwise.
  Once that $A_{\csf}$ runs out of answers, we switch to only
  continueing the enumeration phase of $A_{\mn{mpa}}$. If
  $A_{\mn{mpa}}$ produces an answer with a wildcard, we output it.  If
  $A_{\mn{mpa}}$ produces a complete answer, we output one of the
  stored answers instead. Clearly, the number of complete answers to
  be replaced by a stored answer coincides exactly with the number of
  answers stored.
\end{proof}

\lemtwosemantics*
\begin{proof}
  We only consider $Q(D)^{\ast}$ and $q(\mn{ch}_\Omc(D))^{\ast}_\Nbf$,
  the case of $\cdot^\Wmc$ is similar. First assume that $\bar c \in
  Q(D)^\ast$. Then for every model $I$ of \Omc with $I \supseteq D$,
  there is a $\bar c' \in q(I)$ with $\bar c' \preceq \bar c$. In
  particular, this is true for $I=\mn{ch}_\Omc(D)$.  We observe the
  following:
  \begin{itemize}

  \item If $\bar c'$ has a constant from $N$ in some position, then
    $\bar c$ has `$\ast$' in the same position.

    This is because $\bar c' \preceq \bar c$ and $\bar c$ may only
    contain elements from $\mn{adom}(D) \cup \{ \ast \}$.

  \item If $\bar c'$ has a constant $c \in \mn{adom}(D)$ in some
    position, then $\bar c$ has $c$ in the same position.

    From $\bar c' \preceq \bar c$, it follows that $\bar c$ has $c$ or
    `$\ast$' in the same position. Assume that there is at least one
    position $i$ where $\bar c'$ has a constant $c \in \mn{adom}(D)$
    and $\bar c$ has `$\ast$'. Let $\bar c''$ be obtained from $\bar
    c$ by replacing `$\ast$' in position $i$ with $c$. Then $\bar c''
    \prec \bar c$ and it follows from Lemma~\ref{pro:chase} that $\bar
    c''$ is a partial answer to $Q$ on $D$, contradicting the fact
    that $\bar c$ is a minimal partial answer to $Q$ on $D$.

  \end{itemize}
  It follows that $\bar c = {(\bar c')}^\ast_\Nbf \in
  q(\mn{ch}_\Omc(D))^{\ast}_\Nbf$. 

  \smallskip

  Conversely, assume that $\bar c \in q(\mn{ch}_\Omc(D))^{\ast}_\Nbf$. By
  Lemma~\ref{pro:chase}, for every $I$ of \Omc with $I \supseteq D$,
  there is a $\bar c' \in q(I)$ with $\bar c' \preceq \bar c$. Thus,
  $\bar c$ is a partial answer to $Q$ on $D$. Assume to the contrary
  of what we want to show that it is not a minimal partial answer, that
  is, there is a $\bar c' \in Q(D)^\ast$ with $\bar c' \prec \bar
  c$. Then there is a $\bar c'' \in q(\mn{ch}_\Omc(D))$ with $\bar c''
  \preceq \bar c'$. We can show as above that $\bar c' = (\bar
  c'')^\ast_\Nbf \in q(\mn{ch}_\Omc(D))^{\ast}_\Nbf$. But now $\bar c' \prec
  \bar c$ contradicts $\bar c \in q(\mn{ch}_\Omc(D))^{\ast}_\Nbf$.
  
\end{proof}
%
% We remark that there is a subtlety here. In contrast to
% Lemma~\ref{lem:twosemantics}, it is not necessarily the case that the
% (not necessarily minimal) partial answers to $Q(\bar x)=(\Omc,\Sbf,q) $
% on an \Sbf-database $D$ coincide with the partial answers to $q$ on
% $\mn{ch}_\Omc(D)$ w.r.t.\ $N=\mn{adom}(\mn{ch}_\Omc(D)) \setminus
% \mn{adom}(D)$. In fact, $(\ast,\cdots,\ast)$ is a partial answer to
% $Q(\bar x)=(\Omc,\Sbf,q) $ on $D$ if there is a partial answer at all,
% but this needs not be the case for the partial answers to $q$ on
% $\mn{ch}_\Omc(D)$ w.r.t.\ $N$. As a trivial example, consider
% $Q(x)=(\emptyset,\{A\},A(x))$ and the database $\{ A(c) \}$.

\section{Proofs for Section~\ref{sect:singletesting}}

\subsection{Proof of Theorem~\ref{thm:singlelin}}
\label{app:sect:singletest}

% SHORT
% The
% query-directed chase also provides an important foundation for
% enumeration and all-testing later on.

% A set $S \subseteq \mn{adom}(D)$ is a \emph{guarded
%   set in $D$} if some fact in $D$ contains all constants in $S$, and
% possibly additional constants.
% Let $m$ be the maximum arity of
% relation symbols in $\mn{sch}(\Omc)$. 

%{\color{red}properly take care of constants in CQs / adapt notation}

% We first define the query-directed chase $\mn{ch}^q_\Omc(D)$ in
% full detail.
% We use $\mn{cl}(Q)$ to denote the set of all CQs that are connected
% and use only relation symbols that occur in \Omc, no constants, and
% only variables from a fixed set $V$ whose cardinality is the maximum
% of $|\mn{var}(q)|$ and the arities of relation symbols in \Omc. Note
% that the CQs in $\mn{cl}(Q)$ can have any arity, including arity zero,
% and that the cardinality of $\mn{cl}(Q)$ is independent of $D$. Then
% $\mn{ch}^q_\Omc(D)$ is obtained from $D$ by adding, for every CQ
% $p(\bar y) \in \mn{cl}(Q)$ and every
% $\bar c \in \mn{adom}(D)^{|\bar y|}$ such that
% $D \cup \Omc \models q(\bar c)$ and the constants in $\bar c$
% constitute a guarded set in $D$, a copy of $D_p$ that uses the
% constants in $\bar c$ in place of the answer variables $\bar y$ of $p$
% and only fresh constants otherwise. The following lemma states the
% main property of $\mn{ch}^q_\Omc(D)$.

Our first aim is to prove Lemma~\ref{prop:restrictedchaseworks}.
We start with observing the following.
\begin{lemma}
\label{lem:pseudochase} 
Let $Q(\bar x)=(\Omc,\Sbf,q) \in (\class{G},\class{CQ})$ and let $D$
be an \Sbf-database. Further let $q'(\bar x')$ be a CQ with
$|\mn{var}(q')| \leq |\mn{var}(q)|$ and $\bar c \in \mn{adom}(D)^{|\bar x'|}$. Then
$D \cup \Omc \models q'(\bar c)$ iff
$\bar c \in q'(\mn{ch}^q_\Omc(D))$.
\end{lemma}
\begin{proof}
  % Let $Q(\bar x)=(\Omc,\Sbf,q)$, $D$, $q'$, and $\bar c$ be as in
  % Lemma~\ref{lem:pseudochase}. We have to show that $D \cup \Omc
  % \models q'(\bar c)$ iff $\bar c \in q'(\mn{ch}^q_\Omc(D))$.
  The `if' direction is an immediate consequence of the definition of
  $\mn{ch}^q_\Omc(D)$.  For the `only if' direction, assume that $D
  \cup \Omc \models q'(\bar c)$. Then there is a homomorphism $h$ from
  $q'$ to $\mn{ch}_\Omc(D)$ such that $h(\bar x')=\bar c$. For every
  guarded set $S$ in $D$, let $A_S$ denote the set of atoms $R(\bar
  t)$ in $q'$ such that $h(\bar t)$ contains at least one null and
  $\mn{source}(R(h(\bar t)))=S$ (as defined in
  Appendix~\ref{sect:thechase}). Let \Smf be the set of guarded sets
  $S$ in $D$ with $A_S \neq \emptyset$. For each $S \in \Smf$, let
  $q'_S(x'_S)$ denote the CQ obtained by first restricting $q'$ to the
  atoms in $A_S$ and then making a variable $x$ an answer variable
  if $h(x) \in \mn{adom}(D)$ and a quantified variable otherwise. 

  To prove that $D \cup \Omc \models q'(\bar c)$, it suffices to show
  the following:
  \begin{enumerate}

  \item $h(\bar t) \in q'(\mn{ch}^q_\Omc(D))$ for every atom $R(\bar t)$
    in $q'$ such that $h(\bar t)$ contains no null;

  \item $h(\bar x'_S) \in q'_S(\mn{ch}^q_\Omc(D))$ for every $S \in \Smf$.

  \end{enumerate}
  In fact, composing homomorphism $h$ with the homomorphisms into
  $\mn{ch}^q_\Omc(D)$ that witness~(2), one can
  straightforwardly construct a homomorphism $h'$ from
  $q'$ to $\mn{ch}^q_\Omc(D)$ such that $h'(q')=\bar c$.

  For Point~(1), let $R(\bar t)$ be an atom in $q'$ such that $h(\bar
  t)$ contains no null. Then $D \cup \Omc \models p(\bar c')$ for
  $p=R(\bar t)$ and $\bar c'=h(\bar t)$, and $\bar c'$ is a guarded set
  in $D$. Thus, $R(h(\bar t))$ has been added to $\mn{ch}^q_\Omc(D)$
  during its construction.

  For Point~(2), let $S \in \Smf$.  Then $D \cup \Omc \models p(\bar
  c')$ for $p=q'_S(\bar x'_S)$ and $\bar c'=h(\bar x'_S)$, and $\bar c'$ is a
  guarded set in $D$. Thus, a copy of $D_{q'_S}$ that uses the
  constants
  in $\bar c'$ in place of the answer variable $\bar x'_S$ of $q'_S$ has 
  been added to $\mn{ch}^q_\Omc(D)$
  during its construction.
\end{proof}
%
% The next lemma implies that $\mn{ch}_\Omc(D)$ can be replaced with
% $\mn{ch}^q_\Omc(D)$ for the purposes of this paper. 
%
\proprestrictedchaseworks*

\begin{proof}
  Let $Q(\bar x) =(\Omc,\Sbf,q)$, $D$ and $N$ be as in
  Lemma~\ref{prop:restrictedchaseworks}.  It is an immediate
  consequence of Lemma~\ref{lem:pseudochase}, instantiated with
  $q'=q$, that $Q(D)=q(\mn{ch}^q_\Omc(D)) \cap \mn{adom}(D)^{|\bar
    x|}$.  The cases $Q(D)^{\ast}_\Nbf= q(\mn{ch}^q_\Omc(D))^{\ast}$ and
  $Q(D)^{\Wmc}_\Nbf= q(\mn{ch}^q_\Omc(D))^{\Wmc}$ are very similar, we
  concentrate on the latter. It clearly suffices to prove that a tuple
  $\bar a^\Wmc$ is a (not necessarily minimal) partial answer with
  multi-wildcards to $q$ on $\mn{ch}^q_\Omc(D)$ if and
  only if $\bar a^\Wmc$ is a partial answer with multi-wildcards to
  $Q$ on~$D$.

  Let $\bar x = x_1,\dots,x_n$, let $\bar a^\Wmc = a_1,\dots,a_n$, and
  let the wildcards from \Wmc that occur in $\bar a^\Wmc$ be
  $\ast_1,\dots,\ast_\ell$.  Consider the CQ $q'$ obtained from $q$ in
  the following way:
  \begin{itemize}

  \item introduce fresh quantified variable $z_1,\dots,z_\ell$;

  \item if $a_i=\ast_j$, then replace in $q'$ the answer variable
    $x_i$ with quantified variable $z_j$.
    
  \end{itemize}
  Further let $\bar c$ be obtained from $\bar a^\Wmc$ by removing all
  wildcards. It is easy to see that $\bar c \in q'(\mn{ch}^q_\Omc(D))$
  iff $\bar a^\Wmc$ is a partial answer with multi-wildcards to $q$ on
  $\mn{ch}^q_\Omc(D)$ and that $D\cup \Omc \models q'(c)$
  iff $\bar a^\Wmc$ is a partial answer with multi-wildcards to $Q$ on
  $D$, both by definition of partial answers with multi-wildcards and
  by construction of $q'$. It thus remains to invoke
  Lemma~\ref{lem:pseudochase}.
\end{proof}
We next establish Proposition~\ref{prop:chaseinlineartime}.

\propchaseinlineartime*

To prove Proposition~\ref{prop:chaseinlineartime}, we make use of the
fact that minimal models for propositional Horn formulas can be
computed in linear time~\cite{dowling-gallier-horn}. More precisely,
we derive a satisfiable propositional Horn formula $\theta$ from $D$
and $Q$, compute a minimal model of $\vp$ in linear time, and then
read off $\mn{ch}^q_\Omc(D)$ from that model.

Let $\mn{sch}(\Omc)$ denote the set of all relation symbols that occur
in $\Omc$.  We introduce a propositional variable $x_{p(\bar c)}$ for
every CQ $p(\bar y) \in \mn{cl}(Q)$ and every $\bar c \in
\mn{adom}(D)^{|\bar y|}$ such that the constants in $\bar c$
constitute a guarded set in $D$. Observe that the cardinality of
$\mn{cl}(Q)$ is bounded by $2^{2^{O(||Q||^2)}}$. Moreover, since
$Q$ is fixed and \Sbf only contains relation symbols that occur in
$\Omc$ or on $q$, the arity of relation symbols in $D$ is bounded by
the constant $||Q||$. Consequently, the number of guarded sets in $D$
is bounded by $2^{||Q||} \cdot ||D||$ and the number of variables
$x_{p(\bar c)}$ is bounded by $2^{2^{O(||Q||^2)}} \cdot ||D||$.

Consider the Horn formula $\theta$ that consists of the following conjuncts:
\begin{enumerate}

\item $x_{R(\bar c)}$ for every $R(\bar c) \in D$;

\item
  $\bigwedge_{S(\bar d) \in D'} x_{S(\bar d)} \rightarrow x_{p(\bar
    c)}$ for every $\mn{sch}(\Omc)$-database $D'$, every CQ
  $p(\bar y) \in \mn{cl}(Q)$, and every $\bar c \in
  \mn{adom}(D')^{|\bar y|}$ such that $D' \cup \Omc \models p(\bar c)$
  and
  $\mn{adom}(D')$ is a guarded set $S$ in $D$.

\end{enumerate}
The size of $\theta$ is bounded by $2^{2^{O(||Q||^2)}} \cdot
||D||$ and $\theta$ can be constructed in time
$2^{2^{O(||Q||^2)}} \cdot ||D||$. This again depends on the
arity of relation symbols being (implicitly) bounded by a constant.

Since $\theta$ contains no negative literals, it is clearly
satisfiable and thus has a unique minimal model. Let $V$ be the truth
assignment that represents this minimal model. We construct a database
$D_\theta$ as follows. Start with~$D$. Then iterate over all
propositional variables~$x_{p(\bar c)}$. If $V(x_{p(\bar c)})=1$, then
take a copy $D'_p$ of $D_p$ that uses the constants in $\bar c$ in
place of the answer variables $\bar y$ of $p$ and only fresh constants
otherwise, and take the union of the database constructed so far and
$D'_p$. It is clear that the construction of $D_\theta$ is in time
$2^{2^{O(||Q||^2)}} \cdot ||D||$. Thus, the following implies
Proposition~\ref{lem:pseudochase}.
\begin{lemma}
\label{lem:thetaispseudochase}
  % Let $q'(\bar x')$ be a CQ with
  % $|\mn{var}(q')| \leq |\mn{var}(q)|$ and 
  % $\bar c \in \mn{adom}(D)^{|\bar x'|}$. Then $D \cup \Omc \models
  % q'(\bar c)$ iff $\bar c \in q'(\mn{ch}^q_\Omc(D))$.
$D_\theta = \mn{ch}^q_\Omc(D)$.
\end{lemma}
\begin{proof}
  First assume that for some CQ $p(\bar y) \in \mn{cl}(Q)$ and tuple
  $\bar c \in \mn{adom}(D)^{|q|}$ such that the constants in $\bar c$
  constitute a guarded set in $D$, a copy $D'_p$ of $D_p$ has been
  added to $D_\theta$ during the construction of this
  database. Then $V(x_{p(\bar c)})=1$. We argue that this implies $D
  \cup \Omc \models p(\bar c)$. By construction of
  $\mn{ch}^q_\Omc(D)$, this means that $D'_p$ has also been added
  during the construction of this database.
  
  Recall that the minimal model of a satisfiable propositional Horn
  formula can be computed (though not in linear time) through a
  straightforward algorithm that generates a least fixed
  point. Applied to $\theta$, the algorithm starts with the set of
  variables $V_0$ from Point~1 of the definition of $\theta$ and
  then exhaustively applies the rules from Point~2, which yields a
  sequence of variable sets $V_0 \subseteq V_1 \subseteq \cdots
  \subseteq V_k$ whose limit $V_k$ is the minimal model $V$
  of~$\theta$.  It suffices to prove that $x_{p'(\bar c')} \in V_i$,
  with $0 \leq i \leq k$, implies $D \cup \Omc \models p'(\bar c')$.
  This is straightforward be induction on $i$, using the definition
  of the Horn formula $\theta$.

  \smallskip

  Conversely, assume that for some CQ $p(\bar y) \in \mn{cl}(Q)$ and
  tuple $\bar c \in \mn{adom}(D)^{|q|}$ such that the constants in
  $\bar c$ constitute a guarded set in $D$, a copy $D'_p$ of $D_p$ has
  been added to $\mn{ch}^q_\Omc(D)$ during its construction. Then $D
  \cup \Omc \models p(\bar c)$. %  and thus $\bar c \in
  % p(\mn{ch}_\Omc(D))$.

  Let $D_0,D_1,\dots$ be a chase sequence of $D$ with~\Omc and let $V$
  be the minimal model of $\theta$. We first show that for all $i \geq
  0$,
    \begin{enumerate}

    \item[($*$)] if $R(\bar c) \in D_i$ and $\bar c \in
      \mn{adom}(D)^{\mn{ar}(R)}$, then $x_{R(\bar c)} \in V$. % ,
      % implying 
      % $R(\bar c) \in \mn{ch}^q_\Omc(D)$

    % \item[(ii)] if $p(\bar y) \in \mn{cl}(q)$, $S$ is a guarded set in
    %   $D$, and there is a homomorphism $h$ from $p$ to
    %   $D_i|^\downarrow_{S}$ with $h(\bar y)=\bar c$, then
    %   $\bar c \in p(\mn{ch}^q_\Omc(D))$.
      
    \end{enumerate}
    The proof is by induction on $i$ and the induction start holds since
    $D_0=D$ and by Point~1 of the definition of $\theta$.
  
    For the induction step, let $R(\bar c) \in D_i \setminus D_{i-1}$
    with $\bar c \in \mn{adom}(D)^{\mn{ar}(R)}$ and $i >0$. Then
    $R(\bar c)$ was added by the chase step that produced $D_i$ from
    $D_{i-1}$, by applying a TGD $T= \phi(\bar x,\bar y) \rightarrow
    \exists \bar z \, \psi(\bar x,\bar z)$ at a tuple $(\bar d, \bar
    d')$. Since the chase step has added $R(\bar c)$ and $\bar c \in
    \mn{adom}(D)^{\mn{ar}(R)}$, every constant that occurs in $\bar c$
    must also occur in $\bar d$. % Let the CQ $\phi(\bar u)$ be
    % obtained from $\phi(\bar d,\bar d')$ by making every constant in
    % $\bar c$ an answer variable and all other constants a quantified
    % variable. 
    Consider the fact $R'(\bar d,\bar d') \in \phi(\bar d,\bar d')$
    that corresponds to the guard atom and let $S=\bar d \cup \bar d'$
    if all constants in $\bar d \cup \bar d'$ are from $\mn{adom}(D)$
    and $S=\mn{source}(R'(\bar d, \bar d'))$ otherwise. Then $S$ is a
    guarded set in $D$ that contains all constants from $\bar c$, and
    $\phi(\bar d, \bar d') \subseteq
    D_i|^\downarrow_S$. By~Lemma~\ref{lem:treesarenice}, there is a
    homomorphism $h$ from $D_i|^\downarrow_S$ to
    $\mn{ch}_\Omc(D_i|_S)$ that is the identity on all constants in
    $S$. Thus $\phi(h(\bar d),h(\bar d')) \subseteq
    \mn{ch}_\Omc(D_i|_S)$ and thus $T$ was applied at $(h(\bar
    d),h(\bar d'))$ during the construction of $\mn{ch}_\Omc(D_i|_S)$,
    yielding $R(\bar c) \in \mn{ch}_\Omc(D_i|_S)$. Consequently
    $D_{I_i}|_S\cup\Omc \models R(\bar c)$. Thus Point~(2) from the
    definition of $\theta$ includes the implication
    $\bigwedge_{R'(\bar e) \in D_i|_S} x_{R'(\bar e)} \rightarrow
    x_{R(\bar c)}$. From the induction hypothesis, we know that
    $x_{R'(\bar e)} \in V$ for all $R'(\bar e) \in D_i|_S$, and thus
    $x_{R(\bar c)} \in V$ as desired.  This finishes the proof
    of~($*$).

    \smallskip

    We now show that $D'_p$ has also been added to $D_\theta$ during
    the construction of this database. Since $D \cup \Omc \models
    p(\bar c)$, there is a homomorphism $h$ from $p$ to
    $\mn{ch}_\Omc(D)|^\downarrow_{S}$ with $h(\bar y)=\bar c$. It
    follows that there is an $i \geq 0$ such that $h$ is also a
    homomorphism from $p$ to $D_i|^\downarrow_S$.
    Lemma~\ref{lem:treesarenice} implies that $D_i|_S \cup \Omc
    \models p(\bar c)$. Thus Point~(2) from the definition of $\theta$
    includes the implication $\bigwedge_{R'(\bar e) \in D_i|_S}
    x_{R'(\bar e)} \rightarrow x_{p(\bar c)}$. By ($*$), $x_{R'(\bar
      e)} \in V$ for all $R'(\bar e) \in D_i|_S$ and thus $x_{p(\bar
      c)} \in V$.  By construction of $D_\theta$, it follows that
    $D'_p$ has been added to this database.
\end{proof}
For later use, we make explicit the structure of $\mn{ch}^q_\Omc(D)$.
In fact, $\mn{ch}^q_\Omc(D)$ can be obtained from $D$ by grafting a
database onto every guarded set in $D$, but the grafted databases are
not connected other than through the guarded sets onto which they are
grafted.  Note that the size of a guarded set is bounded by a constant
when we assume the OMQ to be fixed.
% To make this formal, take a
% database $E$ and let $N \subseteq \mn{adom}(E)$. We say that $N$ is
% \emph{chase-like in} $E$ if there are (possibly empty) databases
% $C,D_{S_1},\dots,D_{S_n}$, with $S_1,\dots,S_n$ the guarded sets in
% $C$, such that $N \cap C = \emptyset$,  $\mn{adom}(D_{S_i}) \cap \mn{adom}(C) \subseteq S_i$
% for $1 \leq i \leq n$,
% $\mn{adom}(D_{S_i}) \cap \mn{adom}(D_{S_j}) \subseteq \mn{adom}(C)$
% whenever $i \neq j$, and
% $E=C \cup \bigcup_{1 \leq i \leq n} 
% D_{S_i}$. We call $C,D_{S_1},\dots,D_{S_n}$ a \emph{witness} for
% $E$ being chase-like.
To make this formal, we say that a database $E$ is \emph{chase-like}
if there are databases $D_{1},\dots,D_{n}$ such that
\begin{itemize}

\item $E=D_1 \cup \cdots \cup D_n$,

\item $D_i$ contains exactly one fact that uses no nulls,
  and that fact contains all constants in  $\mn{adom}(D_i) \setminus \Nbf$,
  
  % $\mn{adom}(D_i) \setminus N$ is a guarded 
  % set in $D_i$ for $1 \leq i \leq n$, and 

\item
 % $(\mn{adom}(D_i) \setminus N) \neq (\mn{adom}(D_j) \setminus N)$ 
%and
  $\mn{adom}(D_i) \cap \mn{adom}(D_j) \cap \Nbf = \emptyset$
    for
  $1 \leq i < j \leq n$.

\end{itemize}
We call $D_{1},\dots,D_{n}$ a
\emph{witness} for $E$ being chase-like.
\begin{lemma}
  \label{lem:pseudochasestructure}
  $\mn{ch}_\Omc^q(D)$ is chase-like and
  there is a witness $D_{1},\dots,D_{n}$ such that
  $|\mn{adom}(D_{i})|$ does not depend on $D$ for $1 \leq i \leq n$.
%
  % Let $\Smf$ be the set of all guarded sets in $D$.
  % There is a family of databases $(D_S)_{S \in \Smf}$,
  % such that $\mn{adom}(D_S) \cap \mn{adom}(D) \subseteq S$
  % and $\mn{adom}(D_S) \cap \mn{adom}(D_{S'}) \subseteq \mn{adom}(D)$
  % for all $S,S' \in \Smf$, and $\mn{ch}^q_\Omc(D) = D \cup \bigcup_{S \in
  %   \Smf} D_S$.
\end{lemma}
We are now ready to prove Theorem~\ref{thm:singlelin}.

\thmsinglelin*

Assume that we are
given a weakly acyclic OMQ $Q(\bar x)=(\Omc,\Sbf,q) \in
(\class{G},\class{CQ})$, an \Sbf-database $D$, and a $\bar c \in
\mn{adom}(D)^{|\bar x|}$, and we have to decide whether $\bar c \in
Q(D)$ (complete answers case). We first compute $\mn{ch}^q_\Omc(D)$ in
time linear in $||D||$. Introduce a fresh unary relation symbol
$P_{\mn{db}}$. We next extend $\mn{ch}^q_\Omc(D)$ to a database $D'$
by adding the fact $P_{\mn{db}}(c)$ for every $c \in \mn{adom}(D)$ and
obtain the CQ $q'(\bar x)$ from $q$ by adding the atom
$P_{\mn{db}}(x)$ for every answer variable $x$. Note that since $q$ is
weakly acyclic, so is $q'$.  It follows from
Lemma~\ref{lem:pseudochase} that $Q(D)=q'(D')$ and thus it suffices to
check whether $\bar c \in q'(D')$. Construct the Boolean CQ $q''$
which is obtained from $q'$ by replacing the answer variables with the
constants from~$\bar c$. Clearly, $q''$ is acyclic and we have to
check whether $() \in q''(D')$.  This can be done in linear time using
existing procedures such as Yannakakis'
algorithm~\cite{yannakakis-algotrithm}.

\medskip

{%\color{orange} 
Now for the case of minimal partial answers with a single wildcard. We start with observing
that it suffices to show that, given an acyclic OMQ
$Q(\bar x)=(\Omc,\Sbf,q) \in (\class{G},\class{CQ})$, an \Sbf-database
$D$, and a
$\bar c^\ast \in (\mn{adom}(D) \cup \{ \ast \})^{|\bar x|}$, it can be
decided in linear time whether $\bar c^\ast$ is a (not necessarily
minimal) partial answer to $Q$ on $D$.

Assume that we have a linear time algorithm for this task. Given an
acyclic OMQ $Q(\bar x)=(\Omc,\Sbf,q) \in (\class{G},\class{CQ})$, an
\Sbf-database $D$, and a
$\bar c^\ast \in (\mn{adom}(D) \cup \{ \ast \})^{|\bar x|}$, we can then
decide in linear time whether $\bar c^\ast \in Q(D)^\ast$ in the
following way. First, we check whether $\bar c^\ast$ is a partial
answer to $Q$ on $D$ and return `no' if this is not the case.
Next, let $V$ be the set of all answer variables $z$ in $q$ such
that the positions in $\bar c^\ast$ that correspond to $z$ are
filled with `$\ast$'. Introduce a fresh unary relation
symbol $P_{\mn{db}}(c)$ and let $D'$ be obtained from $D$ by adding
the fact $P_{\mn{db}}(c)$ for every $c \in \mn{adom}(D)$. For every
$z \in V$, let $Q_z(\bar x')=(\Omc,\Sbf,q_z)$ where $q_z$ is obtained
from $q$ by adding the atom $P_{\mn{db}}(z)$. We then test whether
$\bar c^\ast$ is a partial answer to $Q_z$ on $D'$ and return `no' if
the check succeeds for any $z \in V$ and `yes' otherwise. To
see that this is correct note that if $\bar c^\ast$ is a partial
answer to $Q_z$ on $D'$, then we also find a partial answer
to $Q$ on $D$ in which all positions in $\bar c^\ast$ that correspond
to $z$ are replaced with a constant from $\mn{adom}(D)$, thus
$\bar c^\ast$ is not a minimal partial answer.

We now show that the linear time algorithm for single-testing (not
necessarily minimal) partial answers indeed exists.  Assume that we are
given an acyclic OMQ
$Q(\bar x)=(\Omc,\Sbf,q) \in (\class{G},\class{CQ})$, an \Sbf-database
$D$, and a
$\bar c^\ast \in (\mn{adom}(D) \cup \{ \ast \})^{|\bar x|}$, and we
want to decide whether $\bar c^\ast$ is a partial answer to $Q$
on $D$. This can be done as follows.  We first check whether
$\bar c^\ast$ is coherent with $\bar x$ in the sense that $x_i = x_j$
implies $c_i = c_j$, and return `no' if this is not the case.  Let
$Q'(\bar x')=(\Omc,\Sbf,q')$ where $q'(\bar x')$ is obtained from
$q(\bar x)$ by quantifying all answer variables $z$ such that $x_i=z$
implies $c_i=\ast$, and let $\bar c$ be obtained from $\bar c^\ast$ by
dropping all $c_i$ with $c_i=\ast$.  Note that $q'$ is ayclic since
$q$ is acyclic (whereas $q'$ would not be guaranteed to be weakly
acyclic if $q$ was only weakly acyclic). It then remains to check
whether $\bar c \in Q'(D)$, using the algorithm for complete answers
from above, which is the case if and only if $\bar c^\ast$ is a
partial answer to $Q$ on~$D$.

\medskip

We next consider the case of minimal partial answers with
multi-wildcards. We first make the following observation.  
\\[2mm]
\emph{Claim.}  Let
$Q(\bar x)=(\Omc,\Sbf,q) \in (\class{ELI},\class{CQ})$ be acyclic, $D$
be an \Sbf-database, and $\bar c^\Wmc \in Q(D)^\Wmc$.  Further let
$\widehat q$ be obtained from $q$ by identifying any two answer
variables $x_1,x_2$ such that the corresponding positions in $\bar c$
are filled with the same wildcard. Then
$\widehat q$ is weakly acyclic.
\\[2mm]
To prove the claim, assume that $Q$, $D$, $\bar c$, and $\widehat q$
are as in the claim. Let $\bar x = x_1 \cdots x_n$ and
$\bar c^\Wmc = c_1 \cdots c_n$. For every answer variable~$x_i$, we
use $c(x_i)$ to denote $c_i$. %  Let $\widehat q'$ be obtained from
% $\widehat q$ by dropping all atoms that contain an answer variable $x$
% with $c(x) \in \mn{adom}(D)$.
Since $\bar c^\Wmc$ is a partial answer
to $Q$ on~$D$, there is a homomorphism $h$ from $q$ to
$\mn{ch}_\Omc(D)$ such that for all answer variables $x,x'$,
$c(x) \in \mn{adom}(D)$ implies $h(x)=c(x)$, and
$c(x) = c(x') \in \Wmc$ implies $h(x)=h(x')$. We prove that for all
answer variables $x,x'$ with $c(x),c(x') \in \Wmc$,
  \begin{enumerate}

  \item $h(x) \notin \mn{adom}(D)$ and
    
  \item $h(x_1)=h(x_2)$ iff $c(x_1)=c(x_2)$.

  \end{enumerate}
  For Point~1, assume to the contrary that there is an answer variable
  $x$ with $h(x) \in \mn{adom}(D)$. Let ${\bar{c}}^{\prime\Wmc}$ be obtained
  from $\bar c^\Wmc$ by replacing $c_i$ with $h(x)$ whenever $x_i=x$.
  The homomorphism $h$ witnesses that ${\bar c}^{\prime\Wmc}$ is a partial
  answer to $Q$ on $D$, but ${\bar c}^{\prime\Wmc} \prec \bar c^\Wmc$ in
  contradiction to $\bar c^\Wmc$ being a minimal partial answer.  The
  `if' direction of Point~2 is clear by choice of $h$. For the `only
  if' direction, assume to the contrary that there are answer
  variables $x,x'$ with $c(x),c(x') \in \Wmc$, $c(x) \neq c(x')$, and
  $h(x)=h(x')$. Let ${\bar{c}}^{\prime\Wmc}$ be obtained from $\bar c^\Wmc$ by
  choosing a fresh wildcard $\ast_\ell$, replacing $c_i$ with
  $\ast_\ell$ whenever $c(x_i) \in \{ c(x), c(x')\}$, and then renaming
  wildcards to make them consecutive again. That is, the variables
  $x,x'$ have distinct wildcards in $\bar c^\Wmc$, but the same
  wildcard in ${\bar c}^{\prime\Wmc}$. The homomorphism $h$ witnesses that
  ${\bar c}^{\prime\Wmc}$ is a partial answer to $Q$ on $D$, but
  ${\bar c}^{\prime\Wmc} \prec \bar c^\Wmc$ in contradiction to $\bar c^\Wmc$
  being a minimal partial answer.

  To see that $\widehat q$ is acyclic, first consider the restriction
  $\widehat q_0$ of $\widehat q$ to the answer variables $x$ with
  $c(x) \in \Wmc$. By construction of $\widehat q_0$ and Point~(2), $h$ is a
  homomorphism from $\widehat q_0$ to $\mn{ch}_\Omc(D)$. By Point~(2)
and since no wildcard occurs twice in $\widehat q_0$ (which is
  due to the construction of $\widehat q$), $h$ as a homomorphism from
  $\widehat q_0$ to $\mn{ch}_\Omc(D)$ is injective. By Point~(1)
  above, there is no variable $x$ in $\widehat q_0$ with
  $h(x) \in \mn{adom}(D)$. Since \Omc is formulated in \ELI, however,
  the restriction of $\mn{ch}_\Omc(D)$ to
  $\mn{adom}(\mn{ch}_\Omc(D)) \setminus \mn{adom}(D)$ is acyclic.  The
  injectivity of $h$ thus implies that $\widehat q_0$ is acyclic. Now
  consider the restriction $\widehat q_1$ of $\widehat q$ to the
  variables in $\widehat q_0$ plus the quantified variables.
  Acyclicity of $\widehat q_0$ and of the original query $q$ implies
  that $\widehat q_1$ is acyclic. This, in turn, clearly implies that
  $\widehat q$ is weakly acyclic. This finishes the proof of the
  claim.

  \medskip

  Now assume that we are given an acyclic OMQ
  $Q(\bar x)=(\Omc,\Sbf,q) \in (\class{ELI},\class{CQ})$, an
  \Sbf-database $D$, and a
  $\bar c^\Wmc \in (\mn{adom}(D) \cup \Wmc)^{|\bar x|}$, and that we
  want to decide whether $\bar c^\Wmc \in Q(D)^\Wmc$.  We first verify
  that $\bar x = x_1 \cdots x_n$ is coherent with
  $\bar c^\Wmc = c_1 \cdots c_n$ in the sense that $x_i = x_j$ implies
  $c_i = c_j$, returning `no' if this is not the case. For every
  answer variable $x_i$, we again use $c(x_i)$ to denote $c_i$.  We
  then construct $\widehat q$ as in the claim and check whether it is
  weakly acyclic, returning `no' if this is not the case. Let
  $\bar c^*$ be constructed from $\bar c^\Wmc$ by mirroring the
  construction of~$\widehat q$, that is, whenever two answer variables
  $x_i,x_j$ are identified in $\widehat q$, then the corresponding
  positions (which carry the same wildcard) are identified in
  $\bar c^*$. Clearly, every wildcard $\ast_i$ occurs only once in~$\bar c^*$.
  Thus, we can assume that all these wildcards are `$\ast$' and
  $\bar c^*$ is, in fact, a single-wildcard tuple.  We check whether $\bar c^*$
  is a partial answer to $Q$ on $D$ using the procedure for
  single-testing partial answers with a single wildcard for
  $(\class{G},\class{CQ})$ given above. It can be verified that this
  is the case if and only if $\bar c^\Wmc$ is a partial answer to $Q$
  on $D$, and thus we return `no' if the check fails.
  
  It remains to check whether there is a partial answer
  $\bar a^\Wmc$ to $Q$ on $D$ with
  $\bar a^\Wmc \prec \bar c^\Wmc$.% \footnote{It would of course also be
    % sufficient to check whether there is a (not necessarily least)
    % partial answer of this kind, but it is not clear how to check that
    % as the claim only applies to least partial answers.} 
  Introduce a
  fresh unary relation symbol $P_{\mn{db}}(c)$ and let \Qmc be the
  set of all pairs $(p,\bar b^\Wmc)$ with $p$ a CQ and $\bar b^\Wmc$
  a tuple that can be obtained from $q$ and $\bar c^\Wmc$ in the
  following way:
  \begin{itemize}

  \item choose a set $V$ of answer variables $x$ such that
    the corresponding positions in $\bar c^\Wmc$ have a wildcard; then
    obtain $p$ from $q$ by adding $P_{\mn{db}}(x)$ for all $x \in V$
    and
    set $\bar b^\Wmc=\bar c^\Wmc$;

  \item choose a partition $S_1,\dots,S_k$ of the set of indices of
    wildcards that occur in $\bar c^\Wmc$; then set $p=q$ and obtain
    $\bar b^\Wmc$ from $\bar c^\Wmc$ by replacing, for
    $1 \leq i \leq k$, every wildcard `$\ast_\ell$' with
    $\ell \in S_i$ by `$\ast_{\widehat \ell}$' where
    $\widehat \ell \in S_i$ is a chosen representative; then rename
    wildcards to make them consecutive again.

  \end{itemize}
  This is subject to the condition that $V$ is non-empty or the
  partition $S_1,\dots,S_k$ contains at least one non-singleton set.
  Let $D'$ be obtained from $D$ by adding the fact $P_{\mn{db}}(c)$
  for every $c \in \mn{adom}(D)$. We then do the following for all
  $(p,\bar b^\Wmc) \in \Qmc$:
  \begin{enumerate}

  \item if $\widehat p$ is not weakly acyclic, then skip;

  \item otherwise, proceed as in the case of $q$ to check
    whether $\bar b^\Wmc$ is a partial answer to $(\Omc,\Sbf,p)$ on $D'$
    (by transitioning to $\widehat p$ and the corresponding
    single-wildcard tuple $\bar b^\ast$ obtained from $\bar b^\Wmc$
    and using the algorithm for single-testing with single wildcards).

  \end{enumerate}
  If any of the checks succeeds, then there is a partial answer
  $\bar a^\Wmc$ to $Q$ on $D$ such that
  $\bar a^\Wmc \prec \bar c^\Wmc$, and thus we answer `no'. If, for
  example, $(p,\bar b^\Wmc) \in \Qmc$ due to the choices $V$ and
  $S_1,\dots,S_k$ with $V$ non-empty, then we may obtain $\bar a^\Wmc$
  from $\bar c^\Wmc$ by replacing all wildcard tuples in $\bar c^\Wmc$
  that correspond to variables in $V$ with constants from
  $\mn{adom}(D)$. Conversely, assume that there is a partial answer
  $\bar a^\Wmc$ to $Q$ on $D$ such that
  $\bar a^\Wmc \prec \bar c^\Wmc$. Then we may assume that
  $\bar a^\Wmc$ is a \emph{minimal} partial answer. It should be clear
  that $\bar a^\Wmc$ induces a pair $(p,\bar b^\Wmc) \in \Qmc$ by
  `reading off' $V$ and $S_1,\dots,S_k$ from $\bar a^\Wmc$. By the
  claim, $(p,\bar b^\Wmc)$ is not skipped in Step~1 above and the
  check in Step~2 succeeds, so the algorithm returns `no', as desired.
}

\subsection{Proof of Theorem~\ref{lemma:single-testing-lower-bound}}

\lemmasingletestinglowerbound*

To prove Theorem~\ref{lemma:single-testing-lower-bound}, it suffices
to consider complete answers to Boolean OMQs that are not acyclic,
non-empty and self\=/join free. To see this, first assume that there
is an OMQ $Q(\bar x) =(\Omc,\Sbf,q) \in (\class{ELI}, \class{CQ})$
that satisfies the conditions from
Theorem~\ref{lemma:single-testing-lower-bound} and such that
single-testing complete answers to $Q$ is possible in linear time.
Since $Q$ is non-empty, there is an \Sbf-database $D_0$ and a tuple
$\bar a \in Q(D_0)$. Define the Boolean OMQ $Q_{\bar
  a}:=(\Omc,\Sbf,q[\bar a / \bar x])$.  It is easy to see that
$Q_{\bar a}$ is not acyclic, non-empty and self\=/join free.  In
particular, $D_0$ witnesses its non-emptiness.  Moreover,
single-testing for $Q_{\bar a}$ is then also possible in linear time
because $() \in Q_{\bar a}(D)$ iff
% all constants in $\bar a$ occur in $D$ and
$\bar{a} \in Q(D)$, for all \Sbf-databases $D$. If we have proved
Theorem~\ref{lemma:single-testing-lower-bound} for complete answers
and the class of OMQs described above, of which $Q_{\bar a}$ is a
member, it thus follows that the triangle conjecture fails. Regarding
minimal partial answers and minimal partial answers with multiple
wildcards, it suffices to observe that these agree with complete
answers for Boolean OMQs.

Let $Q() =(\Omc,\Sbf,q) \in (\class{ELI},\class{CQ})$ be not acyclic,
non-empty, and self\=/join free.\footnote{We remark that the proof
  still goes through if $q$ is self\=/join free only regarding the
  binary atoms, but not necessarily regarding the unary ones.} We show
how to construct, given an undirected graph $G=(V,E)$, an
\Sbf-database $D$ such that $D \models Q$ if and only if $G$ contains
a triangle. When speaking about undirected graphs, we generally mean
graphs without self loops and isolated vertices.  Since $Q$ is
non-empty, there is an \Sbf-database $D_0$ with $D_0 \models Q$ and a
homomorphism $h_0$ from $q$ to $\mn{ch}_\Omc(D_0)$. We are going to
use $D_0$ and $h_0$ throughout the subsequent proof.

Since $q$ is not acyclic, the undirected graph $G^{\mn{var}}_q$
contains a cycle $x_1,\dots,x_n$ of length at least three.   We may
 assume w.l.o.g.\ that the cycle is chordless.
To
ease notation, set $x_{n+1}:=x_1$. For $1 \leq i \leq n$, a binary
relation $R$ is an \emph{$i$-relation} if $q$ contains an atom
$R(x_i,x_{i+1})$ or $R(x_{i+1},x_i)$. Note that no $R$ can be an
$i_1$-relation and an $i_2$-relation with $i_1 \neq i_2$ since $q$ is
self\=/join free.
\begin{lemma}
  \label{lem:alliinS}
  If $R$ is an $i$-relation, $1 \leq i \leq n$, then
$R \in \Sbf$. % Moreover, if $R(x_i,x_i) \in q$, then
% $R \in \Sbf$.
\end{lemma}
\begin{proof}
The lemma follows from $Q$ being non-empty and self\=/join free. In
fact, assume to the contrary of what is to be shown that there is an
$i$-relation $R$, $1 \leq i \leq n$, such that $R \notin \Sbf$. We
have $h_0(x_i) \in \mn{adom}(D_0)$ for $1 \leq i \leq n$ since
$x_1,\dots,x_n$ participate in a cycle in the self\=/join free CQ $q$,
and the null part of $\mn{ch}_\Omc(D_0)$ consists of a disjoint union
of trees without self loops and multi-edges (because \Omc is
formulated in \ELI). However, the restriction of $\mn{ch}_\Omc(D_0)$
to $\mn{adom}(D_0)$ contains no $R$-edges with $R \notin \Sbf$, and
thus all $i$-relations must be from \Sbf.
\end{proof}
%
% Consequently, there is a homomorphism $h$ from $q$ to
% $\mn{ch}_\Omc(D_0)$. Since $R \notin \Sbf$, $R$-edges in
% $\mn{ch}_\Omc(D_0)$ can only be inside the trees generated by the
% chase, and thus $h(x_i)$ and $h(x_{i+1})$ are inside such a tree.
% Assume that $h(x_{i+1})$ is further away from the database part of
% $\mn{ch}_\Omc(D_0)$ than $h(x_i)$; the case that $h(x_{i})$ is further
% away than $h(x_{i+1})$ is analogous.  Since $q$ is self\=/join free,
% $h(x_{i+2})$ must be further away from the database part of
% $\mn{ch}_\Omc(D_0)$ than $h(x_{i+1})$ and so on, round robbin' to $x_1$
% and proceeding to $h(x_{i})$ which is thus further away from the
% database part of $\mn{ch}_\Omc(D_0)$ than itself, a
% contradiction.
%
We now construct the \Sbf-database $D$.  Let $C=\{c_0,\dots,c_{k-1}\}$
be a set of $k=\max\{4,\mn{con}(q)\}$ constants with
$\mn{con}(q) \subseteq C$. The active
domain of the database $D$ is
$$\dom(D) = V \cup C
$$
and $D$ contains the following facts:
%For every unary relation $A$ in $\Sbf$ and constant $c \in \dom(D)$
%we add facts $A(c)$ do the database.
%
\begin{itemize}

\item $A(c)$ for every unary $A \in \Sbf$ and every $c \in \dom(D)$;

\item $R(c,c)$ for every $i$-relation $R$, $3 < i \leq
  n$, and every $c \in \dom(D)$;

\item for every $i$-relation $R$, $1 \leq i \leq
  3$:
  \begin{itemize}

  \item $R(u,v)$ for every edge $\{u,v\} \in E$;

 \item $R(c_i,c_{i + 1 \,\mn{mod}\, k})$ for $0 < i \leq k$.
  \end{itemize}

\item $R(c,c)$, $R(c,c_\ell)$, and $R(c_\ell,c)$
%  $R(c,c')$
  for every binary
  $R \in \Sbf$ that is not an $i$-relation for any $i$, all
  $c' \in \dom(D)$, and $1 \leq \ell < k$.
\end{itemize}
Due to Lemma~\ref{lem:alliinS}, $D$ is indeed an \Sbf-database. The
construction of $D$ strongly relies on self\=/join freeness as this
makes the interpretation of $i$-relations in $D$ independent of each
other for different $i$. It should be clear that $D$ can be
constructed in time linear in $|E|$.

We remark that $D$ has two important properties. First, every
$c \in \dom(D)$ has both an incoming $R$-edge and an outgoing $R$-edge
for every binary relation $R \in \Sbf$. And second, every triangle in
$D$ that only uses the relations $R_1,R_2,R_3$ is in $D|_V$; more
precisely: if $D' \subseteq D$ is an $\{ R_1,R_2,R_3\}$-database such
that $|\mn{adom}(D')|=3$ and the Gaifman graph of $D'$ is a triangle,
then $\mn{adom}(D') \subseteq V$.  We refer to this as the
\emph{completeness property} and the \emph{triangle property},
respectively. To complete the reduction, it suffices to show the
following.
\begin{lemma}
  $G$ contains a triangle iff $D \models Q$.
\end{lemma}
\begin{proof}
``if''. Assume that $D \models Q$. Then there is a homomorphism $h$
from $q$ to $\mn{ch}_\Omc(D)$. Due to the interpretation of the
$i$-relations in $D$ with $3 < i \leq n$, we must have
$h(x_4)=\cdots=h(x_n)=h(x_1)$. %  Consider the restriction of $\mn{ch}_\Omc(D)$
% to $h(x_1),h(x_2),h(x_3)$.
Consequently, the Gaifman graph of $\mn{ch}_\Omc(D)$ must contain the
edge $\{ h(x_i),h(x_j) \}$ for $1 \leq i,j \leq 3$.  By construction
and since $G$ has no self loops, $D$ contains no reflexive $R_i$-edges
for $1 \leq i \leq 3$. Since $\Omc \in \class{ELI}$, the same is then
true for $\mn{ch}_\Omc(D)$. It follows that $h(x_1),h(x_2),h(x_3)$ are
all distinct and thus the restriction of $\mn{ch}_\Omc(D)$ to
$h(x_1),h(x_2),h(x_3)$ is a triangle. Due to the triangle property and
since the chase with an \ELI-ontology only adds to $D$ trees without
self loops and multi-edges, this implies that $h(x_i) \in V$ for
$1 \leq i \leq 3$. Since the chase adds no binary facts $R(c_1,c_2)$
with $c_1,c_2 \in \mn{adom}(D)$, the restriction of $D$ to
$h(x_1),h(x_2),h(x_3)$ is also a triangle.  It now follows from the
interpretation of the $i$-relations in $D$ with $1 \leq i \leq 3$,
$h(x_1),h(x_2),h(x_3)$ is a triangle in $G$.

\medskip

``only if''. Assume that $G$ contains a triangle that consists of the
vertices $v_1,v_2,v_3$. Let $q^-$ be $q$ with all unary atoms dropped.
We first construct a homomorphism $h$ from $q^-$ to $\mn{ch}_\Omc(D)$
and then argue that it can be extended to a homomorphism from $q$ to
$\mn{ch}_\Omc(D)$. 

We construct $h$ by first mapping only some variables from $q$ and
then extending in several rounds. After each round, the constructed
$h$ will be a homomorphism from $q^-|_h$ to $\mn{ch}_\Omc(D)$ where
$q^-|_h$ is the restriction of $q^-$ to the domain of $h$.

Start with setting $h(x_1)=v_1$, $h(x_2)=v_2$, and
$h(x_3)=\cdots=h(x_n)=v_3$ and $h(c)=c$ for all $c \in \mn{con}(q)$.
It can be verified that $h$ is a homomorphism from $q^-|_h$ to $D$,
thus also to $\mn{ch}_\Omc(D)$. To see this, it helps to observe that
$q$ contains no atoms $R(x_i,x_j)$ with
$x_j \notin \{ x_{i-1},x_i,x_{i+1} \}$ and that all binary atoms in
$q^-|_h$ use a relation from \Sbf. In fact, the former is a
consequence
of $x_1,\dots,x_n$ being a chordless cycle in $G^{\mn{var}}_q$.
For the latter, recall that $h_0$ is
a homomorphism from $q$ to $\mn{ch}_\Omc(D_0)$ that maps all variables
$x_1,\dots,x_n$ to $\mn{adom}(D_0)$, established in the proof of
Claim~1.  But all binary facts in the restriction of
$\mn{ch}_\Omc(D_0)$ to $\mn{adom}(D_0)$ use relations from \Sbf.

To extend $h$, set $h(x)=c_0$ for all $x \in \mn{var}(q)$ with $h(x)$
not yet defined and $h_0(x) \in \mn{adom}(D_0)$ (any other constant
would also work). It can be verified
that $h$ is still a homomorphism from $q^-|_h$ to $D$. In particular,
we can argue as above that all binary atoms in $q^-|_h$ use a relation
from \Sbf.

For the next extension of $h$, let $q'$ be obtained from $q$ by first
dropping all atoms $R(\bar x)$ such that $h_0(x) \in \mn{adom}(D_0)$
for all $x \in \bar x$ and then identifying any variables $x_1,x_2$
such that $h_0(x_1)=h_0(x_2)$.  Consider all maximal connected
components $p$ of $q'$. We distinguish two cases.

First assume that $p$ contains a variable $x$ with $h(x)$ already
defined. It then only contains a single such variable, that is,
$h_0(x) \in \mn{adom}(D_0)$ and all $y \neq x$ in $p$ are mapped to
the tree that the chase has generated below $h_0(x)$. Since the
restriction of $\mn{ch}_\Omc(D_0)$ to the nulls is a disjoint union of
trees without reflexive loops and multi-edges, it follows that $p(x)$
is an ELIQ. Moreover, $D_0 \models (\Omc,\Sbf,p(x))(h_0(x))$. The
completeness property of $D$ implies that
$(D_0,h_0(x)) \preceq (D,h(x))$, no matter what $h(x)$ is. It thus
follows from Lemma~\ref{lem:sim} that
$D \models (\Omc,\Sbf,p(x))(h(x))$, and consequently there is a
homomorphism $g$ from $p(x)$ to $\mn{ch}_\Omc(D)$ with $g(x)=h(x)$. We
can extend $h$ to $p$ by setting $h := h \cup g$.

Now assume that $p$ contains no variable $x$ with $h(x)$ already
defined. Since $p$ is connected, there is a $c \in \mn{adom}(D_0)$
such that $h_0(x)$ is in the tree that the chase has generated below
$c$ for all $x \in \mn{var}(p)$. Let $\widehat p(x)$ be the minimal
prefix of that tree that contains $h_0(x)$ for all
$x \in \mn{var}(p)$, viewed as an ELIQ with $c$ being the answer
variable. There is a homomorphism $g$ from $p(x)$ to $\widehat p(x))$
with $g(x)=x$. We have $D_0 \models (\Omc,\Sbf,\widehat p(x))(c)$.
Take any constant $c' \in \mn{adom}(D)$.  The completeness property of
$D$ implies that $(D_0,c) \preceq (D,c')$. It thus follows from
Lemma~\ref{lem:sim} that $D \models (\Omc,\Sbf,\widehat p(x))(c')$,
and consequently there is a homomorphism $g'$ from $\widehat p(x)$ to
$\mn{ch}_\Omc(D)$ with $g'(x)=c'$. We can extend $h$ to $p$ by
setting $h := h \cup (g \circ g')$.

At this point, $h$ is defined for all terms in $\mn{var}(q)$ and thus
it is a homomorphism from $q^-$ to $\mn{ch}_\Omc(D)$. Further
extend $h$ by setting $h(x)$ to any constant in $\mn{adom}(D)$ if
$x$ does not occur in any binary atoms. It remains to argue that $h$
satisfies all unary atoms $A(x)$ in $a$. If $A \in \Sbf$, then this is
clear by the interepretation of such symbols in $D$. Otherwise,
we observe that 
$A(h_0(x)) \in \mn{ch}_\Omc(D_0)$. We can argue as above, using
the ELIQ $A(x)$ and Lemma~\ref{lem:sim}, that 
$A(h(x)) \in \mn{ch}_\Omc(D))$.
\end{proof}

The next example illustrates another reason for why we cannot easily
replace $\class{ELI}$ with $\class{G}$ in
Theorem~\ref{lemma:single-testing-lower-bound}, unrelated to
self-join freeness. 
\begin{example}
  \label{ex:GacyclicyetlinearFIRST}
  Let $Q(x)=(\Omc,\Sbf,q) \in (\class{G},\class{CQ})$ where
  $$
  \begin{array}{rcl}
  \Omc &=& \{ A(x) \rightarrow \exists y \exists z \, R(x,y)
           \wedge S(y,z) \wedge T(z,x) \} \\
    \Sbf &=& \{ A \} \\
    q(x) &=& R(x,y)
           \wedge S(y,z) \wedge T(z,x).
  \end{array}
  $$
  Then $Q$ is not acyclic and self-join free, yet single-testing
  $Q$ is in linear time as $Q \equiv (\emptyset,\Sbf,A(x))$.
\end{example}
For CQs without ontologies, we are not aware of any examples which
show that Theorem~\ref{lemma:single-testing-lower-bound} fails when
`self-join free' is replaced with `a homomorphism core'.\footnote{A CQ
  is a \emph{homomorphism core} if every homomorphism from $D_q$ to
  $D_q$ is surjective. For every CQ, there is an equivalent CQ that is
  a homomorphism core.} For $(\class{ELI},\class{CQ})$, however, such
examples are not hard to
find. %, varying Example~\ref{ex:Gacyclicyetlinear}.
\begin{example}
  Let $Q=(\Omc,\Sbf,q) \in (\class{ELI},\class{CQ})$ where
  $$
  \begin{array}{rcl}
  \Omc &=& \{ A(x) \rightarrow \exists y \exists z \, R(x,y)
           \wedge B_1(y) \wedge B_2(y) \wedge R(y,z) \} \\
    \Sbf &=& \{ A \} \\
    q(x) &=& R(x,y_1) \wedge R(x,y_2)
             \wedge B_1(y_1) \wedge B_2(y_2) \, \wedge \\
    && R(y_1,z) \wedge R(y_2,z).
  \end{array}
  $$
  Then $Q$ is not acyclic and $q$ is a homomorphism core, yet
  single-testing $Q$ is in linear time as
  $Q \equiv (\emptyset,\Sbf,A(x))$.
\end{example}

{%\color{orange}

\subsection{Proof of Theorem~\ref{thm:anotherpartiallower}}

\anotherpartiallower*
\begin{proof}
  First for weakly acyclic OMQs and minimal partial answers with a
  single wildcard. 
  We start with showing the result for $(\class{G},\class{CQ})$ in
  place of $(\class{ELI},\class{CQ})$. Let $R$ be a binary relation
  symbol and $\Sbf=\{R\}$. We may view an
  undirected graph $G=(V,E)$ as the $\Sbf$-database
  $$D_G=\{ R(v,v'), R(v',v) \mid \{v,v'\} \in E \}.$$  Consider the
  OMQ $Q(x) = (\Omc,\Sbf,q) \in (\class{G},\class{CQ})$ where \Omc
  contains the TGD
  $$
  \begin{array}{r@{\;}c@{\;}l}
    R(x_1,x_2) &\rightarrow& \exists y_1 \exists y_2 \exists y_3 \, R\{y_1,y_2\} \wedge 
                             R\{y_2,y_3\} \wedge R\{y_3,y_1\}
    % \\[0.5mm]
    % R(x_2,x_1) &\rightarrow& \exist y_1 \exists y_2 \exists y_3 \, R\{y_1,y_1\} \wedge 
    % R\{y_1,y_2\} \wedge R\{y_2,x_2\}
  \end{array}
  $$
  with $R\{x,y\}$ an abbreviation for  $R(x,y) \wedge R(y,x)$ and with
  $$
  q(x,y,z)=R\{x,y\} \wedge 
    R\{y,z\} \wedge R\{z,x\}.
  $$
  Note that $q$ is weakly acyclic.  Let $G=(V,E)$ be an undirected
  graph.  Then $(\ast,\ast,\ast)$ is a partial answer to $Q$ on $D_G$,
  but not necessarily a minimal partial answer. In fact, it is a minimal
  partial answer if and only if $G$ contains no triangle. It clearly
  follows that single-testing minimal partial answers for $Q$ is not in
  \dlc unless the triangle hypothesis fails.

  The challenge in improving the construction to $\class{ELI}$ is
  that \ELI TGDs cannot introduce a triangle that consists of nulls.
  The solution is to construct $q$, \Omc, and $D$  in a more careful
  way. Let us start with~$q$, which is now 
  $$
  \begin{array}{r@{\;}c@{\;}l}
  q(x_1,x_2,x_3,x_4) &=& R(x_1,x_2) \wedge 
                         R(x_2,x_3) \wedge R(x_4,x_3) \wedge R(x_1,x_4) \\[1mm]
    &&\wedge \, A(x_1) \wedge B(x_2) \wedge B(x_3) \wedge A(x_4)
  \end{array}
  $$
  where $A$, $B$, $C$ are additional unary relation symbols and
  the direction of the edges is chosen carefully. %  so that the cyclic
  % can be contracted into a path that can be generated by the ontology.
  %In fact,
  We choose \Omc to contain the \ELI TGD
  $$
  \begin{array}{r@{\;}c@{\;}l}
    R(x_1,x_2) &\rightarrow& \exists y_1 \exists y_2 \exists y_3 \, R(x_1,y_2) \wedge 
                             R(y_2,y_3) \\[1mm]
    && \qquad\qquad\ \ \ \wedge \, A(y_1) \wedge A(y_2) \wedge B(y_2) \wedge B(y_3).
  \end{array}
  $$
  Similarly to before, $(\ast,\ast,\ast,\ast)$ is a partial answer to
  $Q$ on any non-empty $\Sbf$-database, such as the databases
  $D_G$. It remains to modify $D_G$ so that  $q(D_G)$ is non-empty if
  and only if $G=(V,E)$ contains a triangle, as $q(D_G)$ is empty if
  and only if $(\ast,\ast,\ast,\ast)$ is a \emph{minimal} partial
  answer. This is achieved by constructing $D_G$ so that
  $\mn{adom}(D_G)=V \times [2]$ and $D_G$ contains the following
  facts:
  \begin{itemize}

  \item $R((v,1),(v',2)), R((v,2),(v',2))$ for all
    $\{v,v'\} \in E$;

  \item $R((v,1),(v,1)), A(v,1), B(v,2)$ for all $v \in V$.

  \end{itemize}
  It thus remains to argue that $q(D_G)$ is non-empty if and only if
  $G=(V,E)$ contains a triangle. For the ``if'' direction, it suffices
  to observe that any triangle $v_1,v_2,v_3$ in $G$ gives rise to
  the answer
  $
      ((v_1,1),(v_2,2),(v_3,2),(v_1,1))
  $
  in $q(G)$. For the ``only if'' direction, let
  $(c_1,c_2,c_3,c_4) \in q(G)$. Due to the atoms
  $A(x_1), R(x_1,x_4),A(x_4)$ in $q$, we must have $c_1=c_4=(v_1,1)$
  for some $v_1 \in V$. Due to the atoms $A(x_1),R(x_1,x_2),B(x_2)$
  and $A(x_4),R(x_4,x_3),B(x_3)$, respectively, we must have
  $c_2=(v_2,2)$ and $c_3=(v_3,2)$ for some $v_2, v_3 \in V$ such that
  $v_1 \neq v_2$, $v_1 \neq v_3$, and $\{v_1,v_2\},\{v_2,v_3\} \in G$.
  Finally, the atoms $B(x_2),R(x_2,x_3),B(x_3)$ ensure that
  $v_2 \neq v_3$ and $\{v_1,v_2\} \in G$. Thus, $v_1,v_2,v_3$
  form a triangle in $G$.

  \medskip We now turn towards acyclic OMQs and minimal partial answers
  with multi-wildcards. The general idea is similar to what was done
  above. In particular, with every undirected graph $G=(V,E)$ we
  associate an \Sbf-database
  $D_G=\{ R(v,v'), R(v',v) \mid \{v,v'\} \in E \}.$ Consider the OMQ
  $Q(x) = (\Omc,\Sbf,q) \in (\class{G},\class{CQ})$ where \Omc
  contains the TGD
  $$
  \begin{array}{r@{\;}c@{\;}l}
    R(x_1,x_2) &\rightarrow& \exists y_1 \exists y_2 \exists y_3 \, R\{y_1,y_2\} \wedge 
                             R\{y_2,y_3\} \wedge R\{y_3,y_1\}
    % \\[0.5mm]
    % R(x_2,x_1) &\rightarrow& \exist y_1 \exists y_2 \exists y_3 \, R\{y_1,y_1\} \wedge 
    % R\{y_1,y_2\} \wedge R\{y_2,x_2\}
  \end{array}
  $$
  with $R\{x,y\}$ an abbreviation for  $R(x,y) \wedge R(y,x)$ and with
  $$
  q(x_1,x'_1,x_2,x'_2,x_3,x'_3)=R(x_1,x'_1)\wedge 
    R(x_2,x'_2) \wedge R(x_3,x'_3).
  $$
  Note that $q$ is acyclic and in fact of a very restricted shape.
  Let $G=(V,E)$ be an undirected graph.  Then
  $(\ast_1,\ast_2,\ast_2,\ast_3,\ast_3,\ast_1)$ is a partial answer
  with multi-wildcards to $Q$ on $D_G$. Moreover, it is a minimal
  partial answer if and only if $D_G$ contains no triangle. It clearly
  follows that single-testing minimal partial answers with
  multi-wildcards for $Q$ is not in \dlc unless the triangle
  hypothesis fails.
\end{proof}
}

\section{Proofs for Section~\ref{sect:enumallcomplete}}

\propallTestingCompleteUpper*
\begin{proof}
  Let $q(\bar x)$ be a CQ that is free-connex acyclic, over some
  schema \Sbf. Since $q$ is free-connex acyclic, its extension
  $q^+(\bar x)$ with atom $R_0(\bar x)$ is acyclic, where $R_0$ is a
  fresh relation symbol of arity $|\bar x|$. Thus $q^+$ has a join
  tree $T=(V,E)$. When removing from $T$ the node $R_0(\bar x)$, we
  obtain a forest that consists of trees $(V_1,E_1),\dots,(V_k,E_k)$.
  Let $q_1(\bar x_1),\dots,q_k(\bar x_k)$ be the corresponding CQs,
  that is, $q_i$ contains exactly the atoms in $V_i$. It is clear that
  every $q_i$ is acyclic since $(V_i,E_i)$ is a join tree for
  $q_i$. It is also free-connex acyclic. In fact, let $q_i^+(\bar x_i)$
  be $q_i$ extended with atom $R_i(\bar x_i)$. We obtain a join
  tree for $q_i^+$ by starting with $(V_i,E_i)$, adding node
  $R_i(\bar x_i)$, and an edge between $R_i(\bar x_i)$ and the node
  in $V_i$ that is adjacent in $T$ to  $R_0(\bar x)$ (there must be a
  unique such node).

  It is known that all-testing is in \cdlin for all CQs that are
  acyclic and free-connex acyclic \cite{berkholz-enum-tutorial}.
  There are thus \cdlin all-testing algorithms $A_1,\dots,A_k$ for
  $q_1,\dots,q_k$.  We devise a \cdlin all-testing algorithm for~$q$
  by combining these. Let the \Sbf-database $D$ be given as input. In
  the preprocessing phase, we run the preprocessing phases of
  algorithms $A_1,\dots,A_k$ on $D$. In the testing phase, we are
  given a tuple $\bar c \in \mn{adom}(D)^{|\bar x|}$. Let
  $\bar x=x_1 \cdots x_n$ and $\bar c = c_1 \cdots c_n$. We first
  check whether $x_i=x_j$ implies $c_i = c_j$ and reject if this is
  not the case. We then use algorithms $A_1,\dots,A_k$ to test in
  constant time whether for $1 \leq i \leq k$, it holds that
  $\bar c_i \in q_i(D)$ where $\bar c_i$ is the `projection' of tuple
  $\bar c$ to the answer variables in $q_i$, that is, if
  $x_{i_1},\dots,x_{i_\ell}$ are the answer variables in $q_i$, then
  $c_i=c_{i_1},\dots,c_{i_\ell}$.  We answer `yes' if all checks
  succeed and `no' otherwise. Since $q_1,\dots,q_k$ is a partitioning
  of (the atoms of)~$q$ and distinct $q_i$ do not share any quantified
  variables, the answer is clearly correct.
\end{proof}

\thmupperG*
\begin{proof}
  Point~(1) of Theorem~\ref{thm:upperG} is easy to prove using the
  query-directed chase established in
  Section~\ref{app:sect:singletest} and the existing result stating
  that enumeration is in \cdlin for CQs that are acyclic and
  free\=/connex acyclic when no ontologies are present
  \cite{bagan-enum-cdlin}. In fact, assume that an OMQ
  $Q(\bar x)=(\Omc,\Sbf,q) \in (\class{G},\class{CQ})$ that is acyclic
  and free\=/connex acyclic is given, as well as an
  \Sbf-database~$D$. As a part of the preprocessing phase, we compute
  in linear time the query-directed chase $\mn{ch}^q(\Omc)$.  As in
  the proof of Theorem~\ref{thm:singlelin}, we introduce a fresh unary
  relation symbol $P_{\mn{db}}$, extend $\mn{ch}^q_\Omc(D)$ to a
  database $D'$ by adding the fact $P_{\mn{db}}(c)$ for every
  $c \in \mn{adom}(D)$, and obtain the CQ $q'(\bar x)$ from $q$ by
  adding the atom $P_{\mn{db}}(x)$ for every answer variable $x$. Note
  that since $q$ is acyclic and free\=/connex acyclic, so is~$q'$.  By
  Lemma~\ref{lem:pseudochase} $Q(D)=q'(D')$ and thus we can use an
  existing procedure as a black box for enumerating $q'(D')$ in
  \cdlin \cite{bagan-enum-cdlin}.

The argument for Point~(2) is identical, based on
Proposition~\ref{prop:allTestingCompleteUpper}.
\end{proof} 

\lemmaenumerationlowerbound*
\begin{proof}
  Let $Q \in (\class{ELI},\class{CQ})$ satisfy the conditions from
  Theorem~\ref{lemma:enumeration-lower-bound}. Assume that enumerating
  complete answers to $Q$ is in \dlc. Let the Boolean OMQ $Q'$
  be obtained from $Q$ by quantifying all answer variables. Then $Q'$
  satisfies the same conditions, that is, it is not acyclic, non-empty
  and self\=/join free. It is even weakly acyclic since acyclicity and
  weak acyclicity coincide for Boolean CQs.  Moreover, single-testing
  for $Q'$ is in linear time because given an \Sbf-database $D$, we
  can check whether $() \in Q'(D)$ by enumerating $Q$ on $D$, but
  accepting after the first ouput and rejecting if there was no
  output. It thus follows from
  Theorem~\ref{lemma:single-testing-lower-bound} that the triangle
  conjecture fails. The same argument works for minimal partial answers
  and for minimal partial answers with multiple wildcards.
\end{proof}

\subsection{Lower Bound for Proposition~\ref{prop:allTestingCompleteUpper}}
\label{app:lowerCQ}

We prove the following counterpart of
Proposition~\ref{prop:allTestingCompleteUpper}. Note that it is not
subsumed by Theorem~\ref{thm:alltestinglower} because, there, the
arity of relation symbols is at most~2.  As we are working without
ontologies here, our proof follows closely the lines of corresponding
lower bounds for enumeration given in
\cite{bagan-enum-cdlin,BraultBaron,berkholz-enum-tutorial}.
%In fact,
%the constructions of the databases are exactly identical, we only use
%them in a slightly different way as we are interested in all-testing
%instead of enumeration. Also note that our queries are only
%free-connex acyclic while the queries in
%\cite{bagan-enum-cdlin,BraultBaron,berkholz-enum-tutorial} are
%both acyclic and free-connex acyclic.
For the same reason,
we rely on the additional assumption
stating that $(k+1,k)$\=/hyperclique problem cannot be solved in time $O(n^k)$.

The $(k+1,k)$\=/\emph{hyperclique problem} is as follows.
Given a $k$\=/uniform hypergraph with $n$ vertices, i.e. every hyperedge consists of exactly $k$ vertices,
decide whether it contains a hyperclique of size $k{+}1$, i.e. a set of $k+1$ vertices where every subset of size $k$ is a~hyperedge.
The hyperclique conjecture
states that for all $k \geq 3$ solving  the
$(k+1,k)$\=/hyperclique problem requires $n^{k+1 - o(1)}$ time \cite{lincoln-soda-grain-complexity}.

\begin{lemma}
    \label{lem:all-testing-cq-lower}
    Let $q(\bar{x})$ be a self\=/join free CQ that is not free-connex acyclic.
    % such $\hat{q}(\bar{x}) \leftarrow q(\bar{x}) \land R(\bar{x})$, where $R$ is a fresh relation symbol,
    % be a self\=/join CQ that is not acyclic.
    %
    Then all-testing for $q$ is not in \dlc unless 
    % If there is an algorithm that takes a database $D$ and
    % after a preprocessing in time $O(||D||)$ allows for membership
    % testing in constant time then
    one of the following holds:
    \begin{enumerate}
      \item the triangle conjecture fails;
        %\item %a triangle in graph $G=([n],E)$ can be detected in time $O(||G||)$;
        %\item % the product of two $n\times n$ Boolean matrices can be
          % computed in time $O(n^2)$;
          \item Boolean $n \times n$ matrices can be multiplied
  in time $O(n^2)$;
        \item $(k+1,k)$\=/hyperclique problem can be solved in time $O(n^k)$.

    \end{enumerate}
\end{lemma}
To prove Lemma~\ref{lem:all-testing-cq-lower},
we make use of the following well-known characterization
of acyclicity.
\begin{theorem}[\cite{beeri-acyclic}]
\label{thm:Beeri}
    A CQ $q$ is acyclic iff it satisfies the following properties:
    \begin{enumerate}
    \item $q$ is conformal, i. e. for every clique of the Gaifman
      graph $G_q$ of $q$ there exists an atom that contains all
      variables in the clique;
    \item $q$ is chordal, i. e. every cycle of length at least $4$ in
      $G_q$ has a chord.  That is, $G_q$ contains an edge
      that is not part of the cycle but connects two vertices of the
      cycle.
    \end{enumerate}
\end{theorem}

We now prove Lemma~\ref{lem:all-testing-cq-lower}.
Let $q$ be as in the lemma and let us assume that all-testing for $q$ is in \dlc.
We show that one of Points~(1)-(3) applies.

Recall that $\hat{q}(\bar{x})$ is the CQ $q(\bar{x})$ with the additional atom $\hat{R}(\bar{x})$
where $\hat{R}$ is a fresh relation symbol.
Since $q$ is not free\=/connex acyclic, $\hat{q}$ is not acyclic.
Hence, $\hat{q}$ is not conformal or not chordal.

\paragraph{Not chordal} If $\hat{q}$ is not chordal, then there is a cordless cycle $z_1, z_2, \dots, z_m$ of length $m \geq 4$ in the Gaifman graph of $\hat{q}$.
%such that if for some $i < j$ there is an edge $(z_i,z_j)$ in the
%Gaifman graph then $j = i+1$. {\color{blue}still don't understand this.}
Moreover, since there is an atom $\hat{R}(\bar{x})$ in $\hat{q}$, there are no more than two answer variables in this cycle.
%Let $Z = \{z_1, \dots, z_m\}$.

\paragraph{At most one answer variable in the cycle.}
If there is no more than one answer variable in the cycle, then
also CQ $q$ is not chordal, as the new edges
in the Gaifman graph of $\hat{q}$ are only those between answer variables.
Hence, we can use the construction for the ``not chordal'' case form Section 6.2. in \cite{berkholz-enum-tutorial}.
Given an undirected graph $G$, it constructs in time $O(n)$ a database $D$ such that $q(D) \neq \emptyset$
if and only if $G$ has a triangle.

Moreover, a careful analysis of the construction reveals that, in
fact, we only need to test $O(n)$ different tuples
$\bar{a} \in \dom(D)^{|\bar{x}|}$ to decide whether
$q(D) \neq \emptyset$, and we can compute the set of those tuples in
time $O(n)$.  Thus, assuming that all-testing is in \dlc, we can
decide whether $q(D) \neq \emptyset$ in time $O(n)$.  Indeed, the database
$D$ can be constructed in time $O(n)$.  Thus, the preprocessing phase
can be done in $O(n)$ time and testing can be carried out in
time $O(n)$.  This gives overall running time $O(n)$, and
disproves the triangle conjecture.

\paragraph{Two answer variables in the cycle.}
If there are two answer variables in the cycle, then without loss of
generality we can assume that they are $x_1 = z_1$ and $x_2 = z_m$.
Indeed, since there is an atom $\hat{R}(\bar{x})$ in $\hat{q}$ and the cycle is chordless,
the answer variables have to be two consecutive vertices in the cycle.
Thus, to obtain $x_1 = z_1$ and $x_2 = z_m$ we can simply rename some variables in $q$.
%{\color{blue}Formulate more carefully? It is not totally clear how/why this is
%  w.l.o.g. when $x_1$ and $x_2$ could also be non-adjacent.}
Now we have two possibilities. Either there is an atom %$R'(\bar{w})$
 in
$q$ that contains both $x_1$ and $x_2$
or there is no such atom in~$q$.

If there is no such atom, then we can use the construction for the
``acyclic but not free\=/connex acyclic query'' case form Section
6.1. in \cite{berkholz-enum-tutorial}. In the terminology of
  \cite{berkholz-enum-tutorial}, applying the construction requires
  that there is a `bad path' in $q$. This is the case when $x_1,x_2$ do
  not co-occur in an atom in $q$.  Given two $n{\times}n$ Boolean
matrices $M_1, M_2$, the construction creates in time $O(n^2)$ a
database $D$ such that the set $q(D)$ projected to the first two
coordinates is the set $M_1M_2$.

Notice that the construction in \cite{berkholz-enum-tutorial} is used
for queries that are not only free-connex acyclic, but also acyclic.
However, acyclicity is only used to guarantee the existence of a
bad path and not in the construction of the database $D$.
% nor the shape of tuples in
% $q(D)$. 
It is easy to verify that the following claim is valid also
in our setting.

\begin{claim}
  For every pair $(i,j) \in [n]^2$ we have $(i,j) \in M_1 M_2$ if and
  only if there is a corresponding answer $\bar{a}_{i,j}$ in
  $q(D)$ where the tuple $\bar{a}_{i,j}$ can be computed from $(i,j)$ 
  in constant time.
\end{claim}

Hence, to compute $M_1 M_2$ we only need to test
$\bar{a}_{i,j} \in q(D)$ for all $(i,j) \in [n]^2$.  If all-testing is
in \dlc, we can thus compute $M_1M_2$ in time $O(n^2)$.
Indeed, the database $D$ can be constructed in time $O(n^2)$.  Thus,
the preprocessing phase can be done in $O(n^2)$ time and testing can
be carried out in time $O(n^2)$.  This gives overall running time
$O(n^2)$.

In the case that $x_1$ and $x_2$ are in an atom $R'(\bar{w})$ in $q$,
{%\color{orange}
we slightly modify the construction by adding to the database $D$ also
the $R'$\=/facts that are total on $x_1,x_2$ and use the unique
constants for the remaining variables.  Since the cycle is chordless,
no other variables from the cycle appear in $R'(\bar{w})$.}  Thus,
every such set of facts is of size $n^2$ and the database can be
constructed in time $O(n^2)$.

\paragraph{Not conformal}
Now assume that $\hat{q}$ is not conformal.
Then there are $k$ and a clique $z_1, \dots, z_{k+1}$
such that the clique is not covered by and atom and every proper subset of $\{z_1, \dots, z_{k+1}\}$ is.
Indeed, it is enough to take a minimal clique that is not covered by an atom.
% Since $\hat{R}(\bar{x})$ is an atom in $\hat{q}$, we have exactly $0 \leq m' \leq k$ answer variables in the clique.
Let $Z = \{z_1, \dots, z_{k+1}\}$.
%Without loss of generality, let those be $x_1, \dots, x_{m'}$.

Again, we have two cases.  Either every proper subset of $Z$ can be
covered by an atom that is not $\hat{R}(\bar{x})$ or there is a proper
subset of $Z$ such that the only atom that covers it
is~$\hat{R}(\bar{x})$.

In the former case, $q$ is not conformal and we can apply the
construction for the ``not conformal'' case form Section 6.2. in
\cite{berkholz-enum-tutorial}.  Given a $k$\=/uniform
hypergraph $G$, it creates in time $O(n^k)$ a database~$D$ such that
$q(D) \neq \emptyset$ if and only if $G$ contains a hyperclique of
size $k{+}1$.
A careful analysis of the construction reveals that, in fact,
we only need to test $O(n^k)$ different tuples $\bar{a} \in \dom(D)^{|\bar{x}|}$
to decide whether $q(D) \neq \emptyset$ and that we can compute the set $W$ of those tuples in time $O(n^k)$.

If there is an algorithm for all-testing in \dlc, we can thus
solve the hyperclique problem $O(n^k)$.  Indeed, database $D$ can be
constructed in time $O(n^k)$, the preprocessing phase can be done in
$O(n^k)$ time, and testing all tuples from $W$ can be carried out in
time $O(n^k)$.  This gives the overall running time $O(n^k)$.

Finally, for the case where there is a proper subset of $Z$ such that
the only atom that covers it is~$\hat{R}(\bar{x})$ {%\color{orange}
we follow the same
construction but using $\hat{q}$ instead of $q$. Then we remove all
$\hat{R}$\=/facts from the database $D$ and adjust the set of tuples
$W$ so that it is consistent with the removed facts. } It is easy to
see that this adaptation can be done in time $O(n^k)$ and does not
change the results of the test.  This ends the proof.

\subsection{Proof of Theorem~\ref{thm:lower-bound-enum-eli}}

There are several characterizations of when an acyclic CQ is 
free\=/connex \cite{berkholz-enum-tutorial}. 
% SHORT
% Without giving
% details, we mention that this is the case if and only if
% it has a generalized hypertree decomposition of width one in which the
% variables from $\bar x$ form a connected subtree. 
A characterization
that we use in what follows is via bad paths.  A
\emph{bad path} in a CQ $q$ is a sequence of variables
$y_1, \dots, y_n$, $n \geq 3$, such that $y_1$ and $y_n$ are distinct
answer variables, $y_2, \dots, y_{n-1}$ are quantified variables,
and $\{y_i,y_{i+1}\}$ is an edge in the Gaifman graph of $q$ while
$\{y_1,y_n\}$ is not. It was shown in
\cite{bagan-enum-cdlin} that an acyclic CQ is free\=/connex if and
only if it has no bad path, see also \cite{berkholz-enum-tutorial}.

\newcommand{\ansSet}{\textit{Ans}}

\thmlowerboundenumeli*
%
% \begin{theorem}
%   \label{thm:lower-bound-enum-eli}
%   Let $Q = (\Omc, \Sbf, q) \in (\class{ELI},\class{CQ})$ be acyclic,
%   but not free-connex, non\=/empty, self-join free, and connected.
%   Then enumeration of complete answers to $Q$ is not in CD$\circ$Lin
%   unless spBMM is possible in time $O(|M_1| + |M_2| + |M_1M_2|)$, and
%   the same is true for least partial answers.
%\end{theorem}
%
%We start presenting the lower bounds with a lemma that connects the complexities of
%enumerating answers and of enumerating least partial answers.
%We say that a path in an OMQ $Q= (\Omc, \Sbf, q)$, i.e.~sequence of neighbouring variables $y_1, \dots, y_m$ in $q$, uses only relations
%from schema if for every $1 \leq i\leq m$ and every unary atom $U(y_i) \in q$ we have that $U \in \Sbf$
%and for every $1 \leq i,j\leq m$ and every binary atom $r(y_i,y_j) \in q$ we have that $r \in \Sbf$.
%
% \begin{theorem}
%     \label{thm:lower-bound-enum-eli}
%     Let $Q = (\Omc, \Sbf, q) \in (\class{ELI},\class{CQ})$ be acyclic,
%     but not free-connex, non\=/empty, self-join free, and connected.
%     Then enumerating (least partial) answers of
%     $Q$ %, nor enumerating least partial answers of $Q$,
%     is not in CD$\circ$Lin unless % one can multiply
%     there is an algorithm that given two Boolean matrices $M_1,M2$
%     returns $M_1M_2$ in time linear in $|M_1| + |M_2| + |M_1M_2|$.
% \end{theorem}
%
{ %\color{blue}
Before we prove the theorem, let us recall some basic relations
between the sets of complete answers, minimal partial answers, and
minimal partial answers with multiple wildcards.

\begin{claim}
    \label{lemma:comparring-sizes-of-answersets}
    For every OMQ $Q \in (\class{G}, \class{CQ})$ there is a constant $K$ such that 
    for every database $D$ holds
    \[
    |Q(D)| \leq |Q(D)^\ast| \leq |Q(D)^{\Wmc}| \leq K |Q(D)^\ast|.
    \]
\end{claim}

Every complete answer is also a minimal partial answer (with multiple wildcards),
thus  $Q(D) \subseteq Q(D)^{\ast}$
and $Q(D) \subseteq Q(D)^{\Wmc}$.
Hence, the first inequality.
%Thus, $Q(D) \subseteq \ansSet$.
The remaining inequalities are a consequence of Lemma~\ref{lemma:cone-1}.
By Point (1) of Lemma~\ref{lemma:cone-1}, we have that $|Q(D)^{\Wmc}| \leq K |Q(D)^{\ast}|$,
where $K$ is some constant that depends only on $Q$. % the number of answer variables.
Indeed, every minimal partial answer with multiple wildcards is in a cone of some minimal partial answer.
Since every cone has no more than $k^{2^k} $ elements, where $k$ is the number of answer variables, the last inequality holds.
Finally, by Point (2) and Point (3) of Lemma~\ref{lemma:cone-1} there is an injective function from $Q(D)^\ast$ to $Q(D)^\Wmc$
and, thus, the middle inequality holds. 
}

%\begin{proof}[Proof of Lemma~\ref{lemma:lower-bound-enum-eli}]
 Let
$Q = (\Omc, \Sbf, q) \in (\class{ELI},\class{CQ})$ be an OMQ
    %such that $q(x_1, \dots, x_k)$ is a connected, acyclic, but not free\=/connex conjunctive query.
as in the theorem. 
To prove Theorem~\ref{thm:lower-bound-enum-eli}  we show that if given a database $D$ we can enumerate 
% (complete or least partial) answers to $Q$%
%\ansSet
any of the sets $Q(D)$, $Q(D)^\ast$, or $Q(D)^{\Wmc}$ in \dlc, then spBMM is possible in
time $\Omc(|M_1|+|M_2|+|M_1M_2|)$.

Assume that we are given Boolean matrices $M_1$ and $M_2$ of size
$n \times n$. Recall that in sparse Boolean 
matrix multiplication (spBMM), $M_1$ and $M_2$ are given as lists 
of pairs $(a,b)$ with $M_i(a,b)=1$ and also $M_1M_2$ is output as such 
a list. 
For a pair $(a,b)$ by $\Pi_{\textit{col}}((a,b))$ we denote the column 
and by $\Pi_{\textit{row}}((a,b))$ the row.

We use $M_1$ and $M_2$ to construct an \Sbf-database $D$ such
that enumerating $Q(D)$ or $Q(D)^\ast$ in \dlc allows us to
construct $M_1M_2$ within the desired time bound.

As a preliminary, we argue that we can w.l.o.g.\ assume $M_1$ and
$M_2$ to satisfy a certain condition that shall prove to be useful in
what follows.  Let $P_M$ denote the set of \emph{productive indices} in
$n \times n$ matrix~$M$,
i.e.~$P_M = \{c \in [n] \mid \exists a\,  M(a,c)=1 \text{ or
} M(c,a)=1\}$. Now the condition is:
\begin{itemize}

\item[($*$)] for all $c \in P_M$, there are $a_1,a_2 \in [n]$ with
  $M(c,a_1)=1$ and $M(a_2,c)=1$.

\end{itemize}
In fact, we can construct from $M_1,M_2$ in time $\Omc(|M_1| + |M_2|)$
two $(n +2)\times (n+2)$ matrices $\hat{M}_1,\hat{M}_2$ that satisfy
($*$) and such that $M_k(a,b)=1$ if and only if $\hat{M}_k(a+2,b+2)=1$
for $1 < a,b \leq n$ and $k= {1,2,3}$ where
$\hat{M}_3 = \hat{M}_1 \cdot \hat{M}_2$.
To construct $\hat M_k$, set
$$
\begin{array}{r@{\;}c@{\;}ll}
\hat{M}_k(a+2,b+2) &=& {M}_k(a,b) & \text{ for } 0 < a,b \leq n\\
\hat{M}_k(1,1) &=&\hat{M}_k(2,2) =1 &\\
\hat{M}_1(i+2,2) &=& 0 & \text{ for } 1 \leq i \leq n  \\
\hat{M}_2(1,i+2)&=& 0 & \text{ for } 1 \leq i \leq n                 
\end{array}
$$
Note that  quite a few entries in the first and second column and
row are yet undefined. Independently of how we define them,
    $\begin{array}{r c l}
    \hat{M}_3(a+2,b+2) & = &  \hat{M}_1(a+2,1)\hat{M}_2(1,b+2) \\
    &+& \hat{M}_1(a+2,2)\hat{M}_2(2,b+2) \\
    &+& \sum_{c=0}^{n}\hat{M}_1(a+2,c+2)\hat{M}_2(c+2,b+2)\\
    & = & \sum_{c=0}^{n}{M}_1(a,c){M}_2(c,b)\\
    & = & M_1 M_2 (a,b) = M_3(a,b).\\
    \end{array}$
    
    \noindent
    Hence, we can use the remaining undefined positions to satisfy
    ($*$). It clearly suffices to add at most
    $2*|P_{M_1} \cup P_{M_2}|$ ones to each matrix and thus
    $|\hat M_1|+|\hat M_2| \in O(|M_1|+|M_2|)$. We can thus
    use $\hat M_1$ and $\hat M_2$ in place of $M_1$ and $M_2$.
    
    %Now, since $q$ is acyclic, but not free\=/connex

    Recall that $q$ is acyclic, but not free-connex acyclic and that, as
    discussed in the preliminaries, this implies that $q$ contains a
    bad path, i.e.~a sequence of variables a sequence of variables
    $x_1, y_1, \dots, y_\ell, x_2$, $\ell \geq 1$, with $x_1,x_2$
    answer variables and $y_1,\dots,y_\ell$ quantified variables such
    that any two consecutive variables form an edge in the Gaifman
    graph of $q$, but not two non-consecutive variables do. Since $q$
    is acyclic and connected, the undirected graph $G^{\mn{var}}_q$ is
    a disjoint union of trees (as the connectedness can be `via' a
    constant). We can impose a direction on these trees. For the tree
    that contains $x_1$, we do this by choosing $x_1$ as the root
    and call a variable in the tree a successor of another variable
    if it is further away from $x_1$, and likewise for descendants and
    ancestors. For the other trees we do the same, choosing a root
    as follows. If $q$ contains an atom that contains a variable $x$
    from the tree and a constant, then choose such an $x$ as the
    root. Otherwise choose the root arbitrarily.

    To present the reduction in a more transparent way, we make
    the following simplifying assumptions:
    \begin{itemize}

    \item[(Dir)] if $y$ is a successor of $x$ in $G^{\mn{var}}_q$,
      then any binary atom in $q$ that involves $x$ and $y$ is
      directed towards $y$, that is, it takes the form $R(x,y)$;
      moreover, any binary atom in $q$ that involves a variable $x$
      and a constant $c$ is directed towards $c$ if $x$ is reachable
      from $x_1$ in $G^{\mn{var}}_q$ and away from $c$ otherwise;
      
    \item[(Mult)] $q$ contain no non-reflexive multi-edges, that is,
      no two distinct non-reflexive binary atom that mention the same
      terms.
      
    \end{itemize}
    It is not difficult to get rid of these assumptions. In fact, this
    can be done as follows. Given $q$, first re-orient the edges and
    drop all edges but one from non-reflexive multi-edges, obtaining a
    CQ $q'$ that satisfies (Dir) and (Mult). Then construct the
    \Sbf-database $D$ as described below and obtain from it another
    \Sbf-database $\hat D$ by re-orienting edges back to their
    original direction and adding back multi-edges. Finally, show that
    every homomorphism from $q$ to $\mn{ch}_\Omc(\hat D)$ is
    also a homomorphism from $q'$ to $\mn{ch}_\Omc(D)$ and
    vice versa. We omit details.
        
    We now turn towards the announced construction of the
    \Sbf-database $D$.  Let
    $R_0(x_1,y_1), R_1(y_1,y_2), \dots, R_{\ell-1}(y_{\ell-1},
    y_{\ell}),R_\ell(y_\ell,x_2)$ be the (unique) atoms in $q$ that give rise
    to the bad path from $x_1$ to~$x_2$.
    \begin{lemma}
      \label{lem:headpathinS}
      $R_0,\dots,R_{\ell} \in \Sbf$.
    \end{lemma}
    \begin{proof}
      Assume to the contrary of what is to be shown that
      $R_i \not\in \Sbf$ for some $R_i$. Since $Q$ is non-empty, there
      is an \Sbf-database $D_0$ such that $Q(D_0) \neq \emptyset$ and
      thus there is a homomorphism $h_0$ from $q$ to
      $\mn{ch}_{\Omc}(D_0)$.  Since \Omc is an~$\ELI$ ontology, the
      restriction of $\mn{ch}_{\Omc}(D_0)$ to $\mn{adom}(D_0)$ does
      not contain any binary facts that use a relation symbol
      $R \notin \Sbf$. Since $R_i \notin \Sbf$, $h_0$ must thus map at
      least one of the variables on the bad path to a null.  Since
      \Omc is an \ELI ontology, $\mn{ch}_{\Omc}(D_0)$ takes the shape
      of $D_0$ with trees without multi-edges and self loops attached
      to every constant.  Since $h_0$ maps some variable on the bad
      path to a null and both $x_1$ and $x_2$ to $\mn{adom}(D_0)$,
      there must be two distinct atoms in the bad path that are
      mapped to a fact that crosses from the database part of
      $\mn{adom}(D_0)$, into the null part. This is impossible since
      $q$ is self-join free.
    \end{proof}

    Let $C=\mn{con}(q)$ if $\mn{con}(q)$ is non-empty and
    $C=\{ \bot \}$ otherwise, $\bot$ a fresh constant.
    Moreover, let $C_\Sbf=\{ \bot_R \mid R \in \Sbf \text{ binary} \}$
    where each $\bot_R$ is a fresh constant. The active domain of the
    database $D$ that we aim to construct
    is %$\mn{adom}(D) = \{1, \dots, n\} \cup \{c_1, \dots, c_m, \bot\} \cup \{\bot_r\}_{r \Sbf'}$ %of size $n+l+1$
        $$\mn{adom}(D) = P_{{M}_1} \cup P_{{M}_2} \cup C \cup C_\Sbf
        $$ %of size $n+l+1$
        and $D$ contains the following facts:
        \begin{itemize}

        \item $A(c)$ for every unary $A \in \Sbf$ and every $c \in \mn{adom}(D)$;

        \item For relation symbol $R_0$:
          \begin{itemize}
          \item 
          $R_0(a,b)$ for all $a,b \in [n]$ such that
          ${M}_1(a,b)=1$;
          
        \item $R_0(c,c)$ for all $c \in C \cup C_\Sbf$;
          \end{itemize}

        \item For relation symbol $R_\ell$:
          \begin{itemize}
          \item 
          $R_\ell(a,b)$ for all $a,b \in [n]$ such that
          ${M}_2(a,b)=1$;
          
        \item $R_\ell(c,c)$ for all
          $c \in C \cup C_\Sbf$;
          \end{itemize}

        \item For each relation symbol $R_i$, $0 < i < \ell$: \\
          $R_i(a,a)$ for all $a \in [n]$;

        \item For each relation symbol
          $R \in \Sbf \setminus \{ R_0,\dots,R_\ell \}$:
          \begin{itemize}

          \item $R(\bot_R,a)$ for all $a \in [n]$;

          \item $R(a,c)$, and $R(c,c')$ for all $a \in [n]$, and
            $c,c' \in C$.

          \item $R(a,a)$ for all $a \in [n]$ if $q$ contains reflexive atom
            $R(z,z)$. %\footnote{\color{blue}or remove during preprocessing?}

          \end{itemize}

        \end{itemize}
        It should be clear that $D$ can be constructed in time linear
        in $|{M}_1|+|{M}_2|$. Moreover, $D$ satisfies the
        \emph{completeness
          property} that every $c \in \mn{adom}(D)$ has an incoming and an
          outgoing edge for every $R \in \Sbf$ since
%
        % \item[(II)] the only ingoing edges to constants in the set
        %   $[n]$ originate from constants in the set $[n] \cup C_\Sbf$.
%
        % \end{itemize}
        %
        $M_1$ and $M_2$ satisfy
        Condition~($*$) above.  We next show how answers to $Q$ on $D$
        are related to one entries in the matrix product $M_1M_2$.
    \begin{lemma}\label{lemma:lower-bound-lemma-eli-answers-to-entires-correspondence}
      Let $(a,b) \in [n]^2$. Then ${M}_1{M}_2(a,b)=1$ if and only if
      there is a complete answer $(a_1,a_2, \dots, a_k) \in Q(D)$ such that
      $a_1 = a$ and $a_2 = b$. 
      %The same is true for least partial
      %answers
      %$(a_1,a_2, \dots, a_k) \in \ansSet$.%\in Q(D)^\ast$.
    \end{lemma}
    \begin{proof}
%      Since $Q(D) \subseteq \ansSet$, it suffices to prove the ``if''
%      direction for least partial answers and the ``only'' if
%      direction for complete answers.
%
%      \smallskip
      
      ``if''.  Let $(a_1,a_2, \dots, a_k) \in Q(D)$ so that
      $a_1 = a$ and $a_2 = b$.  Then there is a homomorphism $h$ from
      $q$ to $\mn{ch}_\Omc(D)$ such that $h(x_1)=a$ and $h(x_2)=b$.
      Consider the image of the path
      $R_0(x_1,y_1), \dots,R_\ell(y_\ell,x_2)$ under $h$ in
      $\mn{ch}_\Omc(D)$. We first argue that this image is actually
      contained in $D$.

      Recall the \Omc is an \ELI ontology and $\mn{ch}_{\Omc}(D)$
      takes the shape of $D$ with trees without multi-edges and self
      loops attached to each constant. It follows that none of
      $h(y_1),\dots,h(y_\ell)$ is a null as then the image of the path
      would start at $a$ (in the database part), then cross to the
      null part and eventually cross back to end at $b$. But this is
      not possible because of the described form of
      $\mn{ch}_{\Omc}(D)$ and since the relation symbols
      $R_0,\dots,R_\ell$ are all distinct (as $q$ is self join free).
      Since the chase with an \ELI ontology does not add any
      facts $R(c_1,c_2)$ with $c_1,c_2 \in \mn{adom}(D)$, the
      image of the path must indeed by in $D$.

      Considering the construction of $D$, it is now easy to see that
      $h(y_1)= \cdots = h(y_\ell) =u$ for some $u \in \{1,\dots,n\}$.
      In particular, note that the only $R_0$- and $R_\ell$-edges that
      originate at a constant from $[n]$ also end at a constant from
      $n$, and that the relation symbols $R_i$, $1 \leq i < \ell$,
      only occur in reflexive facts. Thus, $D$ contains atoms
      $R_0(a,u),R_{\ell}(u,b)$, which by the definition of $D$ implies
      that ${M}_1(a,u)=1$ and ${M}_2(u,b)=1$. Consequently,
      $M_1M_2(a,b)=1$.

      \smallskip ``only if''. Let ${M}_1{M}_2(a,b)=1$.  Then we
      have ${M}_1(a,u)=1$ and ${M}_2(u,b)=1$ for some $u$ with
      $1 \leq u \leq n$.  Moreover, since $Q$ is not empty, there is a
      database $D_0$ and a homomorphism $h_0$ from $q$ to
      $\mn{ch}_{\Omc}(D_0)$ such that
      $h_0(x_1), h_0(x_2) \in \dom(D_0)$. 

      We construct a homomorphism $h$ from $q$ to $\mn{ch}_\Omc(D)$.
      Start with setting $h(c) = c$ for every constant $c$ in $q$,
      $h(x_1) = a$, $h(x_2) = b,$ and $h(y_v) = u$ for
      $1 \leq v \leq \ell$.  We also choose a $c_\bot \in C$ and set
      $h(x) = c_\bot$ for all $x$ such that $h(x)$ was not previously
      defined and $h_0(x) \in \dom(D_0)$.

      It can be verified that $h$ is a homomorphism from $q|_h$ to
      $\mn{ch}_\Omc(D)$ where $q|_h$ is the restriction of $q$ to the
      domain of $h$. We first argue that $h$ respects all binary atoms
      in $q|_h$ and the consider the unary atoms.

      First note that $h$ respects all binary atoms that involve any
      pair of variables $(x_1,y_1)$, $(y_i,y_{i+1})$ for
      $1 \leq i < \ell$, and $(y_\ell,x_2)$ by construction of $D$ and
      due to (\mn{Mult}). Reflexive binary atoms on a variable from
      the bad path are also defined by definition of $D$. There are
      no atoms that involve other combinations of variables on the
      bad path such as $(x,y_2)$ since $G^{\mn{var}}_q$ is a disjoint
      union of trees. By definition of $D$, it is also clear that all
      binary atoms are respected that involve only constants and
      non-path variables for which $h$ is defined. Finally, atoms that
      involve a path variable and a constant or a non-path variable
      are respected by (\mn{Dir}) and construction of $D$.

      Now for the unary atoms $A(x)$. If $A \in \Sbf$, then $A(x)$ is
      satisfied since $A(c) \in D$ for every $c \in \mn{adom}(D)$.
      Thus let $A \notin \Sbf$. The completeness property of $D$
      implies that $(D_0,h_0(x)) \preceq (D,h(x))$, no matter what
      $h(x)$ is. It thus follows from Lemma~\ref{lem:sim} that
      $D \models (\Omc,\Sbf,A(x))(h(x))$.
      
    We next extend $h$ to all remaining variables in $q$. Assume that
    $h(t)$ is defined, $q$ contains an atom that uses term $t$ and
    variable $y$, and $h(y)$ is not yet defined.  First assume that
    $t$ is a constant $c$. Since $h(y)$ is undefined, $h_0(y)$ is a
    null. Due to (\mn{Dir}) and (\mn{Mult}),
    $q$ contains a single atom $R(y,c)$ that mentions $y$ and $c$.
    Due to the self-join freeness of $q$, $h_0$ maps all variables in
    the subtree of $G^{\mn{var}}_q$ rooted at $y$ to a null, the exact
    argument for this is similar to the one used in the ``if''
    direction.  It follows that the CQ $p(z)$, which is the
    restriction of $q$ to the variables in the subtree of
    $G^{\mn{var}}_q$ rooted at $y$, extended with atom $R(y,z)$,
    contains no multi-edges and reflexive loops, thus is an ELIQ. The
    completeness property of $D$ implies that
    $(D_0,h_0(z)) \preceq (D,h(z))$, no matter what $h(z)$ is. It thus
    follows from Lemma~\ref{lem:sim} that
    $D \models (\Omc,\Sbf,p(x))(h(x))$, and consequently there is a
    homomorphism $g$ from $p(x)$ to $\mn{ch}_\Omc(D)$ with
    $g(x)=h(x)$. We can extend $h$ to the variables in $p$ by setting
    $h := h \cup g$.
    
    Now assume that $t$ is a variable $z$.  Then $h_0(z)$ is not a
    null and $h_0(y)$ is a null.  Moreover, $y$ is a successor of $z$
    in the direction that we have imposed on $G^{\mn{var}}_q$ because
    $h_0(x_1)$ is non-null (as $h_0$ witnesses non-emptiness w.r.t.\
    \emph{complete} answers).  Due to the self-join freeness of $q$,
    $h_0$ maps all variables in the subtree of $G^{\mn{var}}_q$ rooted
    at $y$ to a null. It follows that the restriction $p(z)$ of $q$ to
    $z$ and the variables in the subtree of $G^{\mn{var}}_q$ rooted at
    $y$ contains no multi-edges and reflexive loops, thus is an ELIQ.
    We can proceed as in the case where $t$ is a constant.
\end{proof}

     To finish the proof of Theorem~\ref{thm:lower-bound-enum-eli} we use the following lemma and
     Claim~\ref{claim:lower-bound-lemma-eli-it-bound-answers}.
%     a previous observation that $|\ansSet| \leq k^k |Q(D)^{\ast}|$.
     
%     If $(a_1, a_2, \dots, a_k)$ belongs to $Q(D)$ then for $3 \leq i \leq k$ we have that
%     \begin{itemize}
%         \item if $a_1 \in \{1, \dots n\}$ then $a_i \in \{a_1,a_2, \bot, \ast, c_1, \dots, c_m\}$,
%         \item if $a_1 = \ast$ then there is $a \in \{1, \dots, n\}$ such that $a_i \in \{a,a_2, \bot, \ast, c_1, \dots, c_m\}$.
%     \end{itemize}
     %
%     \\[2mm]
%     {\sc Claim.} If $(a_1, a_2, \dots, a_k) \in Q(D)^{\ast}$ then there is
%     an  $a \in \{1, \dots, n\}$ such that $a_i \in \{a, a_1, a_2, \bot, \ast, c_1, \dots, c_m\}$   for $3 \leq i \leq k$.
%     \\[2mm]
     %
     %
%    \begin{enumerate}
%        \setcounter{enumi}{2}
%        \item \label{item:lower-bound-lemma-eli-it-bound-answers}
%        The set $Q(D)^{\ast}$ has no more than $\Omc(|{M}_1| + |{M}_2| + |{M}_3|)$ elements.
%    \end{enumerate}
    %
    \begin{lemma}
        \label{claim:lower-bound-lemma-eli-it-bound-answers}
        The number of minimal partial answers $Q(D)^{\ast}$ is bounded by $\Omc(|M_1|+|M_2|+|M_1M_2|)$.
    \end{lemma}
    %
%    Note that the same is then true for the complete answers as
%    $Q(D) \subseteq Q(D)^\ast$.
    We can now prove
    Theorem~\ref{thm:lower-bound-enum-eli} as follows.  Assume that
    there is an algorithm that enumerates
    % (complete or least partial) answers to $Q$
    set $Q(D)$ (or $Q(D)^\ast$ or $Q(D)^\Wmc$) in \dlc. Given matrices $M_1,M_2$, we can
    construct the database $D$ in time $\Omc(|{M}_1|+ |{M}_2|)$ and
    use and enumerate the answers in time  
%    for $Q(D)$ or for $Q^*(D)$ to 
 %  generate $Q(D)$.
    %to generate $Q(D)$ or $Q^\ast(D)$.
%    the set \ansSet.
      $\Omc(|M_1|+|M_2|+|M_1M_2|)$, cf.~Lemma~\ref{claim:lower-bound-lemma-eli-it-bound-answers}.
    Finally, by
    Lemma~\ref{lemma:lower-bound-lemma-eli-answers-to-entires-correspondence},
    the projection
    %$Q(D)$ or $Q^\ast(D)$
%    \ansSet
    to the first
    two positions of the enumerated complete answers gives us a list representation of
    $M_1M_2$ in total time $\Omc(|{M}_1| + |{M}_2| + |M_1M_2|)$.
     
    A central ingredient to the proof of
    Lemma~\ref{claim:lower-bound-lemma-eli-it-bound-answers} is
    the following.
    \begin{lemma}
      \label{lem:sizetechnical}
         Let $h$ be a homomorphism from $q$ to
         $\mn{ch}_\Omc(D)$. If there are two distinct answer variables
         $x_i$, $x_j$ such that $h(x_i),h(x_j) \in [n]$, then $\{x_i,x_j\}=\{x_1,x_2\}$.
       \end{lemma}
       \begin{proof}
        Assume that $x_i,x_j$ are distinct
       answer variables such that $h(x_i),h(x_j) \in [n]$. Due to the way
       in which we have oriented the edges in $G^{\mn{var}}_q$, it
       must be the case that 
       \begin{enumerate}

       \item[(1)] $x_i$ and $x_j$ are in the same tree in
         $G^{\mn{var}}_q$ as $x_1$ and $x_2$.

       % \item each of $x_i$ and $x_j$ is in a tree in $G^{\mn{var}}_q$
       %   that is not connected to a constant, that is, there is no
       %   atom in $q$ that contains a variable from the tree and a
       %   constant.
         
       \end{enumerate}
       In fact, assume to the contrary that $x_i$ (or $x_j$) is in a
       different tree in $G^{\mn{var}}_q$ than $x_1$ and $x_2$.  Then
       since $q$ is connected, $q$ must contain a binary atom that
       contains a constant $c$ and a variable $x$ from the tree.  By
       the way we have imposed a direction on the trees in
       $G^{\mn{var}}_q$ and by (\mn{Dir}), we even find such a $c$ and
       an $x$ such that all edges in the tree are oriented away from
       $x$. Also by (\mn{Dir}), the edge between $c$ and $x$ is
       oriented towards $x$. Thus, there is a directed path in $q$
       from $c$ to $x_i$.  But there is no such (directed!) path from
       $c$ to a constant in $[n]$ in $D$. Since \Omc is an \ELI
       ontology, the same holds for $\mn{ch}_\Omc(D)$. We have thus
       shown ($*$).

       We next observe that since $q$ is connected and since all
       edges in $G^{\mn{var}}_q$ are oriented away from the roots
       of the trees, one of the following must hold:
       \begin{enumerate}

       \item[(2)] there is a variable $y$ such that $x_i$ and $x_j$ are
         both reachable in $G^{\mn{var}}_q$ from $y$ on a directed
         path;

       \item[(3)] $x_i$ and $x_j$ are in different trees in
         $G^{\mn{var}}_q$, but there is a constant $c$ such that
         $x_i$ and $x_j$ are reachable in $q$ from $c$ on a
         (not necessarily directed) path.

       \end{enumerate}
       Point~(3), however, is ruled out by Point~(1). We are thus left
       with Point~(2).  Consider the combined path from $x_i$ to $y$
       to $x_j$.  Since $q$ is self-join free, all relation symbols on
       this path are distinct.  Now take the image under $h$ of the
       path in $\mn{ch}_\Omc(D)$.  It starts and ends at a constant in
       $[n]$. Moreover, (i)~the edge from $x_i$ to the next variable
       $y$ on the path must be directed towards $x_i$ or (ii)~the edge
       from the variable $z$ before $x_j$ on the path to $x_j$ must be
       directed towards $x_j$. We only consider Case~(i) as Case~(ii)
       is completely symmetric.

       In Case~(i), the construction of $D$ implies that $h(y)$ can
       only be a null or of the form $\bot_R$. First assume the
       former. Then the $h$-image of the path crosses from the
       database part of $\mn{ch}_\Omc(D)$ to the null part
       and eventually back to reach $h(x_j) \in [n]$. But since \Omc
       is an \ELI ontology, the shape of $\mn{ch}_\Omc(D)$ is that
       of $D$ with trees without multi-edges and reflexive loops
       attached to each constant from $\mn{adom}(D)$. Thus,
       there is no path in $\mn{ch}_\Omc(D)$ of the described form
       in which no relation symbol occurs twice.

       Now assume that $h(y)$ is of the form $\bot_R$. Then the
       $h$-image of the path starts at a constant of $[n]$, then
       reaches a constant of the form $\bot_R$, and eventually again a
       constant of the form $[n]$. But by construction $D$ contains no
       path of this form on which no relation symbol occurs twice.
       Since \Omc is an \ELI ontology, the same is true for
       $\mn{ch}_\Omc(D)$.
       \end{proof}
       
     \begin{proof}[Proof of 
       Lemma~\ref{claim:lower-bound-lemma-eli-it-bound-answers}.]
%       {\color{blue} Here we want retain $Q(D)^\ast$. Still this can be made simpler.
%       Just consider $\bar{a}$ where at most one answer is in $[n]$ and the rest. See commented out part below.}
       
       Let $\bar{a} = (a_1, a_2, \dots, a_k) \in Q(D)^{\ast}$ be a partial answer. By Lemma~\ref{lem:sizetechnical},
       either $\bar{a} \cap [n] = \{a_1,a_2\}$ or 
       there is at most one value from $[n]$ in $\bar{a}$, i.e.~$\{a\} = \bar{a} \cap [n]$ for some $a$.
       
       In the former case, we observe that for every variable $x_i$, $1\leq i \leq k$, such that  $a_i \neq a$, 
       the value $a_i$ can be chosen from at most $|C| + |C_\Sbf| +1$ possibilities.
       Since $k$ is a fixed constant depending only on the query, we have $O(|M_1|+|M_2|)$ answers of the former kind.
       Similarly, by Lemma~\ref{lemma:lower-bound-lemma-eli-answers-to-entires-correspondence}
       we can conclude that the number of answers satisfying the latter case is bounded by $O(|M_1M_2|)$.
       Indeed, the number of possible pairs $(a_1,a_2)$ is $|M_1M_2|$ and the remaining undefined values in $\bar{a}$
       are chosen from a set of size $O(|q|)$. This ends the proof the lemma.

%       We analyse separately the least partial answers
%       $(a_1, a_2, \dots, a_k) \in Q(D)^{\ast}$ such that
%       at most one $a_i \in [n]$ and those for which at least two different indices $i,j \in [k]$
%       satisfy $\{a_i,a_j \subseteq [n]\}$.
%       
%       For answers of the former kind, there are
%       only $|C| + |C_\Sbf|$ choices for each $a_j$ where $j \neq i$.
%       Since $k$ is a constant, the number of answers of this kind is bounded by $O(|M_1|+|M_2|)$.
%       
%       For the answers of the latter kind, Lemma~\ref{lem:sizetechnical} implies that
%       $\{i,j\} = \{1,2\}$ and for each $a_l$ where $l>2$ we have at most $|C| + |C_\Sbf|$.
%       Hence, by Lemma~\ref{lemma:lower-bound-lemma-eli-answers-to-entires-correspondence}
%       we can conclude that the number of such answers is bounded by $O(|M_1M_2|)$.
       
       We analyse separately the minimal partial answers
       $(a_1, a_2, \dots, a_k) \in Q(D)^{\ast}$ such that
       $a_i \in [n]$ for some $i > 2$ and those for which $a_i \in [n]$ only if $i \in \{1,2\}$.

       For answers of the former kind, Lemma~\ref{lem:sizetechnical} implies that
       $a_j \notin [n]$ for all $j \in
       \{1,\dots,i-1,i+1,k\}$. Consequently, the number of answers of
       this kind is bounded by $O(|M_1|+|M_2|)$.
       
       Now for answers of the latter kind. In this case, there are
       only constantly many choices for $a_3,\dots,a_k$ as we
       have at most $|C| + |C_\Sbf|$ different possibilities for indices in
       $i \in \{3,\dots,k\}$, and $k$ is a constant. We may further
       distinguish between answers where at most one of $a_1,a_2$ is
       from $[n]$ and where both are. Due to the limited number of
       possible choices for $a_3,\dots,a_k$, the number of answers of
       the former kind is bounded by $O(|M_1|+|M_2|)$ and 
       Lemma~\ref{lemma:lower-bound-lemma-eli-answers-to-entires-correspondence}
       implies that the number of answers of the latter kind is
       bounded by $O(|M_1M_2|)$. This ends the proof of the conditional lower bound.

    \end{proof}

\medskip

We give some remarks on (im)possible generalizations of
Theorem~\ref{thm:lower-bound-enum-eli}. Replacing $\class{ELI}$ with
$\class{G}$ in Theorem~\ref{thm:lower-bound-enum-eli} would allow us
to also remove `self-join free' from
Theorem~\ref{thm:lower-bound-enum-eli}, as in
Example~\ref{ex:sjf}. But there are CQs (with self-joins) that are
acyclic and not free-connex acyclic, yet their answers can be
enumerated in \cdlin (without ontologies)
\cite{berkholz-enum-tutorial}.

\connectedisneeded*
\begin{proof}
      Let
      \[
      \begin{array}{rcl}
          \Omc &=& \{A_1(x) \rightarrow A_2(x), B_1(x) \rightarrow B_2(x),\\
               & &   C_1(x) \rightarrow C_2(x)\} \\
          \Sbf &=& \{A_1,B_1, C_1, L, R\} \\
          q(x_1,z_1,x_2,y_2,z_2) &=& L(x_1,y_1) \land R(y_1,z_1) \, \land \\
               & & A_1(x_1)  \land B_1(y_1) \land C_1(z_1) \, \land \\
               & & A_2(x_2) \land B_2(y_2) \land C_2(z_2).
      \end{array}
      \]
      Note that $Q=(\Omc,\Sbf,q)$ satisfies all properties listed in
      Proposition~\ref{prop:connectedisneeded}. We argue that
    complete answers to $Q$ can be enumerated in \dlc.  The
    idea is that the additional answer variables $x_2,y_2,z_2$ enlarge
    the answer set, which gives additional computational power in the
    enumeration phase of the algorithm. 

    To make this precise, let $D$ be an $\Sbf$\=/database and, by
    slight abuse of notation, $A = \{c \in \dom(D) \mid A(c) \in D\}$
    for every unary relation symbol $A$, and $|A|$ the size of set
    $A$.  It is clear
    that % $\mn{ch}_{\Omc}(D)$ can be computed in linear
      % time since the chase only needs to add unary relation symbols
      % to
      % existing constants. Now, it is easy to see that
      $Q(D) = p(D) \times A_1 \times B_1 \times C_1$ where
      $$p(x_1, z_1)=L(x_1,y_1) \land R(y_1,z_1) \land A_1(x_1)  \land B_1(y_1) \land C_1(z_1).$$

      Since $q$ is acyclic, we can find an answer
      $(a_1,c_1,a_2,b_2,c_2) \in Q(D)$ in linear time during the
      preprocessing phase. In fact, we can compute $\mn{ch}_{\Omc}(D)$
      in linear time since the chase only needs to add unary relation
      symbols to existing constants, and then use the standard
      Yannakakis algorithm to find an answer to $q$ on
      $\mn{ch}_{\Omc}(D)$.
      This actually reveals the $|A_1|\cdot|B_1|\cdot |C_1|$ distinct
      answers in the set
      $U = \{ a_1 \}\times\{c_1\} \times A_1 \times B_1 \times C_1$.
      Now, while the algorithm enumerates the answers in $U$ with
      constant delay, it can
      in parallel compute the set $p(D)$ by simply
      checking all possible triples
      $(a_1,b_1,b_2) \in A_1 \times B_1 \times C_1$. After finishing
      the enumeration of $U$, it is then easy to enumerate the
      remaining answers with constant delay.
\end{proof}

\subsection{Proof of Theorem~\ref{thm:alltestinglower}}

\thmalltestinglower*

Let $Q(\bar x) = (\Omc, \Sbf, q)$, and let
%
% Recall that the CQ $q(\bar{x})$ in $Q = (\Omc, \Sbf, q)$ is a
% conjunctive query using only binary and unary relations.
$\hat{q}(\bar{x})$ be the CQ obtained from $q$ by adding an atom
$R(\bar{x})$ with $R$  a fresh relation symbol.  Since $Q$ is not
acyclic free\=/connex $\hat{q}$ is not acyclic.

% \begin{theorem}[BEERI ET AL. 1983]
%     \label{thm:not-acyclic-query-characterisation}
%     A CQ $q$ is acyclic iff it satisfies the following properties:
%     \begin{enumerate}
%     \item $q$ is conformal, i. e. for every clique of the Gaifman
%       graph $G_q$ of $q$ there exists an atom that contains all
%       variables in the clique;
%     \item $q$ is chordal, i. e. every cycle of length at least $4$ in
%       $G_q$ has a chord.  That is, $G_q$ contains an edge
%       that is not part of the cycle but connects two vertices of the
%       cycle.
%     \end{enumerate}
% \end{theorem}

We observe the following consequence of Theorem~\ref{thm:Beeri}.

\begin{claim}
    There is a chordless cycle $y_1, y_2, \dots, y_\ell$ in $G_{\hat{q}}$
    that has no repeated vertices and no more than two answer variables.
\end{claim}

Since $\hat{q}$ is not acyclic, it is not conformal or not chordal.
If it is not conformal, then there is a clique not contained in an
atom.  This clique has to have at least 3 vertices, since every clique
of size 2 is an edge, and at least one quantified variable, as the
clique of all answer variables is induced by the atom $R(\bar{x})$.  To
obtain the cycle, choose a triangle that contains this quantified
variable.

If $\hat{q}$ is not chordal, then there is a chordless cycle of length at least $4$.
Every such cycle has no more than two answer variables, as there is an edge between every two answer variables.

\medskip

We distinguish three cases according to the number of answer
variables in the cycle.

\paragraph{No answer variables in the cycle}
If no variable in $y_1, \dots, y_\ell$ is an answer variable, 
then there is a cycle in $q$ and, thus, $q$ is not weakly acyclic.

Clearly, an algorithm for all-testing for $Q$ in \dlc can be
used to implement single-testing for $Q$ in linear time. Thus,
Theorem\ref{lemma:single-testing-lower-bound} implies that the
triangle conjecture fails.

\paragraph{One answer variable in the cycle}
Without loss of generality, we can assume that $x_1 = y_1$
is the single answer variable in the cycle $y_1, \dots, y_\ell$.
Hence the query $q_0$ obtained from $q$ by making $x_1$ a quantified 
variable is not weakly acyclic.

If all-testing for $Q$ is possible in \dlc, then
single-testing for $Q_0:=(\Omc, \Sbf, q_0)$ is possible in
linear time: given a candidate answer $\bar{a}$ to $Q_0$ on an
\Sbf-database $D$, we use the all-testing algorithm to decide whether
$b\bar{a} \in Q(D)$ for any $b \in \dom(D)$, in overall linear time.
By Theorem~\ref{lemma:single-testing-lower-bound}, this implies
that the triangle conjecture fails.

\paragraph{Two answer variables in the cycle}
We assume that $x_1 = y_1$ and $x_2 = y_\ell$ are the two answer
variables in the cycle $y_1, \dots, y_\ell$, it shall be clear how 
the proof can be adapted if the answer variables are located elsewhere
on the cycle.
We show that if all-testing for $Q$ is in \dlc, then given the
two $n\times n$ Boolean matrices $M_1,M_2$ we can compute $M_1M_2$
in time~$\Omc(n^2)$.

We use a construction similar to that in the proof of
Theorem~\ref{lemma:enumeration-lower-bound}.
However, since $Q$ may not be
acyclic the construction needs to be suitably adapted.

As shown in proof of Theorem~\ref{thm:lower-bound-enum-eli}
we can assume that both $M_1$ and $M_2$
have $1$ in every row and in every column.
Also, we make the following observation.

\begin{claim}
    If $R(x,y)$ is an atom in $q$ where $x$ and $y$ are variables in the cycle $y_1, \dots, y_\ell$
    then $R \in \Sbf$.
\end{claim}

The proof follows the same pattern as the proof of Lemma~\ref{lem:headpathinS}.

\smallskip

%Let $C=\mn{con}(q)$. % if $\mn{con}(q)$ is non-empty and
%$C=\{1\}$ otherwise.
%Moreover, let $C_\Sbf=\{ \bot_R \mid R \in \Sbf \text{ binary} \}$
%where each $\bot_R$ is a fresh constant.
The active domain of the
database $D$ that we aim to construct
is %$\mn{adom}(D) = \{1, \dots, n\} \cup \{c_1, \dots, c_m, \bot\} \cup \{\bot_r\}_{r \Sbf'}$ %of size $n+l+1$
$\mn{adom}(D) = [n] \cup \mn{con}(q)
$ %of size $n+l+1$
and $D$ contains the following facts:
\begin{itemize}
    
    \item $A(c)$ for every unary $A \in \Sbf$ and every $c \in \mn{adom}(D)$;
    
    \item for every relation symbol $R \in \Sbf$ we add facts:
    \begin{itemize}
        \item $R(a,b)$ for all $a,b \in [n]$ such that
        ${M}_1(a,b)=1$,
        \item[] \qquad if $R(x_1, y_2)$ is an atom in $q$;
        
        \item $R(b,a)$ for all $a,b \in [n]$ such that
        ${M}_1(a,b)=1$,
        \item[] \qquad if $R(y_2, x_1)$ is an atom in $q$;
        
        \item $R(a,b)$ for all $a,b \in [n]$ such that
        ${M}_2(a,b)=1$,
        \item[] \qquad if $R(y_{\ell -1}, x_2)$ is an atom in $q$;
        
        \item $R(b,a)$ for all $a,b \in [n]$ such that
        ${M}_1(a,b)=1$,
        \item[] \qquad if $R(x_2, y_{\ell- 1})$ is an atom in $q$;
        
        \item $R(a,a)$ for all $a\in [n]$,
        \item[] \qquad if $R(y_i, y_{i+1})$ is an atom in $q$ with $2 \leq i \leq \ell-2$;    
        
        \item $R(a,b)$ for all $a,b\in \dom(D)$,
        \item[] \qquad if $R(x,y)$ is an atom in $q$ and $x$ and $y$ are not
          both in the cycle;
        
        \item $R(c,c)$ for every $c \in \mn{con}(q)$.
    \end{itemize}
\end{itemize}

The database $D$ can be constructed in time $\Omc(n^2)$.
Moreover, for every relation symbol $R \in \Sbf$ and every
element $c \in \dom(D)$ there are $a,b$ such that $R(a,c), R(c,b) \in D$.
Thus, as before, we will be able to extend a partial homomorphism
from $q$ to $D$
to the atoms whose relation symbols are not present in the schema.

To end the reduction it is enough to show the following.
\begin{claim}
    \label{claim:att-testin-answer-matrix-correspondence}
    Let $(a,b) \in [n]^2$.
Then    $(a,b, 1, \dots, 1) \in Q(D)$ iff ${M}_1M_2(a,b)=1$.
\end{claim}
    
Indeed, assume that all-testing complete answers for $Q$ is in
\dlc. Then given two matrices $M_1,M_2$, we can compute
$M_1M_2$ in time $\Omc(n^2)$ by first computing the database $D$ in
time $\Omc(n^2)$, then executing the preprocessing phase of the
all-testing algorithm in time linear in $||D||$, which is
$\Omc(n^2)$, and finally testing the tuple $(a,b, 1, \dots, 1)$ for
every pair $(a,b) \in [n]^2$.  $M_1M_2$ is the set of all pairs
$(a,b)$ for which the test
succeeded. % Claim~\ref{claim:att-testin-answer-matrix-correspondence}
% assures the correctness and completeness of the algorithm.
% The overall
% running time of the algorithm is $\Omc(n^2)$.

All that is left is to prove the claim.
The ``$\Rightarrow$'' is proven the same way as in Lemma~\ref{lemma:lower-bound-lemma-eli-answers-to-entires-correspondence}.
For the other direction we cannot use the proof of Lemma~\ref{lemma:lower-bound-lemma-eli-answers-to-entires-correspondence} as the CQ may not be acyclic.
We thus do the following. First, we define $h$ only on $\bar x$ so
that $h(\bar{x}) = (a,b, 1, \dots, 1)$.
Then we argue as in the proof of Theorem~\ref{lemma:single-testing-lower-bound}
that it can be extended to a homomorphism $h$ from $q$ to
$\mn{ch}_{\Omc}(D)$. Details are left to the reader.

For the cases of minimal partial answers and minimal partial answers with multiple wildcards,
recall that $Q(D) \subseteq Q(D)^{\ast}$ and $Q(D) \subseteq Q(D)^{\Wmc}$. 
Thus, both the construction and the algorithm work with no modifications.

\section{Proofs for Section~\ref{sect:LPAsingleWildcardUpper}}

\thmlowerGpartial*
\begin{proof}
  We start with showing the result for $(\class{G},\class{CQ})$ in
  place of $(\class{ELI},\class{CQ})$. Let $R$ be a binary relation
  symbol and $\Sbf=\{R\}$. We may view an
  undirected graph $G=(V,E)$ as the $\Sbf$-database
  $$D_G=\{ R(v,v'), R(v',v) \mid \{v,v'\} \in E \}.$$  Consider the
  OMQ $Q(x) = (\Omc,\Sbf,q) \in (\class{G},\class{CQ})$
  where \Omc contains the TGD
  $$
  \begin{array}{r@{\;}c@{\;}l}
    R(x_1,x_2) &\rightarrow& \exists y_1 \exists y_2 \, R\{x_1,y_1\} \wedge 
                             R\{y_1,y_2\} \wedge R\{y_2,x_1\}
    % \\[0.5mm]
    % R(x_2,x_1) &\rightarrow& \exist y_1 \exists y_2 \, R\{x_1,y_1\} \wedge 
    % R\{y_1,y_2\} \wedge R\{y_2,x_2\}
  \end{array}
  $$
  with $R\{x,y\}$ an abbreviation for  $R(x,y) \wedge R(y,x)$ and with
  $$
  q(x,y,z,u)=R\{x,y\} \wedge 
    R\{y,z\} \wedge R\{z,u\}.
  $$
  Let $G=(V,E)$ be an undirected graph.  Then $(v,\ast,\ast,v)$ is a
  partial answer to $Q$ on $D_G$ for every $v \in V$, but not
  necessarily a minimal partial answer. In fact, it is a minimal partial
  answer if and only if $v$ is not part of a triangle in $G$. It
  follows that all-testing for $Q$ is not in \dlc unless the
  triangle hypothesis fails: to decide whether a given graph $G=(V,E)$
  contains a triangle, we can construct $D_G$ in time linear in
  $||E||$, then execute the preprocessing phase of all-testing for
  $Q$, and then iterate over all $v \in V$ and test in constant time
  whether $(v,\ast,\ast, v) \in Q(D)^\ast$. We answer `no' if this is
  the case for all $v$ and `yes' otherwise. The same arguments work in
  the multi-wildcard case, with $(v,\ast,\ast, v)$ replaced by
  $(v,\ast_1,\ast_1,v)$.

  The challenge in improving the construction to $\class{ELI}$ is
  that \ELI TGDs cannot introduce a triangle that consists of nulls.
  The solution is to construct $q$, \Omc, and $D$  in a more careful
  way. Let us start with $q$, which is now 
  $$
  q(x_1,x_2,x_3,x_4,x_5)=R(x_1,x_2) \wedge R(x_2,x_3) \wedge
  R(x_4,x_3) \wedge R(x_5,x_4),
  $$
  that is, it is a path of length~4 rather than of length~3 and
  the direction of the edges is chosen carefully. %  so that the cyclic
  % can be contracted into a path that can be generated by the ontology.
  %In fact,
  We choose \Omc to contain the \ELI TGD
  $$
  \begin{array}{r@{\;}c@{\;}l}
    R(x_1,x_2) &\rightarrow& \exists y_2 \exists y_3 \, R(x_1,y_2) \wedge 
    R(y_2,y_3).
  \end{array}
  $$
  Similarly to before, $(v,\ast,\ast,\ast,v)$ is a partial answer to
  $Q$ on any $\Sbf$-database $D$ such that every constant in $D$ has
  an outgoing $R$-edge (in the multi-wildcard case, we use $(v,\ast_1,
  \ast_2,\ast_1,v)$. It remains to modify $D_G$ so that $q(D_G)$
  contains a tuple $(c_1,c_2,c_3,c_4,c_5)$ with $c_1=c_5$ if and only
  if $G=(V,E)$ contains a triangle (as this makes at least one of the
  partial answers be not a minimal partial answer). This is achieved by
  constructing $D_G$ so that $\mn{adom}(D_G)=V \times [4]$ and it
  contains the following facts,
  % :
  % %
  % \begin{itemize}
%
  % \item $R((v,2),(v,3))$ for each $v \in V$;
%
  % \item $R((v,1),(v',2)), R((v,2),(v,3)), R((v',4),(v,3)), R((v',1),(v,4)). 
    %
  % \end{itemize}
  %
 for each $\{v,v'\} \in E$:
 $$
 R((v,1),(v',2)), R((v,2),(v,3)), R((v',4),(v,3)), R((v',1),(v,4)). 
  $$
  To finish the proof, it suffices to observe the following.
  \\[2mm]
  {\sc Claim.} $q(D_G)$ contains a tuple $(c_1,c_2,c_3,c_4,c_5)$ with
  $c_1=c_5$ iff $G$ contains a triangle.
  \\[2mm]
  For the ``if'' direction, assume that $v_1,v_2,v_3$ is a triangle in
  $G$. Then the tuple
  $$
    ((v_1,1),(v_2,2),(v_2,3),(v_3,4),(v_1,1))
  $$
  is in $q(D_G)$.

  \smallskip ``only if''. Assume that
  $q(D_G)$ contains a tuple $(c_1,c_2,c_3,c_4,c_5)$ with
  $c_1=c_5$. Consider $c_3$.  Only constants of the form
  $(v_1,2)$ and $(v_1,4)$ have both incoming and outgoing
  $R$-edges in $D_G$, so
  $c_2$ must be of one of these forms. Assume that
  $c_2=(v_1,2)$, the case
  $c_2=(v_1,4)$ is symmetric.  The construction of
  $D_G$ yields the following. We must have
  $c_3=(v_1,3)$.  The fact %s $R((v,2),(v,3)) \in D_G$ and
  $R(c_4,(v_1,3))
  \in D_G$ was introduced due to some edge $\{ v_2,v_1\} \in
  E$ and $c_4$ must be of the form $(v_2,4)$. The fact $R(c_1,(v_1,2))
  \in D_G$ was introduced due to some edge $\{v_3,v_1\} \in E$ and
  $c_1$ must take the form
  $(v_3,1)$; likewise, the fact
  $R(c_5,(v_2,4))$ was introduced due to some edge $\{v_4,v_1\} \in
  E$ and $c_5$ must take the form $(v_4,1)$. Since
  $c_1=c_5$, it follows that $v_3=v_4$. Thus, the nodes
  $v_1,v_2,v_3$ constitute a triangle in $G$.
\end{proof}

\subsection{Missing Details for Proof of
  Theorem~\ref{thm:upperGpartial}}
\label{app:singlewild}

In the main part of the paper, we have declared the goal to be the
enumeration of
$q_0(\mn{ch}^{q_0}_\Omc(D))^{\ast}_\Nbf$. Here, we actually prove
something slightly more general (based on exactly the algorithm
presented in the main part), as follows.
\begin{proposition}
  \label{prop:enumwildcards}
  For every CQ $q(\bar x)$ that is acyclic and free\=/connex ayclic,
  enumerating the answers $q(D)^{\ast}_\Nbf$ is in \dlc for
  databases $D$ and sets of nulls $N \subseteq \mn{adom}(D)$ such that
  $D$ is chase-like with witness $D'_{1},\dots,D'_{n}$ where
  $|\mn{adom}(D'_{i})|$ does not depend on $D$ for $1 \leq i \leq n$.
\end{proposition}
To prove Proposition~\ref{prop:enumwildcards}, let $q_0(\bar x)$
%\leftarrow \phi(\bar x,\bar
%y)$
be a CQ that is acyclic and free\=/connex acyclic,
$D_0$ a database over the same schema as
$q_0$ that is chase-like with witness
$D'_{1},\dots,D'_{n}$ where
$|\mn{adom}(D'_{i})|$ does not depend on $D$ for $1 \leq i \leq
n$.  Note that this notation is completely compatible with the one
used in the main part of the paper, only that there
$D_0$ is the concrete chase-like database $\mn{ch}^{q_0}_\Omc(D)$.

We first note that we can assume w.l.o.g.\ that the tuple $\bar x$ has
no repeated variables and that $q_0$ contains no constants. In fact,
answer enumeration in the general case can be reduced in time linear in
$||D_0||$ to answer enumeration in this restricted case. We give
the reduction for removing constants. For every atom $R(\bar t)$ in
$q_0$ where $\bar t$ contains at least one constant, introduce a fresh
relation symbol $R_{\bar t}$ whose arity is the number of positions in
$\bar t$ that have a variable. Then replace $R(\bar t)$ in $q_0$ with
$R_{\bar t}(\bar y)$ where $\bar y$ is obtained from $\bar t$ by removing
all constants. Furthermore, for each fact $R(\bar c)$ in $D_0$ such that each
constant $c$ that occurs in some position $i$ of $\bar t$ also occurs
in position $i$ of $\bar c$, add to $D_0$ the fact $R_{\bar t}(\bar c')$ where
$\bar c'$ is obtained from $\bar c$ by removing the constants in
the positions where $\bar t$ has a constant.

We can also assume $q_0$ to be connected.}
{
This is the appendix of the conference version of the paper, which is 
subject to space restrictions. {\color{red}Full proofs can be found 
  in \cite{arxive}.}

\section{Connected Queries in
  Section~\ref{sect:LPAsingleWildcardUpper}}
\label{app:connected}

% We provide some of the missing details for the enumeration algorithm 
% for minimal partial answers with a single wildcard. We start with 
% arguing that, as claimed, we can assume the query to be connected. 
We argue that, when enumerating minimal partial answers with a single
wildcard,
we can indeed assume the query to be connected. 
}
For assume  that we have
found an enumeration algorithm for connected CQs that runs in \dlc. We
can then enumerate $q_0(D_0)^\ast_\Nbf$ in \dlc when $q_0$ has maximal
connected components $p_0,\dots,p_k$, $k >0$, in the following way. We
first do preprocessing for all $p_0,\dots,p_k$. We then start an
algorithm that enumerates the answers to $p_0$. After the first answer
was found, it calls the enumeration algorithm for $p_1$, which upon
finding an answer calls the enumeration algorithm for $p_2$, and so
on. Only when the innermost algorithm found an answer to $p_k$, the
answers are combined and output as an answer to $Q$. Note the
algorithms for $p_1,\dots,p_k$ have to start from scratch multiple
times which is problematic since the data structures computed in the
preprocessing phase are modified in the enumeration phase and we
cannot repeat preprocessing because that would introduce a linear time
delay into the enumeration phase. An easy solution is as follows.
When first enumerating the answers to $p_k$, we store all of them in
the form of a linked list. When we need to enumerate the answers to
$p_k$ again, we can just use that list without any preprocessing.  We
do the same for the subqueries $p_{k-1} \cup p_k$,
$p_{k-2} \cup p_{k-1} \cup p_k$, and so on, which fixes the
problem. The above argument requires a polynomial amount of memory
during the enumeration phase. There is, however, an alternative
approach that avoids this. Our algorithm is such that the data structure $S$ computed in
the preprocessing phase is modified in the enumeration phase,
resulting in a data structure $S'$. However, $S'$ is such that it
could have been used in place of $S$ after the preprocessing phase
without affecting the output of enumeration.  This means that the
preprocessing can simply be skipped before restarting the enumeration
algorithm for a connected subquery.

\ifbool{arxive}{
Recall that from CQ $q_0(\bar x)$ and database $D_0$, we have to
construct a CQ $q_1(\bar x)$ and a database $D_1$ that satisfy
Conditions~(i) to~(iv) from the main part of the paper.
\ifbool{arxive}{Note that
Condition~(ii) implies that $D_1$ is chase-like with a
witness $D'_{1},\dots,D'_{n}$ (not necessarily the same as for $D_0$)
such that $|\mn{adom}(D'_{i})|$ does not depend on $D$ for
$1 \leq i \leq n$. } The construction of $q_1$ and $D_1$ has been used many times in the context of
enumerating answers to conjunctive queries (without ontologies)
with constant delay.  We give a rouch sketch and refer the
interested reader to \cite{berkholz-enum-tutorial} for a very clear
exposition of the details. Exploiting that $q_0$ is acyclic and
free\=/connex acyclic, it is possible to first construct a generalized
hypertree decomposition (GHD) of $q_0$ of width~1 in which the answer
variables constitute a connected (`connex') prefix. Then a bottom-up
pass over the GHD is made, manipulating both $q_0$ and $D_0$ in a
synchronized way.  In particular, one introduces a fresh relation
symbol for each node of the GHD and duplicates facts in the database
accordingly, thus achieving self-join freeness. Moreover, one achieves
the progress condition by dropping facts from the database that
violate it. This also turns the GHD into a join tree. Finally, the
quantified variables can simply be dropped because the progress
condition has already been achieved. If $q_0(D_0)$ is empty, then we
find this out during the construction of $q_1$ and $D_1$ as then all
facts from the database that use the relation symbol from the root of
the GHD have been dropped. We then return `end of enumeration'. The
construction only needs time linear in $||D_0||$.

\lemprecomputetreeslistslinear*
\begin{proof}
  Recall that $D_1$ is chase-like with witness
  $D'_{1},\dots,D'_{n}$. We iterate over $i=1,\dots,n$. For each $i$,
  set $S= \mn{adom}(D'_i) \setminus N$ and iterate over all pairs
  $(q,g)$ with $q$ a subtree of $q_1$ and
  $g:\mn{var}(q) \rightarrow S \cup \{\ast\}$.  Note
  that there are only constantly many such pairs: since the OMQ is
  fixed, there are only constantly many subtrees $q$ of $q_1$, and the
  arity of relations in schema \Sbf is bounded by a constant; since
  $S$ is a guarded set in $D'_i$ (by definition of
  chase-likeness), it follows that the
  cardinality of $S$ is bounded by a constant as well,
  implying that the same is true for the number of maps $g$.

  We disregard pairs % $(q,g)$ with $g(\bar z) \neq h(\bar z)$ and
  % pairs $(q,g)$
  such that Conditions~(1) or~(2) of progress trees is violated.
  Note that this can clearly be checked in constant time, and
  that Condition~(4)
  is in fact satisfied by choice of $g$. We then check whether 
  there is a homomorphism $\widehat g$ from $q$ to $D'_i$ such that for all
  $x,y \in \mn{var}(q)$,
  \begin{itemize}

  \item $\widehat g(x) \in N$ if $g(x) = \ast$ and $\widehat g(x)=g(x)$ otherwise;

  \item $\widehat g(x) = h(y)$ iff $g(x) = g(y)$.
    
  \end{itemize}
  If this is not the case, we disregard $(q,g)$. Otherwise,
  Condition~(3) of progress trees is satisfied. Note that we can check
  the existence of $h$ brute force as there are only constantly many
  potential targets because $|\mn{adom}(D'_{i})|$ does not depend on
  $D$. We then add $(q,g)$ to $\mn{trees}(v,h)$ where $v$ is the root
  of $q$ and $h$ is the restriction of $g$ to the predecessor
  variables in $v$.
  
  Regarding the correctness of this construction, there are two
  important observations. The first is that we really generate all
  progress trees despite considering only homomorphisms $\widehat g$
  into the databases $D'_i$ instead of into $D_1$ as a whole, as
  required by Condition~(3) of progress trees. This is guaranteed by
  the following claim.
%
  % of progress trees are not satisfied. Note
  % that Condition~(4) is in fact satisfied by construction of $(q,g)$.
  % Conditions~(1) and (2) can easily be checked in constant time. The
  % same is true for Condition~(3) since the size of $q$ is bounded by a
  % constant and so is the number of possible targets for the desired
  % homomorphism $h$. The latter follows from the fact that $D_2$ is
  % chase-like w.r.t.\ $N$ with witness $D'_{1}\dots,D'_{n}$ such that
  % $|\mn{adom}(D'_{i})|$ does not depend on $D$ for $1 \leq i \leq n$,
  % and the following claim.
%

  \medskip
  \noindent
  {\sc Claim.}
  Let $(q,g)$ be a progress tree and let $h$
  be a homomorphism from $q$ to $D_1$ such that for all
  $x \in \mn{var}(q)$, $h(x) \in N$ if $g(x) = \ast$ and $h(x)=g(x)$
  otherwise. Then $h$ is a homomorphism from $q$ to $D'_{i}$ for some
  $i \in \{1,\dots,n\}$.

  \medskip \emph{Proof of claim.}  Let $v_0=R_0(\bar y_0)$ be the root of
  $T_q$. Then clearly there is an $i$ such that
  $R_0(h(\bar y_0)) \in D'_i$. % Recall that the constants in
  % $\mn{adom}(D'_i) \setminus N$ are required to form a guarded set
  % in $D'_i$.
  We show that $h$ is a homomorphism from $q$ to $D'_i$.

  More precisely we prove by induction on the depth of $v$ in $T_q$
  that that for all $v=R(\bar y) \in V_q$, $R(h(\bar y)) \in D'_i$.
  The induction start is trivial by choice of $D'_i$. So let $v \in V_q$ be a non-root
  node. By Condition~(2) of progress trees, $v$ contains a
  predecessor variable $x$ with $g(x) = \ast$. Thus $h(x) \in N$. Let
  $v'=R'(\bar y')$ be the predecessor of $v$ in $T_q$.  By induction
  hypothesis, $R'(h(\bar y')) \in D'_i$. Since by definition of
  `chase-like'
  $\mn{adom}(D'_i) \cap \mn{adom}(D'_j) \cap N = \emptyset$, this and
  $h(x) \in N$ yields $R(h(\bar y)) \in D'_i$, as required. This
  finishes the proof of the claim.

  \medskip

  The second observation is that we must be careful to avoid
  duplicates as the same $(q,g)$ can be constructed for $D_i$, $D'_j$
  with $i \neq j$, when
  $\mn{adom}(D'_i) \setminus N$ and $\mn{adom}(D'_j) \setminus N$
  overlap. To identify duplicates, we use a lookup table that
  stores a Boolean value for every progress tree $(q,g)$, indicating
  whether we have already seen this tree or not. Note that a progress
  tree $(q,g)$ can essentially be represented by a list of constants
  from $\mn{adom}(D_1)$ of constant length, and by the remarks on RAMs
  under the uniform cost model in Section~\ref{sect:prelims} we can
  access and modify such a table in constant time.

  It remains to sort the lists $\mn{trees}(v,h)$, each of which is of
  length $O(|\mn{adom}(\Dmc_1)|)$, into the desired order in
  overall linear time. Recall from Section~\ref{sect:prelims} that
  sorting a list of short lists equipped with a strict weak order is
  possible on a RAM in linear time.  Now, a progress tree $(q,g)$ is
  essentially a short list.  We can use $q$ as the first element of
  the list (only constantly many choices), fix a total order
  $z_1,\dots,z_\ell$ on the variables in $q$, and use
  $g(z_1),\dots,g(z_\ell)$ as the remaining list. A strict weak order
  on progress trees is then defined by `$\prec_{\mn{db}}$'. If we now
  sort the progress trees in $\mn{trees}(v,h)$ viewed as short lists,
  then we attain the required order, that is,
  $(q,g) \prec_{\mn{db}} (q',g')$ implies that $(q,g)$ occurs before
  $(q',g')$ in the list.  There are linearly many lists
  $\mn{trees}(v,h)$ to be sorted, each in linear time, but we still
  attain overall linear time since each progress tree occurs in only
  one list and thus the total number of items to be sorted (across all
  lists) is bounded by $O(||D_1||)$.

\medskip

It remains to prove the
  `moreover' part. Let $v$ and $h$ be relevant.  We have to show that
  there is a progress tree $(q,g)$ with root $v$ such that
  $g(\bar z)=h(\bar z)$.  We build this tree together with a
  homomorphism $h'$ that witnesses Point~(3) of the definition of
  progress trees. We start with $q$ consisting only of the single node
  $v$. Since $v$ and $h$ are relevant, $h$ extends to a homomorphism
  from $v$ to $D_1$. We use this homomorphism both for (the initial)
  $g$ and $h'$. We then exhaustively extend $q$, $g$, and $h'$ as
  follows.  If $v'$ is an atom in $q$ such that $g$ maps at least one
  variable in $v'$ to `$\ast$' and $v''$ is a successor of $v$ in
  $T_1$, then include $v'$ in $q$. By the progress condition, we can
  extend the homomorphism $h'$ to the extended $q$.  Extend $g$
  accordingly, using `$\ast$' in place of null constants.  It can be
  verified that Conditions~(1) to~(4) of progress trees are
  satisfied. In particular, Condition~(4) is: due to the shape of
  chase-like
  databases and since $q_1$ is connected, all constants in the
  range of $g$ are from the guarded set $h(\bar x)$, with $\bar x$ the
  variables in~$v$.
\end{proof}

\lemneveremptylem*
\begin{proof}  Assume to the contrary that some
  list $\mn{trees}(v,h)$, with $v,h$ relevant, becomes empty when
  $\mn{prune}(h')$ is called. Assume that among the progress trees
  removed during this call, $(q,g)$ is a minimal progress tree
  regarding `$\prec_{\mn{db}}$'. Since $(q,g)$ was still on
  $\mn{trees}(v,h)$ before the call $\mn{prune}(h')$, no progress tree
  $(q',g')$ with $(q',g') \prec_{\mn{db}} (q,g)$ was removed in
  any previous pruning step.  To obtain a contradiction, it thus
  suffices to argue that the initial list $\mn{trees}(v,h)$ contains a
  progress tree $(q',g')$ with $(q',g') \prec_{\mn{db}} (q,g)$.

  We extract the tree $(q',g')$ from $h'$. Let $q' \subseteq q_1$ be
  smallest so that $v \in q'$ and if $u \in q'$ and $u'$ is a
  successor of $u$ in $T_1$ such that $h(x) \in N$ for at least one
  predecessor variable in $u'$, then $u' \in q'$. We define $g'$ to
  be $h'$ restricted to the variables in $q'$.

  Note that we must have $q' \subseteq q$. If this is not the case, in
  fact, then by definition of $q'$ and since $q$ is a progress tree,
  there is an atom $u$ in $q$ such that $g$ maps all predecessor
  variables in $u$ to $\mn{adom}(D_1)$ while $g'$ maps at least
  one predecessor variable in $u$ to `$\ast$'. But this implies
  that the test  `$(q,g) \succ_{\mn{db}} (q, h'|_{\mn{var}(q)})$' made
  in $\mn{enum}(h')$ fails, in contradiction to $(q,g)$ being removed
  during that call.

  If $q' \subsetneq q$, then $(q',g') \prec_{\mn{db}} (q,g)$ and we
  are done. Otherwise, however, the test
  `$(q,g) \succ_{\mn{db}} (q, h'|_{\mn{var}(q)})$' made in
  $\mn{enum}(h')$ guarantees that Conditions~(a)-(d) from the
  definition of `$\prec_{\mn{db}}$' are satisfied for $(q',g')$
  and $(q,g)$, thus again $(q',g') \prec_{\mn{db}} (q,g)$.
\end{proof}

We now work towards a proof of
Proposition~\ref{prop:partialalgocorrect}. Every partial answer
$\bar a^{\ast}$ to $q_1$ on $D_1$ can be seen as a map
$h_{\bar a^{\ast}}: \mn{var}(q_1) \rightarrow (\mn{adom}(D_1)
\setminus N ) \cup \{ \ast \}$ in an obvious way and vice versa. Let
$\bar a^{\ast}$ be a partial answer.  We say that a progress tree
$(q,g)$ is
%
% Let $\bar a^{\ast}$ be a
% partial answer to $q_1$ on $D_1$, $v$ % = R(\bar y)\in V_1$
% a non-root
% node, $S$ a non-empty non-null guarded set in $D_1$, and
% $(q,g) \in \mn{trees}(v,S)$. %\{ h_{\bar a^{\ast}}(x)\mid x
% % \in \bar y \text{ and } h_{\bar a^{\ast}}(x) \neq \ast \})$.
% Then we
% say that $(q,g)$ is
\emph{realized} in $\bar a ^{\ast}$ if
$h_{\bar a^{\ast}}(x)=g(x)$ for all $x \in \mn{var}(q)$.
\begin{restatable}{lemma}{lemnopruning}
  \label{lem:nopruning}
  Without pruning, the algorithm outputs the
  %{\color{violet} set of all}
% clu: I perceive this as redundant and would prefer not to have it
  (not
  necessarily minimal) partial answers to $q_1$ on $D_1$, without
    repetition. Moreover, the enumeration order respects `$\prec$', that is, if
  partial answer $\bar a^{\ast}$ is output before partial answer
  $\bar b^{\ast}$, then $\bar b^\ast \not\prec \bar a^{\ast}$.
\end{restatable}
\begin{proof}
  Let $\bar a^{\ast}$ be a partial answer to $q_1$ on $D_1$. % Then
  % there is a homomorphism $h_{\bar a^{\ast}}$ from $q_4$ to $D_4$
  % such that for every $x \in \mn{var}(q_4)$: if $h_{\bar a^{\ast}}(x)$
  % is a database constant, then $\bar a^{\ast}$ has
  % $h_{\bar a^{\ast}}(x)$ in the $x$-positions and if
  % $h_{\bar a^{\ast}}(x)$ is an existential constant, then
  % $\bar a^{\ast}$ has $\ast$ in the $x$-positions.
  We argue that there is a path in the recursion tree generated by the
  initial call $\mn{enum}(\mn{nextat}_{h_0}(v_0))$ that leads
  to $\bar a^{\ast}$ being output.
%  {\color{violet} the next sentence is not clear at all.}
  This path can be identified by choosing, for each call
  $\mn{enum}(v,h)$, a recursive call made during it that identifies
  the successor on the path. Let $\mn{enum}(v,h)$ be such a call, and
  let $h_{\bar a^{\ast}}(v)= R(\bar c)$.  % We distinguish two
  % cases. If $\bar c$ contains no null, then choose the recursive call
  % in the upper \mn{forall} loop associated with $R(\bar
  % c)$. Otherwise, $h_{\bar a^{\ast}}$ identifies a unique
  % $S$-candidate tree $(q,g)$, with $S$ the set of non-null constants
  % in $\bar c$, and we can choose the corresponding recursive call in
  % the lower \mn{forall} loop. Let us make this a bit more precise.
  Using Condition~(2) of progress trees, it can be shown in
  a straightforward way that the algorithm satisfies the following
  invariant:
  \begin{enumerate}

  \item[($*$)] In each call $\mn{enum}(v',h')$, all predecessor
    variables in $v'$ are mapped to database constants.

  \end{enumerate}
  Hence, the same is true for all predecessor variables $\bar z$ in
  $v$.  % There must be at least one such variable since $q_1$ is
  % connected.  Moreover,
  We identify a candidate tree
  $(q,g) \in \mn{trees}(v,h|_{\bar z})$ as outlined after the
  definition of candidate trees. More precisely, choose
  $V_q \subseteq V_1$ to be the smallest set that contains $v$ and
  such that if $u \in V_q$ and $u'$ is a successor of $u$ in $T_1$
  such that $h_{\bar a^\ast}(x) = \ast$ for at least one predecessor
  variable in $u'$, then $u' \in V_q$. This defines a subtree $q$ of
  $q_1$, and for $g$ the restriction of $h_{\bar a^\ast}$ to the
  variables in~$q$, it can be verified that
  $(q,g) \in \mn{trees}(v,h|_{\bar z})$. %  with $g(\bar z)=h(\bar
  % z)$.
  This identifies the recursive call made during $\mn{enum}(v,h)$
  that we follow on the path towards the output of $\bar a^\ast$. It
  can be verified that the map $h$ built up on this path is exactly
  $h_{\bar a^\ast}$, and thus $\bar a^\ast$ is indeed output at the
  end of the path.
  
  Now assume that the algorithm outputs a tuple $\bar a^{\ast}$ and
  let
  $h_{\bar a^{\ast}}: \mn{var}(q_1) \rightarrow (\mn{adom}(D_1)
  \setminus N ) \cup \{ \ast \}$ be the associated map. We argue that
  we can obtain from $h_{\bar a^{\ast}}$ a homomorphism $h$ from $q_1$
  to $D_1$ such that for all $x \in \bar x$,
  (i)~$h_{\bar a^{\ast}}(x) \in \mn{adom}(D_1)$ implies
  $h(x)=h_{\bar a^{\ast}}(x)$ and (ii)~$h_{\bar a^{\ast}}(x) = \ast$
  implies $h(x) \in N$. Consequently, $\bar a^{\ast}$ is a partial
  answer to $q_1$ on $D_1$. Let
  $\mn{enum}(v_0,h_0),\dots,\mn{enum}(v_{\ell-1},h_{\ell-1})$ be the
  sequence of recursive calls that led to the output of
  $\bar a^{\ast}$ (during the last call). Note that
  $v_0,\dots,v_{\ell-1}$ are not necessarily all nodes of $V_1$ due to
  the use of progress trees and \mn{nextat}. We define a sequence of
  partial mappings $h'_0 \subseteq \cdots \subseteq h'_\ell$ from
  $\mn{var}(q_1)$ to $\mn{adom}(D_1)$ that for all
  $x \in \mn{var}(q_1)$ on which they are defined satisfy
  Conditions~(i) and~(ii).  Then $h'_\ell$ will be a total function
  and
  thus the desired
  homomorphism $h$ from $q_1$ to $D_1$. Start with setting
  $h'_0=\emptyset$. Now assume that $h'_i$ is already defined and
  consider the recursive call $\mn{enum}(v_i,h_i)$.  By definition of
  the algorithm, there is a progress tree
  $(q,g) \in \mn{trees}(v_i,h_i|_{\bar z})$, with $\bar z$ the
  predecessor variables of $v_i$, such that $h_{\bar a^{\ast}}(z)=g(z)$ for all
  $z \in \mn{var}(q)$. % Also be definition of the algorithm
  % $h_i$ agrees with $h$ on the joint domain, and thus
  % $h(\bar z)=g(\bar z)$. 
  By Property~(3) of progress trees, there is a homomorphism $g'$ from
  $q$ to $D_1$ such that for all $x \in \mn{var}(q)$, $g'(x) \in N$ if
  $g(x) = \ast$ and $g'(x)=g(x)$ otherwise. We define
  $h'_{i+1} = h'_i \cup g'$. It can be verified that the map
  $h=h'_\ell$ is indeed the desired homomorphism from $q_1$ to $D_1$.
 
\smallskip 
  
We next argue that there are no repetitions. Consider again the
recursion tree generated by the initial call
$\mn{enum}(\mn{nextat}_{h_0}(v_0)$, and consider two paths
in the tree that lead to the output of partial answers $\bar a_1^\ast$
and $\bar a_2^\ast$. Let $\mn{enum}(v,h)$ be the call in which the two
paths diverge, that is, $h_{\bar a_1^*}$ and $h_{\bar a_2^*}$ both agree
with $h$ on the variables on which $h$ is defined, but during the
call $\mn{enum}(v,h)$ the two paths follow recursive calls made for
different $(p_1,g_1),(p_2,g_2)\in \mn{trees}(v,h|_{\bar z})$. We argue
that then, $h_{\bar a_1^*}$ and $h_{\bar a_2^*}$ must also be different,
and thus so are $\bar a_1^\ast$ and $\bar a_2^\ast$. This is clear if
$p_1 =p_2$ as this implies $g_1 \neq g_2$ and $h$ is extended by
setting $h=h \cup g_1$ and $h=h \cup g_2$, respectively. Now assume
that $p_1 \neq p_2$. Then there must be nodes $u,u'$ in $T_{q_1}$ with
$u'$ a successor of $u$ and such that $u$ is part of $T_{p_1}$ and
$T_{p_2}$, but $u'$ is part of $T_{p_1}$, but not part of $T_{p_2}$
(or vice versa). By Condition~(2) of progress trees, $g_1$ maps some
predecessor variable of $u'$ to $\ast$, but $g_2$ does
not. Consequently, the extensions $h=h \cup g_1$ and $h=h \cup g_2$
differ on some variable on which they are both defined.
  
\smallskip 
  
For the `moreover' part, let $\bar a^{\ast}, \bar b^{\ast}$ be partial
answers with $\bar b^{\ast} \prec \bar a^{\ast}$.  Since $\bar a^\ast$
is a partial answer, there must be a homomorphism $h_{\bar a}$ from
$q_1$ to $D_1$ such that for all $x \in \mn{var}(q_1)$,
$h_{\bar a}(x) = h_{\bar a^{\ast}}(x)$ if
$h_{\bar a^{\ast}}(x) \notin N$ and $h_{\bar a}(x) \in N$ if
$h_{\bar a^{\ast}}(x)=\ast$, and analogously for $\bar b^\ast$ and a
homomorphism $h_{\bar b}$.  Since $\bar b^{\ast} \prec \bar a^{\ast}$,
there must further be a variable $x_0$ with $h_{ \bar b^{\ast}}(x_0)$
a database constant and \mbox{$h_{\bar a^{\ast}}(x_0) = \ast$}.  Let
$v\in V_1$ be the first atom that contains such an $x_0$ in a
pre-order tree walk over~$T_1$. Let $(q_a,g_a)$ be the progress tree
realized in $\bar a^\ast$ with $v \in q_a$, and let $u$ be the root of
$q_a$. Then $h_{\bar a^\ast}(z) \neq \ast$ for all predecessor variables
$z$ in $u$. Since $\bar b^{\ast} \prec \bar a^{\ast}$, the same is
true for $h_{\bar b^\ast}$. Due to Condition~(2) of progress trees, it
follows that there is a progress tree $(q_b,g_b)$ realized in
$\bar b^\ast$ such that the root of $q'$ is~$u$. From
$\bar b^{\ast} \prec \bar a^{\ast}$, it follows that
$V_{q_a} \subsetneq V_{q_b}$ or $V_{q_a}=V_{q_b}$ and $g_b(x)=\ast$
implies $g_a(x)=\ast$ for all $x \in \mn{var}(q_a)$.  From
$h_{ \bar b^{\ast}}(x_0) \neq \ast$ and
$h_{\bar a^{\ast}}(x_0) = \ast$, it further follows that
$g_b(x_0) \neq \ast$ and \mbox{$g_a(x_0)=\ast$}.  Consequently, $(q_a,g_a)$
is before $(q_b,g_b)$ in the list $\mn{trees}(u,h_{\bar z})$ where
$h_{\bar z}=h_{ \bar b^{\ast}}|_{\bar z}=h_{ \bar a^{\ast}}|_{\bar z}$
and thus  the recursive call
that leads to the output of $h_{\bar b^{\ast}}$ takes place before the recursive
call that leads to the output of $h_{\bar a^{\ast}}$.
\end{proof}
Let $\bar a^{\ast}, \bar b^{\ast}$ be partial answers to $q_1$ on
$D_1$. We say that $\bar a^{\ast}$ \emph{prunes}~$\bar b^{\ast}$ if
there is a %non-root
node $v=R(\bar y) \in V_1$ with predecessor
variables $\bar z$ % such that
% $h_{\bar b^{\ast}}(\bar z) \subseteq \mn{adom}(D_1)$,
% $h_{\bar b^{\ast}}(\bar y) \cap \{ \ast \} \neq \emptyset$, and
% there is a
and a progress tree
$(q,g) \in \mn{trees}(v,h_{\bar b^{\ast}}|_{\bar z})$ that is realized
in $\bar b^\ast$ and removed from
$\mn{trees}(v,h_{\bar b^{\ast}}|_{\bar z})$ when
$\bar a^{\ast}$ is output.

\proppartialalgocorrect*
\begin{proof}
  We first argue that whenever a partial answer is pruned, then it is
  not a minimal partial answer. Assume that partial answer
  $\bar a^{\ast}$ prunes partial answer $\bar b^{\ast}$.  Then there
  is a $v \in V_1$ with predecessor variables $\bar z$ and a progress
  tree $(q,g) \in \mn{trees}(v,h_{\bar b^{\ast}}|_{\bar z})$ such
  that
  % $h_{\bar b^{\ast}}(\bar z) \subseteq \mn{adom}(D_1)$,
  % $h_{\bar b^{\ast}}(y) \cap \{ \ast \} \neq \emptyset$, and
  $(q,g)$ is realized in $\bar b^\ast$ and removed from
  $\mn{trees}(v,h_{\bar b^{\ast}}|_{\bar z})$ when $\bar a^{\ast}$ is
  output. Thus $(q,h_{\bar a^\ast}|_{\bar z}) \prec_{\mn{db}} (q,g)$.
  We call the variables in $\bar z$ \emph{root variables} of $q$.  A
  \emph{fringe variable} of $q$ of is a variable $x$ that occurs in a
  leaf $u$ of $T_q$ that has a successor $u'$ in $T_1$ in which $x$
  also occurs. We observe that
  \begin{itemize}

  \item[($*$)] for all variables $x$ that are root variables or fringe
    variables of $q$, $h_{\bar a^\ast}(x)=g(x)$.

  \end{itemize}
  In fact, this follows from
  $(q,h_{\bar a^\ast}|_{\bar z}) \prec_{\mn{db}} (q,g)$ and the fact
  that, by Conditions~1 and~2 of progress trees, $g(x) \neq
  \ast$ for all variables $x$ mentioned in ($*$).

  Let ${\bar c}^{\ast}$ be obtained from $\bar b^{\ast}$ by setting
  $h_{{\bar c}^{\ast}}(x)=h_{\bar a^{\ast}}(x)$ for all
  $x \in \mn{var}(q)$. It can be verified that ${\bar c}^{\ast}$ is a
  partial answer to $q_1$ on $D_1$. In particular, there must be a
  homomorphism $h_{\bar b}$ from $q_1$ to $D_1$ such that
  $h_{\bar b^{\ast}}(x)=\ast$ iff $h_{\bar b}(x) \in N$ for all
  $x \in \mn{var}(q_1)$. By Condition~3 of progress trees, there is a
  homomorphism $h$ from $q$ to $D_1$ such that $h(x) \in N$ if
  $h_{\bar b^{\ast}}=\ast$ and $h(x)=h_{\bar b^{\ast}}(x)$ otherwise.
  Let $h_{\bar c}$ be obtained from $h_{\bar b}$ by setting
  $h_{\bar c}(x)=h(x)$ for all $x \in \mn{var}(q)$. Due to~($*$),
  $h_{\bar c}$ is a homomorphism, and thus ${\bar c}^{\ast}$ is a
  partial answer to $q_1$ on $D_1$.  Moreover, the construction
  of $\bar c^*$ yields 
  ${\bar c}^{\ast} \prec_{\mn{db}} \bar b^{\ast}$, and thus we have shown that
  $\bar b^{\ast}$ is not a minimal partial answer.

We next show that if a partial answer is not a minimal partial answer,
then it is pruned. Assume that $ \bar b^{\ast}$ is not a minimal partial
answer. Then there is a minimal partial answer $\bar a^{\ast}$ with
$\bar a^{\ast} \prec \bar b^{\ast}$. We argue that $\bar a^{\ast}$
prunes $ \bar b^{\ast}$. By Lemma~\ref{lem:nopruning}, $\bar a^{\ast}$
is output before $\bar b^{\ast}$ when no pruning takes place.  Since
$\bar a^{\ast} \prec \bar b^{\ast}$ , there is an
$x_0 \in \mn{var}(q_1)$ with $h_{\bar a^{\ast}}(x_0)$ a database
constant %, say $b_x$,
and $h_{\bar b^{\ast}}(x_0)=\ast$. Let $v \in V_1$ be the first node
encountered in a pre-order tree walk over $T_1$ that contains $x_0$,
and let $(q,g)$ be the progress tree realized in $\bar b^\ast$ with
$v \in q$. Further let $u$ be the root of $q$ and $\bar z$ be the
predecessor variables of $u$.
% Define
% $g = h_{\bar a^\ast}|_{\mn{var}(q)}}$.
It follows from $\bar a^{\ast} \prec \bar b^{\ast}$ and the fact that
$x_0$ occurs in $v$ occurs in $q$ that
$(q,g) \succ_{\mn{db}} (q,h_{\bar a^\ast}|_{\mn{var}(q)})$. Thus
$(g,q)$ is removed from $(q,g) \in \mn{trees}(u, h)$, with
$h=h_{\bar a^{\ast}}|_{\bar z}= h_{\bar b^{\ast}}|_{\bar z}$, when
$\bar a^{\ast}$ is output (which is the case since $\bar a^{\ast}$ is
a minimal partial answer and thus not pruned), and therefore
$\bar a^{\ast}$ prunes $\bar b^{\ast}$.
\end{proof}

\section{Proofs for Section~\ref{sect:enummulti}}
\label{app:multiwild}

In the main part of the paper, we have declared the goal to be the
enumeration of $q_0(\mn{ch}^{q_0}_\Omc(D))^{\Wmc}_\Nbf$. Here, we
actually prove something slightly more general, as follows (this
parallels what is done in Appendix~\ref{app:singlewild}).
\begin{proposition}
  \label{prop:enummultiwildcards}
  For every CQ $q(\bar x)$ that is acyclic and free\=/connex acyclic, enumerating
  the answers $q(D)^{\Wmc}_\Nbf$ is in \dlc for databases $D$ and
  sets of nulls $N \subseteq \mn{adom}(D)$ such that $D$ is chase-like
  with witness $D_{1},\dots,D_{n}$ where
  $|\mn{adom}(D_{i})|$ does not depend on $D$ for $1 \leq i \leq n$.
\end{proposition}
To prove Proposition~\ref{prop:enummultiwildcards}, we use %(almost)
as a blackbox the \dlc algorithm for the enumeration of minimal
partial answers with a single wildcard presented in the previous
section. In addition, we use a \dlc algorithm for all-testing
of (not necessarily minimal) partial answers with multi-wildcards. 
}

\ifbool{arxive}{
\subsection{All-Testing Partial Answers with Multi-Wildcards}

We show that all-testing of (not necessarily minimal) partial answers
with multi-wildcards is in \dlc. The following proposition
makes this precise.
}
{
\section{All-Testing Partial Answers with Multi-Wildcards}
\label{app:alltestingmulti}

The algorithm for enumerating minimal partial answers with 
multi-wildcards presented in Section~\ref{sect:enummulti} relies on a 
\dlc algorithm for all-testing (not necessarily minimal) partial answers 
with multi-wildcards. In this section, we develop such an algorithm. 
}
\ifbool{arxive}{
\begin{proposition}
    \label{lemma:SPA-alltest}
    % Let $Q \in (\class{G},\class{CQ})$ be acyclic and
    % free\=/connex. Then all-testing partial
    % answers to $Q$ with multi-wildcards is in \cdlin.
    For every CQ $q(\bar x)$ that is acyclic and free\=/connex acyclic,
    all-testing of the answers $q(D)^{\Wmc,\not\prec}_\Nbf$ is in
    \dlc for databases $D$ and sets of nulls
    $N \subseteq \mn{adom}(D)$ such that $D$ is chase-like
    with witness $D_{1},\dots,D_{n}$ where $|\mn{adom}(D_{i})|$ does
    not depend on $D$ for $1 \leq i \leq n$.
\end{proposition}
We remind the reader that when $Q(\bar x)=(\Omc,\Sbf,q) \in
(\class{G},\text{CQ})$, $D=\mn{ch}^q_\Omc(D_0)$, and $N=\mn{adom}(D)
\setminus \mn{adom}(D_0)$, then the set of partial answers with
multi-wildcards to $Q$ on $D_0$ is not necessarily identical to
$q(D)^{\Wmc,\not\prec}_\Nbf$. In fact, one can show that for the former
all-testing in \dlc is not possible unless BMM can be done in
quadratic time, and thus it is vital that we work with
$q(D)^{\Wmc,\not\prec}_\Nbf$.

\medskip
}{
\begin{proposition}
    \label{lemma:SPA-alltest}
    % Let $Q \in (\class{G},\class{CQ})$ be acyclic and
    % free\=/connex. Then all-testing partial
    % answers to $Q$ with multi-wildcards is in \cdlin.
    For every CQ $q(\bar x)$ that is acyclic and free\=/connex
    acyclic, all-testing of the answers $q(D)^{\Wmc,\not\prec}_\Nbf$ is
    in \dlc for databases $D$ of the form $\mn{ch}^q_\Omc(D_0)$ and
    sets of nulls $N = \mn{adom}(D) \setminus \mn{adom}(D_0)$.
\end{proposition}
}

To prove Proposition~\ref{lemma:SPA-alltest}, fix a CQ $q(\bar x)$
over schema \Sbf that is acyclic and free\=/connex acyclic, and let
$D$ be an \Sbf-database and $N$ a set of nulls
satisfying the conditions from Proposition~\ref{lemma:SPA-alltest}. In
time linear in $||D||$ we can convert $q$ and $D$ into a CQ $q'(\bar
x)$ without quantified variables and a database $D'$ such that $D$ and
$D'$ have the same Gaifman graph and $q(D)=q'(D')$, and thus also
$q(D)^{\Wmc,\not\prec}_\Nbf=q'(D')^{\Wmc,\not\prec}_\Nbf$. Note that we
achieve this as part of the preprocessing carried out in
\ifbool{arxive}{
Section~\ref{sect:LPAsingleWildcardUpper}, an outline of how to do
this is given in Appendix~\ref{thm:upperGpartial} and details are in
\cite{berkholz-enum-tutorial}. Since $D$ and $D'$ have the same 
Gaifman graph, $D'$ and $N$ also satisfy the conditions from 
Proposition~\ref{lemma:SPA-alltest} and we may in fact simply assume 
that $q$ contains no quantified variables. 
}
{
  Section~\ref{sect:LPAsingleWildcardUpper}. We may substitute
  $q$ with $q'$ and $D$ with $D'$, in effect simply assuming
  that $q$ contains no quantified variables. 
}

% \begin{proof}
%   As a preliminary, we observe that if $q(\bar x)$ is a CQ that is
%   acyclic and free-connex and $q'(\bar x')$ is obtained from $q$ by
%   dropping atoms, %while retaining all answer variables,
%   then $q'$ is
%   also acyclic and free-connex (even if some answer variables
%   `disappear' entirely by removing atoms) .
%   % The condition `while retaining all
%   % answer variables' is essentially w.l.o.g.\ since we may assume that
%   % there is a selected monadic relation symbol $P$ such that $q$
%   % contains $P(x)$ for all answer variables $x$ and all relevant
%   % databases $D$ contain $P(c)$ for all $c \in \mn{adom}(D)$, and
%   % considering only CQs $q'$ such that atoms $P(x)$ are never dropped
%   % from~$q$.

%   Fix a $Q(\bar x)=(\Omc,\Sbf,q) \in (\class{G},\class{CQ})$ that is
%   acyclic and free\=/connex, and let $D$ be an \Sbf-database.
%   The preprocessing phase of the algorithm works as follows. 
% \end{proof}
%

Let $T_q=(V_q,E_q)$ be a join tree for $q$. 
A \emph{multi-progress tree} is a pair $(q',g)$ with $q'$ a subtree of
$q$ (as defined in Section~\ref{sect:LPAsingleWildcardUpper}) and
$g:\mn{var}(q) \rightarrow (\mn{adom}(D) \setminus N) \cup \Wmc$ a map
such that the following conditions are satisfied:
\begin{enumerate}

\item $g(x) \notin \Wmc$ for every predecessor variable $x$ in the root
  of~$T_{q'}$; 

% \item if $v \in V_q$ and $v' \notin V_q$ is a successor of $v$ in $T_1$,
%   then $g(x)\neq\ast$ for all predecessor variables $x$ in $v'$;

% \item if $v,v' \in V_q$ with $v'$ a successor of $v$ in $T_1$, then
%   $g(x) =\ast$ for some predecessor variable $x$ in $v'$;

% \item there is a homomorphism $h$ from $q$ to $D_1$ such that for all
%   $x,y \in \mn{var}(q)$,
%   %
%   \begin{itemize}

%   \item $h(x) \in N$ if $g(x) \in \Wmc$ and $h(x)=g(x)$ otherwise;

%   \item $h(x) = h(y)$ iff $g(x) = g(y)$;
    
%   \end{itemize}

\item if $v \in V_{q'}$ and $v'$ is a successor of $v$ in $T_q$, then
 $v' \in V_{q'}$ if and only if $g(x) \in \Wmc$ for some predecessor
 variable $x$ in $v'$;
  
% \item there is a homomorphism $h$ from $q$ to $D_1$ such that for all
%   $x \in \mn{var}(q)$, $h(x) \in N$ if $g(x) = \ast$ and
%   $h(x)=g(x)$ otherwise;

\item the constants in the range of $g$ form a guarded set in $D$.
  
\end{enumerate}
%
% Note that the root of $T_q$ (which has no predecessor variables) may
% be part of a subtree of $q$. This is in contrast to the constructions
% in Section~\ref{sect:LPAsingleWildcardUpper} where we had made sure
% that all variables in the root of the CQ must map to (non-null)
% database constants.

\ifbool{arxive}{
Recall that $D$ is chase-like with witness
$D_1,\dots,D_n$. }
{
  Let $D_1,\dots,D_n$ be the tree-like structures in $D$ generated
  by the chase.
}
A set $S=\{(q_1,g_1),\dots,(q_\ell,g_\ell)\}$ of
multi-progress trees is \emph{valid}
if % the following conditions are satisfied:
% %
% \begin{itemize}
%
% \item The wildcards in the range of
% $g=g_1 \cup \cdots \cup g_\ell$ are a prefix of the ordered set
% $\Wmc = \{ \ast_1,\ast_1,\dots \}$ and respect the order of
% the answer variables in $\bar x$, that is, if the first occurrence
% of~$x$ is before the first occurrence of $x'$ in $\bar x$,
% $g(x)=\bar_i$, and $g(x')=\bar j$, then $i < j$.
%
% \item T
  there is a homomorphism $h$
  from $q_1 \cup \cdots \cup q_\ell$ to some database $D_i$, with
  $1 \leq i \leq n$, that is \emph{compatible}
with $g$, that is, for all
$x,y \in \mn{var}(q_1) \cup \cdots \cup \mn{var}(q_\ell)$,
\begin{enumerate}

  \item[(a)] $h(x) \in N$ if $g(x) \in \Wmc$ and $h(x)=g(x)$ otherwise;

  \item[(b)] $g(x) = g(y)$ implies $h(x) = h(y)$.
    
\end{enumerate}
%\end{itemize}

Our \dlc algorithm for all-testing $q(D)^{\Wmc,\not\prec}_\Nbf$
uses as a black box a \dlc algorithm $A_{q'}$ for all-testing
of $q'(D)$, for every subquery $q'(\bar x')$ of $q(\bar x)$, that is,
for every CQ $q'(\bar x')$ that can be obtained from $q$ by dropping
atoms. Note that all of these $q'$ contain no quantified variables and
are thus free-connex acyclic, implying that all-testing $q'(D)$ is
  possible in \dlc by
  Proposition~\ref{prop:allTestingCompleteUpper}.  There are clearly
only constantly many such subqueries.

In the preprocessing phase, we run the
preprocessing phases of all the algorithms $A_{q'}$, $q'$ a subquery
of $q$. In addition, we precompute a lookup table \mn{nullhom} that
stores a Boolean value for all sets $S$ of multi-progress trees that
contain at most $|\mn{var}(q)|$ such trees.  Let
$S=\{(q_1,g_1),\dots,(q_\ell,g_\ell)\}$.  The stored value is 1 if $S$
is valid and 0 otherwise.  Such a lookup table can be accessed and
updated in $O(1)$ time on a RAM. The proof of
the following is similar to that of
Lemma~\ref{lem:precomputetreeslistslinear}.
\begin{lemma}
  The lookup table $\mn{nullhom}$ can be computed in time linear in
  $||D||$.
\end{lemma}
\ifbool{arxive}{
\begin{proof}
  To compute the table, we iterate over the databases
  $D_1,\dots,D_n$. For each $D_i$, set $G= \mn{adom}(D_i) \setminus N$
  and iterate over all sets $S$ of pairs $(q',g)$ with $q'$ a subtree
  of $q$ and $g:\mn{var}(q') \rightarrow G \cup \Wmc$, such that there
  are at most $|q|$ pairs in $S$ and Conditions~(1) and~(2) of
  multi-progress trees is satisfied for all pairs in $S$.
  Condition~(3) is satisfied since $G$ is a guarded set in $D_i$.
  This also implies that its cardinality is bounded by a constant and
  thus there are only constantly many pairs $(q',g)$ of the described
  form and consequently also only constantly many sets $S$. Let
  $S=\{(q_1,g_1),\dots,(q_\ell,g_\ell)\}$.  We then check whether
  there is a homomorphism $h$ from $q_1 \cup \cdots \cup q_\ell$ to
  $D_i$ that is compatible with $g_1,\dots,g_\ell$.  If this is the
  case, we set $\mn{nullhom}(S)=1$. Otherwise, $\mn{nullhom}(S)=1$ as
  all memory is initialized with value~0 in our machine model. Note
  that we can check the existence of $h$ brute force: there are only
  constantly many potential targets because $|\mn{adom}(D_{i})|$ does
  not depend on $D$.
\end{proof}
}
We now describe the testing phase of our algorithm. Assume that a
multi-wildcard tuple $\bar a^\Wmc$ of length $|\bar x|$ is to be
tested.  We may first check whether wildcards are used in the required
way and answer `no' if this is not the case. More precisely, we check
that the wildcards in $\bar a^\Wmc$ are a prefix of the ordered set
$\Wmc = \{ \ast_1,\ast_2,\dots \}$ and that multiple occurrences of the
same variable in $\bar x$ are matched by multiple occurrences of the
same wildcard in $\bar a^\Wmc$. If this is the case, we may view
$\bar a^\Wmc$ as a map
$h_{\bar a^{\Wmc}}: \mn{var}(q) \rightarrow (\mn{adom}(D) \setminus N
) \cup \Wmc$ in the obvious way. We may then check that the wildcards
in $\bar a^\Wmc$ respect the order of the answer variables in
$\bar x$, that is, if the first occurrence of~$x$ is before the first
occurrence of $x'$ in $\bar x$, $h_{\bar a^{\Wmc}}(x)=\ast_i$, and
$h_{\bar a^{\Wmc}}(x')=\ast_j$, then $i < j$.
    
We say that a multi-progress tree $(q',g)$ is \emph{realized} in
$\bar a^\Wmc$ if $h_{\bar a^{\Wmc}}(x)=g(x)$ for all
$x \in \mn{var}(q')$. Let $T$ be the set of all multi-progress trees
realized in $\bar a^\Wmc$ and let $\sim$ be the smallest equivalence
relation on $T$ such that $(q_1,g_1) \sim (q_2,g_2)$ if there are
variables $x_1 \in \mn{var}(q_1)$ and $x_2 \in \mn{var}(q_2)$ such
that $g_1(x_1)=g_2(x_2) \in \Wmc$. We consider each equivalence class
$S \subseteq T$ of `$\sim$' and check whether $S$ is valid by testing
if $\mn{nullhom}(S)=1$. If any of the checks fails, we answer `no'.
Since at most $|\mn{var}(q)|$ (and thus only constantly many)
multi-progress trees may be realized in $\bar a^\Wmc$, the required
checks can be done in constant time.

We then do one last check. Let $q'$ be the subquery of $q$ that
consists of all atoms $\alpha$ such that for all variables $x$ in
$\alpha$, \mbox{$h_{\bar a^{\Wmc}}(x) \notin \Wmc$}. Further let $\bar a$ be
the tuple over $\mn{adom}(D) \setminus N$ obtained from
$\bar a^\Wmc=(a^\Wmc_1,\dots,a^\Wmc_{|\bar x|})$ by dropping
$a^\Wmc_i$ whenever the $i$-th position in $\bar x$ is an answer
variable that is not in $\mn{var}(q')$. We then use algorithm $A_{q'}$
to test whether $\bar a \in q'(D)$ and return the result.  The
following lemma asserts that the returned answer is correct, which
finishes the proof of Proposition~\ref{prop:allTestingCompleteUpper}.
\begin{lemma}
  $\bar a^\Wmc \in q(D)_\Nbf^\Wmc$ iff the testing phase returns
  `yes'.
  %every
%  equivalence class $S \subseteq T$ of~`$\sim$' is valid.
\end{lemma}
\ifbool{arxive}{
\begin{proof}
  `if'. Assume that the testing phase returns `yes'. Then
  $\bar a \in q'(D)$ and thus we may view $\bar a$ as a homomorphism
  $h_{\bar a}$ from $q'$ to $D$ in the obvious way. Note that the
  range of $h_{\bar a}$ falls within $\mn{adom}(D) \setminus N$ since
  all constants in $\bar a$ are from this set. We next extend $h$ by
  considering one equivalence class $S \subseteq T$ of `$\sim$' at the
  time. Let $S=\{(q_1,g_1),\dots,(q_\ell,g_\ell)\}$.  Since the
  testing phase has returned `yes', $S$ is valid and thus there is a
  homomorphism $h_S$ from $q_1 \cup \cdots \cup q_\ell$ to some $D_i$,
  with $1 \leq i \leq n$, that is compatible with
  $g=g_1 \cup \cdots \cup g_n$. Taking the union of $h_{\bar a}$ and
  all the homomorphisms $h_S$ yields a homomorphism $h$ from $q$ to
  $D$ that yields an answer $\bar b \in q(D)$ such that $\bar a^\Wmc$
  is obtained from $\bar b$ by replacing nulls with wildcards from \Wmc.
  Consequently, $\bar a^\Wmc \in q(D)_\Nbf^\Wmc$.

  `only if'. Assume that $\bar a^\Wmc \in q(D)_\Nbf^\Wmc$. Then there is
  a homomorphism $h$ from $q$ to $D$ that yields an answer
  $\bar b \in q(D)$ such that $\bar a^\Wmc$ is obtained from $\bar b$
  by replacing nulls with wildcards from \Wmc. Clearly, $h$ is also a
  homomorphism from the subquery $q'$ of $q$ constructed during the
  testing phase to $D$. Consequently, $\bar a \in q(D)$ where $\bar a$
  is the tuple over $\mn{adom}(D) \setminus N$ constructed along with
  $q'$, and thus the test for $\bar a \in q(D)$ made in the testing
  phase succeeds. It remains to argue that every equivalence class
  $S \subseteq T$ w.r.t.\ `$\sim$' is valid, and thus also the checks
  associated with that succeed.  Let
  $S=\{(q_1,g_1),\dots,(q_\ell,g_\ell)\}$.  We first observe the
  following. The proof is identical to the proof of the analogous
  claim for single-wildcard progress trees in the proof of
  Lemma~\ref{lem:precomputetreeslistslinear}.  Details are omitted.
  \\[2mm]
  {\sc Claim.}
  Let $(q',g)$ be a multiple-progress tree and let $h$
  be a homomorphism from $q'$ to $D$ such that for all
  $x \in \mn{var}(q')$, $h(x) \in N$ if $g(x) \in \Wmc$ and $h(x)=g(x)$
  otherwise. Then $h$ is a homomorphism from $q$ to $D_{i}$ for some
  $i \in \{1,\dots,n\}$.
  \\[2mm]
  Recall the definition of `$\sim$' via shared wildcards and the fact
  that $\mn{adom}(D_i) \cap \mn{adom}(D_j) \cap N = \emptyset$ for
  $1 \leq i < j \leq n$, by definition of chase-like instances. From
  this and the claim it follows that there is a single $D_i$, with
  $1 \leq i \leq n$, such that $h$ is a homomorphism from
  $q_1 \cup \cdots \cup q_\ell$ to $D_i$. Moreover, $h$ is clearly be
  compatible with $g=g_1 \cup \cdots \cup g_\ell$ and thus $S$ is
  valid.
\end{proof}
}

\ifbool{arxive}{
\subsection{Enumeration with Multi-Wildcards}

We prove Proposition~\ref{prop:enummultiwildcards} using exactly the
algorithm described in the main part of the paper.  Balls and cones
play a crucial role in the algorithm. The following lemma
explains how they link the set $q(D)_\Nbf^\Wmc$ that we aim to enumerate
to the set $q(D)_\Nbf^\ast$ that we enumerate in the outer {\bf forall}
loop. Note that, by Point~(2), we can indeed choose an $\bar a^\Wmc$
with the required properties in the algorithm.
%
% The cones are of interest to us for two reasons: they are finite sets, thus
% computable in constant time, and every least strong partial answer $\bar{a}^{\Wmc} \in Q^{\Wmc}(D)$
% lies in a cone over some least partial answer $\bar{b}^{\ast} \in Q^{\ast}(D)$.
%
% Indeed, for a wildcard tuple $\bar{a}^{\ast}$
% we have $|\cone(\bar{a}^{\ast})| \leq (k+1)^k$ where $k = |\bar{x}|$,
% thus $\cone(\bar{a}^{\ast})$ can be computed in constant time.
% Moreover, cones over least partial answers cover the set of least strong partial answers.
%
\begin{lemma}
    \label{lemma:cone-1}
    ~\\[-4mm]
    \begin{enumerate}

    \item $\displaystyle q(D)_\Nbf^\Wmc \subseteq \bigcup_{\bar{a}^{\ast} \in q(D)_\Nbf^\ast}
      \mn{cone}^\Wmc(\bar{a}^{\ast})$;

    \item for all $\bar{a}^{\ast} \in q(D)^\ast_\Nbf$, \mbox{$\emptyset \neq 
      \mn{min}^{\!\prec}(B^\Wmc(\bar{a}^{\ast}) \cap
      q(D)_\Nbf^{\Wmc,\not\prec})  \subseteq q(D)_\Nbf^\Wmc$};

    \item for all distinct $\bar{a}^{\ast}, \bar{b}^{\ast} \in q(D)_\Nbf^{\ast}$,
      $B^{\Wmc}(\bar{a}^{\ast}) \cap  \mn{cone}^\Wmc(\bar{b}^{\ast}) =
      \emptyset$.

    \end{enumerate}
\end{lemma}
\begin{proof}
  For Point~(1), let $ \bar{b}^{\Wmc} \in q(D)_\Nbf^\Wmc$ and let
  $\bar{b}^{\ast}$ be obtained from $\bar{b}^{\Wmc}$ by replacing
  every wildcard from \Wmc by `$\ast$'.  Then $\bar{b}^{\ast}$ is a
  partial answer to $q$ on $D$, and thus there is an
  $\bar{a}^{\ast} \in q(D)_\Nbf^\ast$ such that
  $\bar{a}^{\ast} \preceq \bar{b}^{\ast}$. It is easy to verify that
  $\bar{b}^{\Wmc} \in \mn{cone}^\Wmc(\bar{a}^{\ast})$.

  For Point~(2), let $\bar{a}^{\ast} \in q(D)_\Nbf^\ast$. To prove that the
  set $\mn{min}^{\!\prec}(B^\Wmc(\bar{a}^{\ast}) \cap
  q(D)_\Nbf^{\Wmc,\not\prec})$ is non-empty, consider the multi-wildcard
  tuple $\bar a^\Wmc$ obtained from $\bar a^\ast$ by replacing every
  occurrence of `$\ast$' with a different wildcard from \Wmc. It is
  clear that $\bar a^\Wmc \in B^\Wmc(\bar{a}^{\ast}) \cap
  q(D)_\Nbf^{\Wmc,\not\prec}$ and thus this set is non-empty. Since it is
  finite (and in fact of constant size),
  $\mn{min}^{\!\prec}(B^\Wmc(\bar{a}^{\ast}) \cap
  q(D)_\Nbf^{\Wmc,\not\prec})$ is also non-empty. To show that this set is
  a subset of $q(D)_\Nbf^\Wmc$, first observe that any tuple $\bar a^\Wmc$
  in it is from $q(D)_\Nbf^{\Wmc,\not\prec}$. Now assume to the contrary of
  what is to be shown that there is a $\bar b^\Wmc \in
  q(D)_\Nbf^{\Wmc,\not\prec}$ with $\bar b^\Wmc \prec \bar a^\Wmc$.  Then
  $\bar b^\Wmc$ cannot be in $B^\Wmc(\bar{a}^{\ast})$ as otherwise
  $\bar a^\Wmc$ would not be minimal. This and $\bar b^\Wmc \prec \bar
  a^\Wmc$ means that for $\bar a^\Wmc = (a_1,\dots,a_{|\bar x|})$ and
  $\bar b^\Wmc = (b_1,\dots,b_{|\bar x|})$, 
  \begin{itemize}
  \item[($\dagger$)]
there is an $i \in
  \{1,\dots,{|\bar x|}\}$ with $b_i \in \mn{adom}(D)$ and $a_i \in
  \Wmc$.  
\end{itemize}
Let $\bar b^\ast$ be $\bar b^\Wmc$ with every wildcard from \Wmc
replaced by `$\ast$'. It is not hard to verify that $\bar b^\ast$ is a
partial answer to $q$ on $D$ and, using
$\bar b^\Wmc \prec \bar a^\Wmc$ and ($\dagger$), that
$\bar b^\ast \prec \bar a^\ast$, in contradiction to
$\bar{a}^{\ast} \in q(D)_\Nbf^\ast$.

We prove Point~(3) by contradiction.  Assume that $\bar{a}^{\ast},
\bar{b}^{\ast} \in q^{\ast}_\Nbf(D)$ are distinct and that there is a
multi-wildcard tuple $\bar{a}^{\Wmc} \in B^{\Wmc}(\bar{a}^{\ast}) \cap
\mn{cone}^\Wmc(\bar{b}^{\ast})$.  Then by definition of cones,
there is a wildcard tuple $\bar{c}^{\ast}$ such that $\bar{a}^{\Wmc}
\in B^{\ast}(\bar{c}^{\ast})$ and $\bar{b}^{\ast} \prec
\bar{c}^{\ast}$.  But, by definition of $B^{\Wmc}$, we have that
$B^{\ast}(\bar{a}^{\Wmc}) = \bar{a}^{\ast}$ and $
B^{\ast}(\bar{a}^{\Wmc}) = \bar{c}^{\ast}$. Thus, $\bar{a}^{\ast} =
\bar{c}^{\ast}$.

    Therefore $\bar{b}^{\ast} \prec \bar{a}^{\ast}$ which is impossible as $\bar{a}^{\Wmc}$ and $\bar{b}^{\ast}$
    are both minimal partial answers and, thus, incomparable.
\end{proof}
\lemmultiwildcorr*
\begin{proof}
  We first argue that all tuples output by the algorithm are from
  $q(D)_\Nbf^\Wmc$. Indeed, tuples output during the {\bf forall} loop are
  from $q(D)_\Nbf^\Wmc$ by Point~(2) of Lemma~\ref{lemma:cone-1}. It thus
  remains to consider tuples that were output because they remained on
  the list $L$ after the execution of the {\bf forall} loop. We first
  observe the following invariant, which follows from an easy analysis
  of the algorithm.
  \\[2mm]
  {\bf Claim.} Throughout the run of the algorithm, $F(\bar a^\Wmc)=1$
  implies that $\bar a^\Wmc$ was already added to $L$ or is not in $q(D)_\Nbf^\Wmc$.
  \\[2mm]
  Now consider a tuple $\bar a^\Wmc$ that was output after the
  execution of the {\bf forall} loop and assume to the contrary of
  what is to be shown that there is a $\bar b^\Wmc \in q(D)_\Nbf^{\Wmc}$
  such that $\bar b^\Wmc \prec \bar a^\Wmc$. Let $\bar b^\ast$ be
  obtained from $\bar b^\Wmc$ by replacing every wildcard from \Wmc
  with~`$\ast$'. Then $\bar b^\ast$ is a (not necessarily minimal)
  partial answer to $q$ on $D$ and thus there is a
  $\bar c^\ast \in q(D)_\Nbf^\ast$ with $\bar c^\ast \preceq \bar
  b^\ast$. Then $\bar b^\Wmc \in \mn{cone}^\Wmc(\bar c^\ast)$.  We
  next argue that at some point $\bar b^\Wmc$ is appended to the list
  $L$. In fact, consider the iteration of the outer {\bf forall} loop
  that processes~$\bar c^\ast$. If $F(\bar b^\Wmc)=0$ at that point,
  then $\bar b^\Wmc$ is appended to $L$. If $F(\bar
  b^\Wmc)=1$, then by the claim $\bar b^\Wmc$ was appended to
  $L$ in a previous iteration. In both cases, when $\bar
  b^\Wmc$ was added to $L$, $\mn{prune}(\bar
  b^\Wmc)$ was called and $\bar a^\Wmc$ was removed from
  $L$ and $F(\bar a^\Wmc)$ set to $1$, ensuring that $\bar
  a^\Wmc$ is never added back to $L$. This is a contradiction to $\bar
  a^\Wmc$ being output because it remained on $L$.

  \medskip We next argue that all tuples from
  $q(D)^\Wmc_\Nbf$ are output. Let $\bar a^\Wmc \in
  q(D)_\Nbf^\Wmc$. By Point~(1) of Lemma~\ref{lemma:cone-1}, there is a
  $\bar b^\ast \in q(D)_\Nbf^\ast$ with $\bar a^\Wmc \in
  \mn{cone}^\Wmc(\bar
  b^\ast)$. Consider the iteration of the outer {\bf forall} loop that
  processes~$\bar b^\ast$. If $F(\bar
  a^\Wmc)=0$ at that point, then $\bar a^\Wmc$ is appended to
  $L$. If $F(\bar a^\Wmc)=1$, then by the claim $\bar a^\Wmc$ was
  appended to $L$ in a previous iteration. Since
  $\mn{prune}(\bar c^\Wmc)$ is only ever called for tuples
  $\bar c^\Wmc \in q(D)_\Nbf^{\Wmc,\not\prec}$, the only way $\bar a^\Wmc$
  can be removed from $L$ is when it is chosen to be output in the
  outer {\bf forall} loop. If that never happens, it is still on $L$
  after that loop has terminated and thus also output.

  \medskip Finally, we argue that there are no repetitions. This,
  however, is an immediate consequence of the use of the lookup table
  $F$ to make sure that every multi-wildcard is appended to list
  $L$ at most once and of the fact that when a tuple is output in
  the outer  {\bf forall} loop, then it is removed from $L$.
\end{proof}
}

\section{Illustrating the Algorithm}
\label{sect:example}

%some tikz magic!
\def\sc{1}

\newcommand{\tshape}{
    \node (j1) at ({0*\sc}, {0*\sc}) {};
    \node (j2) at ({-1*\sc},{-2*\sc}) {};
    \node (j3) at ({0*\sc}, {-1*\sc}) {};
    \node (j4) at ({1*\sc}, {-2*\sc}) {};
    \node (j5) at ({1*\sc}, {-3*\sc}) {};
    \node (jC) at ({-1*\sc}, {-3*\sc}) {};
}

\newcommand{\qshape}{
    \node (i0) at ({0*\sc}, {0*\sc}) {};
    \node (i1) at ({-.5*\sc}, {-3*\sc}) {};
    \node (i2) at ({-.5*\sc},{-2*\sc}) {};
    \node (i3) at ({0*\sc}, {-1*\sc}) {};
    \node (i4) at ({.5*\sc}, {-2*\sc}) {};
    \node (i5) at ({.5*\sc}, {-3*\sc}) {};
}

\tikzset{
    pics/Piece/.style n args={6}{
        code = { %
            
            \tshape
            
%            \node (treeq1)     at ($#1 + (j1)$) {\footnotesize $c_0$};
            \node (treeq2)     at ($#1 + (j2)$) {\footnotesize $c_0 #2 #3$};
            \node (treeq3)     at ($#1 + (j3)$) {\footnotesize $c_0 #3 #4$};
            \node (treeq4)     at ($#1 + (j4)$) {\footnotesize $c_0 #5 #4$};
            \node (treeq5)     at ($#1 + (j5)$) {\footnotesize $c_0 #6 #5$};
            \node (treeqC)     at ($#1 + (jC)$) {\footnotesize $c_0 #2$};
            
            \draw (treeqC) -> (treeq2) -> (treeq3) -> (treeq4) ->(treeq5);
            \draw (treeq3) -> (treeq1);
        }
    }
}

\tikzset{
    pics/PiecePA/.style n args={6}{
        code = { %
            
            \tshape
            
%            \node (treeq1)     at ($#1 + (j1)$) {\footnotesize $c_0$};
            \node (treeq2)     at ($#1 + (j2)$) {\footnotesize $#2 #3$};
            \node (treeq3)     at ($#1 + (j3)$) {\footnotesize $#3 #4$};
            \node (treeq4)     at ($#1 + (j4)$) {\footnotesize $#5 #4$};
            \node (treeq5)     at ($#1 + (j5)$) {\footnotesize $#6 #5$};
            \node (treeqC)     at ($#1 + (jC)$) {\footnotesize $#2$};
            
%            \node (treeq1M)[blue]     at ($#1 + (j1) +(.2,-.5)$) {\footnotesize $c_0$};
            \node (treeq2M)[blue]     at ($#1 + (j2) + (.3,.55)$) {\footnotesize $#3$};
            \node (treeq3M)[blue]     at ($#1 + (j4) + (-.1,.55)$) {\footnotesize $#4$};
            \node (treeq4M)[blue]     at ($#1 + (j5) + (.2,.5)$) {\footnotesize $#5$};
            \node (treeqCM)[blue]     at ($#1 + (jC) + (.2,.5)$) {\footnotesize $#2$};
            
            \draw (treeqC) -> (treeq2) -> (treeq3) -> (treeq4) ->(treeq5);
%            \draw (treeq3) -> (treeq1);
        }
    }
}

\tikzset{
    pics/Query/.style n args={6}{
        code = { %
            
            \qshape
            
%            \node (treeq0)     at ($#1 + (i0)$) {\footnotesize $c_0$};
            \node (treeq1)     at ($#1 + (i1)$) {\footnotesize $#2$};
            \node (treeq2)     at ($#1 + (i2)$) {\footnotesize $#3$};
            \node (treeq3)     at ($#1 + (i3)$) {\footnotesize $#4$};
            \node (treeq4)     at ($#1 + (i4)$) {\footnotesize $#5$};
            \node (treeq5)     at ($#1 + (i5)$) {\footnotesize $#6$};
            
            \draw[->] (treeq1) -> (treeq2) -> (treeq3);
            \draw[->] (treeq5) -> (treeq4) ->(treeq3);
        }
    }
}

\tikzset{
    pics/Cx0/.style n args={2}{
        code = { %
            \tshape            
            \draw [rounded corners=2mm, #2, dashed] ($#1 + (.3,.3)$)--($#1 + (.3,-.3)$)--($#1 + (-.3,-.3)$)--($#1 + (-.3,.3)$)--cycle;
        }
    }
}

\tikzset{
    pics/Cx12/.style n args={2}{
        code = { %
            
            \tshape 
            \draw [rounded corners=5mm, #2, dashed] ($#1 + (j3) + (.5,.5)$)--($#1 + (j3) + (.5,-.5)$)--($#1 + (j2) + (-.5,-.5)$)--($#1 + (j2) + (-.5,.5)$)--cycle;
        }
    }
}

\tikzset{
    pics/Cx23/.style n args={2}{
        code = { %            
            \draw [rounded corners=5mm, #2, dashed] ($#1 + (j3) + (-.5,-.5)$)--($#1 + (j3) + (-.5,.5)$)--($#1 + (j4) + (.5,.5)$)--($#1 + (j4) + (.5,-.5)$)--cycle;
        }
    }
}

\tikzset{
    pics/Cx34/.style n args={2}{
        code = { %            
            \draw [rounded corners=5mm, #2, dashed] ($#1 + (j4) + (-.5,.5)$)--($#1 + (j4) + (.5,.5)$)--($#1 + (j5) + (.5,-.5)$)--($#1 + (j5) + (-.5,-.5)$)--cycle;
        }
    }
}

%%%%%wrong sets
\tikzset{
    pics/Cx34w/.style n args={2}{
        code = { %            
            \draw [rounded corners=5mm, #2, dotted] ($#1 + (j4) + (-.5,.5)$)--($#1 + (j4) + (.5,.5)$)--($#1 + (j5) + (.5,-.5)$)--($#1 + (j5) + (-.5,-.5)$)--cycle;
        }
    }
}
\tikzset{
    pics/Cx0w/.style n args={2}{
        code = { %
            \tshape            
            \draw [rounded corners=2mm, #2, dotted] ($#1 + (.3,.3)$)--($#1 + (.3,-.3)$)--($#1 + (-.3,-.3)$)--($#1 + (-.3,.3)$)--cycle;
        }
    }
}
%%%%%%%%%%%%%%%%%%%%%%%%
\tikzset{
    pics/CxC12w/.style n args={2}{
        code = { %            
            \draw [rounded corners=4mm, #2, dashed] ($#1 + (j3) + (.4,.4)$)--($#1 + (j3) + (.4,-.4)$)--($#1 + (j2) + (.4,-.3)$)--($#1 + (jC) + (+.4,-.4)$)
            --($#1 + (jC) + (-.4,-.4)$)--($#1 + (j2) + (-.4,.4)$)--cycle;
        }
    }
}

\tikzset{
    pics/CxC1w/.style n args={2}{
        code = { %            
            \draw [rounded corners=4mm, #2, dashed] ($#1 + (j2) + (-.4,.4)$)--($#1 + (j2) + (.4,.4)$)--($#1 + (jC) + (.4,-.4)$)--($#1 + (jC) + (-.4,-.4)$)--cycle;
        }
    }
}
%end of tikz magic

We give examples that showcase important aspects of the enumeration
algorithm for minimal partial answers with a single wildcard presented
in Section~\ref{sect:LPAsingleWildcardUpper}. % At the end of the
% section, we also make some remarks on the multi-wildcard case.

Assume that the enumeration algorithm is started on the OMQ
$Q(\bar x)=(\Omc,\Sbf,q) \in (\class{G},\class{CQ})$ where \Omc consists of
the TGDs
$$
\begin{array}{rcl}
  A(x) &\rightarrow& \exists y_1 \exists y_2\ R(y_1,y_2) \land
                     R(y_2,x) \land C(y_2) \\[1mm]
  B(x) &\rightarrow& \exists y_1 \exists y_2\ R(y_1, x) \land R(y_2,x)
                     \land C(y_1) \\[1mm]
E(x) &\rightarrow& \exists y_1\ R(x,y_1) \\[1mm]
R(x,y) &\rightarrow& L(x,x) \land L(y,y),
\end{array}
$$
the schema $\Sbf$ is $\{A, B, C, E, R\}$, % with $A$, $B$, $C$,
% $E$ unary relation symbols and $R$ a binary relation symbol
 and
where 
$q$ is the CQ
\[
\begin{array}{rcl}
q(\bar x) & {\gets} &\exists y_1 \exists y_5\ L(y_1, x_1),  R(x_1,x_2),  R(x_2,x_3), \\[1mm]
                                    &          &\ \ R(x_4,x_3), R(x_5,x_4), L(y_5,x_5), C(x_1).
\end{array}
\]
with $\bar{x} = (x_1,x_2,x_3,x_4,x_5)$.
The CQ $q$ is displayed in Figure~\ref{fig:ex-q}.
\begin{figure}
    %\begin{wrapfigure}{r}{0.45\textwidth}
    \centering
    \begin{tikzpicture}[scale=.66]

    \node (top) at (0,.2) {};
    \node (bottom) at (0,-4) {};

    %CQ q shape
    \node (b1) at (-1.5, -3) {};
    \node (a1) at (-1,-2) {};
    \node (a2) at (-.5,-1) {};
    \node (a3) at (0, 0) {};
    \node (a4) at (.5, -1) {};
    \node (a5) at (1, -2) {};
    \node (b5) at (1.5, -3) {};
    \node (c0) at (0, -3.8) {};
    
    %delimiters 
    \node (q0) at (0, 0) {};
%    \node (p0) at (4, 0) {};
%    \node (phat0) at (8, 0) {};
%    \node (of) at (2,0) {};

    % query q
    \node[quantified] (qb1)     at ($(q0) + (b1)$) {};
    \node[answer] (qa1)         at ($(q0) + (a1)$) {};
    \node (qa1A)                at ($(q0) + (a1) +(.4,-.0)$) {\footnotesize C};
    \node[answer] (qa2)         at ($(q0) + (a2)$) {};
    \node[answer] (qa3)         at ($(q0) + (a3)$) {};
    \node[answer] (qa4)         at ($(q0) + (a4)$) {};
    \node[answer] (qa5)         at ($(q0) + (a5)$) {};
    \node[quantified] (qb5)     at ($(q0) + (b5)$) {};
    %names of edges
    \node (qa1l)                at ($(q0) + (a1) + (-.0,.5)$) {\footnotesize $R$};
    \node (qb1l)                at ($(q0) + (b1) + (-.0,.5)$) {\footnotesize $L$};
    \node (qa2l)                at ($(q0) + (a2) + (-.0,.5)$) {\footnotesize $R$};
    \node (qa4l)                at ($(q0) + (a4) + (.0,.5)$) {\footnotesize $R$};
    \node (qa5l)                at ($(q0) + (a5) + (.0,.5)$) {\footnotesize $R$};
    \node (qb5l)                at ($(q0) + (b5) + (.0,.5)$) {\footnotesize $L$};
    %names of vars
    \node (qa1E)                at ($(q0) + (a1) + (.2,-.3)$) {\footnotesize $x_1$};
    \node (qb1E)                at ($(q0) + (b1) + (.2,-.3)$) {\footnotesize $y_1$};
    \node (qa2E)                at ($(q0) + (a2) + (.2,-.3)$) {\footnotesize $x_2$};
    \node (qa3E)                at ($(q0) + (a3) + (.4,0)$) {\footnotesize $x_3$};
    \node (qa4E)                at ($(q0) + (a4) + (-.2,-.3)$) {\footnotesize $x_4$};
    \node (qa5E)                at ($(q0) + (a5) + (-.2,-.3)$) {\footnotesize $x_5$};
    \node (qb5E)                at ($(q0) + (b5) + (-.2,-.3)$) {\footnotesize $y_5$};
    \node         (qc0)         at ($(q0) + (c0)$) {\footnotesize a) CQ $q$ };

    \draw[->] (qb1) -> (qa1);
    \draw[->] (qa1) -> (qa2);
    \draw[->] (qa2) -> (qa3);
    \draw[->] (qb5) -> (qa5);
    \draw[->] (qa5) -> (qa4);
    \draw[->] (qa4) -> (qa3);

    \node (cj0) at (0,-3.8) {};
\node (treeq0) at (6, 0) {};
%    \node (treeqhat0) at (4, 0) {};
\node (of) at (2,0) {};
%\node (bottom) at (0,-3.5) {};
%\node (top) at (0,.2) {};

% join tree shape
\node (j1) at (-1, -2) {};
\node (j2) at (-1,-1) {};
\node (j3) at (0, 0) {};
\node (j4) at (1, -1) {};
\node (j5) at (1, -2) {};
\node (j6) at (1, -3) {};
\node (j7) at (-1, -3) {};

%join tree of q
\node (treeq1)     at ($(treeq0) + (j1)$) {\footnotesize $C_1(x_1)$};
\node (treeq7)     at ($(treeq0) + (j7)$) {\footnotesize $L(y_1,x_1)$};
\node (treeq2)     at ($(treeq0) + (j2)$) {\footnotesize $R(x_1,x_2)$};
\node (treeq3)     at ($(treeq0) + (j3)$) {\footnotesize $R(x_2,x_3)$};
\node (treeq4)     at ($(treeq0) + (j4)$) {\footnotesize $R(x_4,x_3)$};
\node (treeq5)     at ($(treeq0) + (j5)$) {\footnotesize $R(x_5,x_4)$};
\node (treeq6)     at ($(treeq0) + (j6)$) {\footnotesize $L(y_5,x_5)$};
\node (treeqc0)   at ($(treeq0) + (cj0)$) {\footnotesize b) join tree of $q$ };

\draw (treeq7) -> (treeq1) -> (treeq2) -> (treeq3) -> (treeq4) ->(treeq5) -> (treeq6);

% join tree shape
\node (k1) at (-1, -2) {};
\node (k2) at (-1, -1) {};
\node (k3) at (-1, -0) {};
\node (k4) at (1, -0) {};
\node (k5) at (1, -1) {};
\node (k6) at (1, -2) {};
\node (k7) at (0, -1) {};
\node (k8) at (0, -2) {};

\draw[dashed]
    ($(3,0) + (top)$) -- ($(3,0) + (bottom)$);

    \end{tikzpicture}
   \caption{CQ $q$ and its join tree.}
    \label{fig:ex-q}
\end{figure}
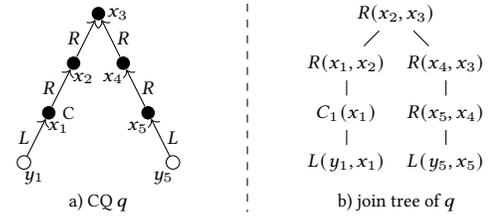
It is acyclic and free\=/connex acyclic, as witnessed by the join
trees for $q$ and its extension $\hat q$ with the atom
$\hat R(x_1,\dots,x_5)$.
The join tree for $q$ is given in Figure~\ref{fig:ex-q}. Note that the atoms that contain
only answer variables constitute a connected prefix of the join tree
of $q$. This can (almost\footnote{It can always be achieved
  when using a generalized hypertree decomposition of width~1 in place
  of a join tree, see \cite{berkholz-enum-tutorial}.}) always be achieved for CQs that are acyclic and
free-connex acyclic and is exploited in the  preprocessing phase.
%
% Note that the similar CQ
% \[
% \begin{array}{r c l}
% p(x_1 ,x_5) & {\gets} &  L(y_1, x_1),
%                         R(x_1,x_2), R(x_2,x_3), \\[1mm]
%   && \ \ R(x_4,x_3), R(x_5,x_4), L(y_5,x_5)
% \end{array}
% \]
% is acyclic, but not free\=/connex acyclic.  Indeed, the CQ $\hat{p}$
% obtained from $p$ by adding the atom $\hat{R}(x_1,x_5)$ is not acyclic.
% Figure~\ref{fig:ex-q} also depicts CQs $p$ and $\hat{p}$.
%

%\input{draws/tikzu/ex-q-join-tree}

Assume that the input database $D$ is as depicted on the left-hand
side of Figure~\ref{fig:ex-database}, where all edges represent the
relation symbol $R$.
% with active domain
% $\dom(D) = \{a, b, c, d, e\}$ and the following set of $\Sbf$\=/facts:
% \begin{itemize}
%     \item $A(a), B(b), B(d), C(b), E(a), E(e)$
%     \item $R(c, b), R(b, a), R(d, a), R(e, a)$.
% \end{itemize}

\begin{figure}
    %\begin{wrapfigure}{r}{0.45\textwidth}
    \centering
    \begin{tikzpicture}%[scale=.5]

    \node (top) at (0,2.2) {};
    \node (bottom) at (0,-3) {};

    %CQ q shape
    \node (aA) at (.4, -.3) {};
%    \node (aE) at (0,-4) {};
    \node (eE) at (1.2,-.2) {};
    \node (bB) at (.4,-.9) {};
    \node (dB) at (-.8, -.2) {};
    \node (a) at (0, 0) {};
    \node (b) at (0,-1) {};
    \node (c) at (0,-2) {};
    \node (d) at (-1, 0) {};
    \node (e) at (1, 0) {};
    \node (n1) at (-1, 1) {};
    \node (n2) at (-2, 2) {};
    \node (n2C) at (-2.3, 1.8) {};
    \node (n3) at (1, -2) {};
    \node (n3C) at (1.3, -2.3) {};
    \node (n4) at (-1, -2) {};
    \node (n5) at (-2, -1) {};
    \node (n5C) at (-1.7, -1.3) {};
    \node (n6) at (-2, 1) {};
    \node (n7) at (1, 1) {};
    \node (n8) at (2, 1) {};
    \node (c0) at (0, -2.2) {};
    
    %delimiters 
    \node (D) at (-2, 0) {};
    \node (DC) at (3, 0) {};
    \node (of) at (2, 0) {};
    
    % database D
    \node (Da)         at ($(D) + (a)$) {$a$};
    \node (Db)         at ($(D) + (b)$) {$b$};
    \node (Dc)         at ($(D) + (c)$) {$c$};
    \node (Dd)         at ($(D) + (d)$) {$d$};
    \node (De)         at ($(D) + (e)$) {$e$};
    \node (DaA)         at ($(D) + (aA)$) {$A$, $E$};
    \node (DeE)         at ($(D) + (eE)$) {$E$};
    \node (DbB)         at ($(D) + (bB)$) {$B,C$};
    \node (DdB)         at ($(D) + (dB)$) {$B$};

%    \node ( )         at ($(q0) + (a4)$) {$a$};
%    \node (qa5)         at ($(q0) + (a5)$) {};
%    \node (qb5)     at ($(q0) + (b5)$) {};
    \node (qc0)              at ($(D) + (0,-3)$) {\footnotesize database $D$ };

    \draw[->] (Db) -> (Da);
    \draw[->] (Dc) -> (Db);
    \draw[->] (Dd) -> (Da);
    \draw[->] (De) -> (Da);

        % query directed chase of D
    \node (Ca)         at ($(DC) + (a)$) {$a$};
    \node (Cb)         at ($(DC) + (b)$) {$b$};
    \node (Cc)         at ($(DC) + (c)$) {$c$};
    \node (Cd)         at ($(DC) + (d)$) {$d$};
    \node (Ce)         at ($(DC) + (e)$) {$e$};
    \node (Cn1)         at ($(DC) + (n1)$) {$n_1$};
    \node (Cn2)         at ($(DC) + (n2)$) {$n_2$};
    \node (Cn3)         at ($(DC) + (n3)$) {$n_3$};
    \node (Cn4)         at ($(DC) + (n4)$) {$n_4$};
    \node (Cn5)         at ($(DC) + (n5)$) {$n_5$};
    \node (Cn6)         at ($(DC) + (n6)$) {$n_6$};
    \node (Cn7)         at ($(DC) + (n7)$) {$n_7$};
    \node (Cn8)         at ($(DC) + (n8)$) {$n_8$};
    \node (DaA)         at ($(DC) + (aA)$) {$A$, $E$};
    \node (DeE)         at ($(DC) + (eE)$) {$E$};
    \node (DbB)         at ($(DC) + (bB)$) {$B,C$};
    \node (DdB)         at ($(DC) + (dB)$) {$B$};
    \node (DbB)         at ($(DC) + (n2C)$) {$C$};
    \node (DbB)         at ($(DC) + (n3C)$) {$C$};
    \node (DbB)         at ($(DC) + (n5C)$) {$C$};
    %    \node ( )         at ($(q0) + (a4)$) {$a$};
    %    \node (qa5)         at ($(q0) + (a5)$) {};
    %    \node (qb5)     at ($(q0) + (b5)$) {};
    %    \node (qc0)              at ($(q0) + (c0)$) {\footnotesize CQ $q$ };
    \node (qDc0)              at ($(DC) + (0,-3)$) {\footnotesize database $D_0$ };

    \draw[->] (Cb) -> (Ca);
    \draw[->] (Cc) -> (Cb);
    \draw[->] (Cd) -> (Ca);
    \draw[->] (Ce) -> (Ca);
    \draw[->] (Cn1) -> (Ca);
    \draw[->] (Cn2) -> (Cn1);
    \draw[->] (Cn3) -> (Cb);
    \draw[->] (Cn4) -> (Cb);
    \draw[->] (Cn5) -> (Cd);
    \draw[->] (Cn6) -> (Cd);
    \draw[->] (Ca) -> (Cn7);
    \draw[->] (Ce) -> (Cn8);

    \draw[dashed]
    ($(D) + (of) + (top)$) -- ($(D) + (of) + (bottom)$);

    \end{tikzpicture}
   \caption{Database $D$ and query-directed chase $D_0$.
     All edges represent relation $R$ and every constant has an
     $L$-self\=/loop
     that is not shown.}
    \label{fig:ex-database}
\end{figure}
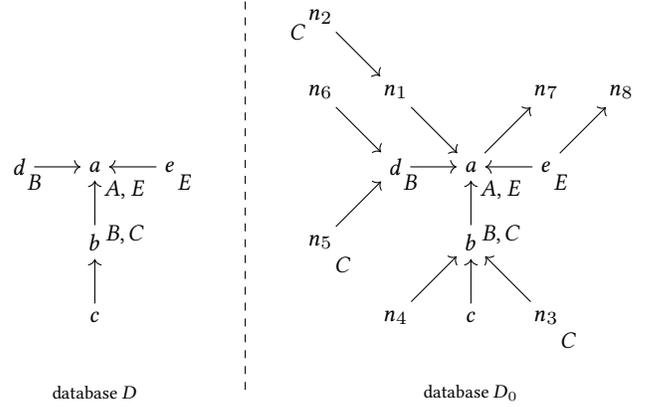

\paragraph{Preprocessing}

In the preprocessing phase, we modify the query $q$ and database
$D$ to obtain the CQ $q_2$ and database $D_2$ that are used in the 
enumeration phase. This is done in several steps. In the very first step, we
set $q_0=q$ and replace $D$ with the query-directed chase
$D_0=\mn{ch}_\Omc^q(D)$, displayed on the right-hand side of
Figure~\ref{fig:ex-database}.

\smallskip 

The next step is to construct from $q_0$ and $D_0$ a self-join free CQ
$q_1$ without quantified and a database $D_1$ that has been adjusted
accordingly. It is this step that exploits the special shape of the
join tree of $q_0$ mentioned above. In our case, $q_1$ is
$$q
_1(\bar x) \gets R_{1}(x_1,x_2),   \allowbreak
R_{2}(x_2,x_3), \allowbreak R_{4}(x_4,x_3), R_{5}(x_5,x_4), C_1(x_1).
$$
Informally, $q_1$ was obtained from $q$ by renaming relation symbols
to achieve self-join freeness and dropping atoms that involve a
quantified variable.  The join tree
of $q_1$ is the join tree of $q$ except that relation symbols in atoms
change and nodes/atoms that contain any of the variables $y_1,y_5$ are
removed. % We root the tree in the atom
% $R_{2}(x_2,x_3)$. Strange to say this HERE

\smallskip 

The database $D_1$ is shown in Figure~\ref{fig:ex-database} where, for
better readability, we only show the index $i$ of edge labels
$R_i$. Observe that the constant $n_8$ was removed and that edges
are now multi\=/edges. % For instance, between $a$ and $b$ there are the
% edges
% $$R_{1}(b,a), R_{2}(b,a), R_{4}(b,a),R_{5}(b,a).$$ %  and between $n_4$
% % and $b$ the edges $$R_{1}(n_4,b), R_{5}(n_4,b).
% % $$
To get an intuition of the construction of $D_1$, consider the fact
$R(n_2,n_1)$ in $D_0$. In principle, any of the four $R$-atoms in
$q_0$ can map to it, and in $q_1$ the relation symbol $R$ in those
atoms has been renamed to $R_1$, $R_2$, $R_4$, and $R_5$,
respectively. Thus, we should be prepared to include in $D_1$ the fact
$R_i(n_2,n_1)$ for all $i \in \{1,2,4,5\}$.  However, a closer
inspection shows that the atom $R(x_2,x_3)$ in $q_0$ cannot map to the
fact $R(n_2,n_1)$ in $D_0$ since then $x_1$ would have to be mapped to
an $R$-predecessor of $n_2$, which does not exist. A similar
observation holds for the atom $R(x_4,x_3)$ in $q_0$ and thus we only
include in $D_1$ the facts $R_1(n_2,n_1)$ and $R_5(n_2,n_1)$.  The
`right' facts to include are identified during a bottom-up walk over
the join-tree of $q_0$. Note that the relation symbols $A$, $B$, have
been dropped since they do not occur in $q_0$. 
%
% I found this too hard to understand
%
% Notice that, after
% this step, we can constrict the database so it contains only the facts
% that use relation symbols present in $q$.  Even further, we can remove
% all facts using unary relation symbols, as they now are implied by the
% binary relations.  
%
\begin{figure}[h]
    %\begin{wrapfigure}{r}{0.45\textwidth}
    \centering
    \begin{tikzpicture}[scale=.5]
    
    %delimiters 
    \node (DC) at (0, 0) {};
    
    %database D_2 shape
%    \node (aA) at (.7, -.5) {};
%    \node (aE) at (0,-4) {};
%    \node (eE) at (2.3,-.5) {};
%    \node (bB) at (.5,-1.7) {};
%    \node (dB) at (-1.7, -.5) {};
    \node (a) at (0, 0) {};
    \node (b) at (0,-2) {};
    \node (c) at (0,-4) {};
    \node (d) at (-4, 0) {};
    \node (e) at (4, 0) {};
    \node (n1) at (-2, 1.5) {};
    \node (n2) at (-4.5, 3) {};
    \node (n3) at (3, -3) {};
    \node (n4) at (-3, -3) {};
    \node (n5) at (-5, -1.5) {};
    \node (n6) at (-5, 1.5) {};
    \node (n7) at (2.7, 1.5) {};
    \node (n8) at (4.5, 1.5) {};
%    \node (c0) at (0, -4.2) {};

    % database D2
    \node (Ca)         at ($(DC) + (a)$) {$a$};
    \node (Cb)         at ($(DC) + (b)$) {$b$};
    \node (Cba)        at ($(DC) + (b) +(1,1)$) {\scriptsize 1,2,4,5};%{\tiny $R_1,R_2,R_4, R_5$};
    \node (Cc)         at ($(DC) + (c)$) {$c$};
    \node (Ccb)        at ($(DC) + (c) +(.3,.9)$) {\scriptsize 5}; %{\tiny $R_5$};
    \node (Cd)         at ($(DC) + (d)$) {$d$};
    \node (Cda)        at ($(DC) + (d) +(1.5,-.4)$) {\scriptsize 1,2,4,5}; %{\tiny $R_1{,}R_2{,}R_4{,}R_5$};
   %\node (Cda)        at ($(DC) + (d) +(1.2,-.4)$) {\tiny $1{,}2{,}4{,}5$}; %{\tiny $R_1{,}R_2{,}R_4{,}R_5$};
    \node (Ce)         at ($(DC) + (e)$) {$e$};
    \node (Cea)        at ($(DC) + (e) +(-1.5,-.4)$){\scriptsize 5}; %{\tiny $R_5$};
    \node (Cn1)         at ($(DC) + (n1)$) {$n_1$};
    \node (Cn1a)        at ($(DC) + (n1) +(1.8,-.6)$) {\scriptsize 2,4}; %{\tiny $R_2,R_4$};
    \node (Cn2)         at ($(DC) + (n2)$) {$n_2$};
    \node (Cn2n1)       at ($(DC) + (n2) +(1.5,-.6)$) {\scriptsize 1,5}; %{\tiny $R_1,R_5$};
    \node (Cn3)         at ($(DC) + (n3)$) {$n_3$};
    \node (Cn3b)        at ($(DC) + (n3) +(-.5,.6)$) {\scriptsize 1,5}; %{\tiny $R_1, R_5$};
    \node (Cn4)         at ($(DC) + (n4)$) {$n_4$};
    \node (Cn4b)        at ($(DC) + (n4) +(.8,.6)$) {\scriptsize 5}; %{\tiny $R_5$};
    \node (Cn5)         at ($(DC) + (n5)$) {$n_5$};
    \node (Cn5d)        at ($(DC) + (n5) +(1,.5)$) {\scriptsize 1,5}; %{\tiny $R_1,R_5$};
    \node (Cn6)         at ($(DC) + (n6)$) {$n_6$};
    \node (Cn6d)        at ($(DC) + (n6) +(.9,-.4)$) {\scriptsize 5}; %{\tiny $R_5$};
    \node (Cn7)         at ($(DC) + (n7)$) {$n_7$};
    \node (Can7)        at ($(DC) + (n7) +(-.6,-.9)$) {\scriptsize 2,4}; %{\tiny $R_2, R_4$};
    \node (DbB)         at ($(DC) + (b) + (.4,.3)$) {\scriptsize $C_1$};
    \node (Dn2)         at ($(DC) + (n2) + (-.35,-.4)$) {\scriptsize $C_1$};
    \node (Dn3)         at ($(DC) + (n3) + (.4,-.4)$) {\scriptsize $C_1$};
    \node (Dn4)         at ($(DC) + (n5) + (-.35,-.4)$) {\scriptsize $C_1$};
%    \node (Cn8)         at ($(DC) + (n8)$) {$n_8$};
%    \node (DaA)         at ($(DC) + (aA)$) {$A$, $E$};
%    \node (DeE)         at ($(DC) + (eE)$) {$E$};
%    \node (DbB)         at ($(DC) + (bB)$) {$B$};
%    \node (DdB)         at ($(DC) + (dB)$) {$B$};
    %    \node ( )         at ($(q0) + (a4)$) {$a$};
    %    \node (qa5)         at ($(q0) + (a5)$) {};
    %    \node (qb5)     at ($(q0) + (b5)$) {};
    %    \node (qc0)              at ($(q0) + (c0)$) {\footnotesize CQ $q$ };
%    \node (Cc0)   at ($(DC) + (c0)$) {\footnotesize database  $D_1$ };
    
    \draw[->] (Cb) -> (Ca);
    \draw[->] (Cc) -> (Cb);
    \draw[->] (Cd) -> (Ca);
    \draw[->] (Ce) -> (Ca);
    \draw[->] (Cn1) -> (Ca);
    \draw[->] (Cn2) -> (Cn1);
    \draw[->] (Cn3) -> (Cb);
    \draw[->] (Cn4) -> (Cb);
    \draw[->] (Cn5) -> (Cd);
    \draw[->] (Cn6) -> (Cd);
    \draw[->] (Ca) -> (Cn7);
%    \draw[->] (Ce) -> (Cn8);
    \end{tikzpicture}
   \caption{Database $D_1$. Edges are labeled with indices of the relation symbols $R_1,R_2,R_4,R_5$
        that constitute the edge. %Unary relations in $D_1$ are ommited.
    }
    \label{fig:ex-D1}
\end{figure}
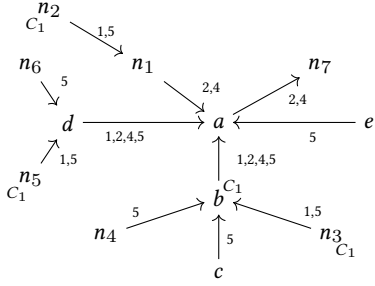

\smallskip

%In the last step of preprocessing, we introduce a fresh unary relation
%symbol $P$ and increase the arity of each relation symbol (except
%$P$) by one. We obtain $q_2$ from $q_1$ by adding the atom $P_0(x_0)$
%where $x_0$ is a fresh answer variable and putting $x_0$ into the
%additional (first) position in each atom. $D_2$ is obtained from $D_1$
%by adding a fact $P_0(c_0)$ where $c_0$ is a fresh constant and
%putting $c_0$ into the additional (first) position in each fact. Thus
%we obtain
%%
%$$
%\begin{array}{rcl}
%  q_2(x_0\bar x)&\gets& R_{1}(x_0,x_1,x_2),   
%                        R_{2}(x_0,x_2,x_3),\\[1mm]
%&&R_{4}(x_0,x_4,x_3), R_{5}(x_0,x_5,x_4), C_1(x_0,x_1).
%\end{array}
%$$
%%
%A join tree of $q_2$ is shown in Figure~\ref{fig:ex-q2-join-tree}.
%We refrain from depicting $D_2$, it is structurally still very similar
%to $D_1$ except that the fresh constant $c_0$ is now shared by
%all facts.
%
% From now on,
% we shall work with CQ $q_2$ and database $D_2$, and the notions
% of join trees, constants, and guarded set refer to $q_2$ or $D_2$.
% {\color{blue}vague. which is which?}

\begin{figure}[h]
    %\begin{wrapfigure}{r}{0.45\textwidth}
    \centering
    \begin{tikzpicture}[scale=.5]
    
    %delimiters 
    \node (treeq0) at (0,0) {};
    \node (c0) at (0,-5) {};
    
    % join tree shape
%    \node (j1) at (0, 2) {};
    \node (j2) at (-2,-2) {};
    \node (j3) at (0, 0) {};
    \node (j4) at (2, -2) {};
    \node (j5) at (2, -4) {};
    \node (jC) at (-2, -4) {};
%    \node (j6) at (1, -3) {};
    
    %join tree of q
%    \node (treeq1)     at ($(treeq0) + (j1)$) {\footnotesize $P(x_0)$};
    \node (treeq2)     at ($(treeq0) + (j2)$) {\footnotesize $R_1(x_1,x_2)$};
    \node (treeq3)     at ($(treeq0) + (j3)$) {\footnotesize $R_2(x_2,x_3)$};
    \node (treeq4)     at ($(treeq0) + (j4)$) {\footnotesize $R_4(x_4,x_3)$};
    \node (treeq5)     at ($(treeq0) + (j5)$) {\footnotesize $R_5(x_5,x_4)$};
    \node (treeqC)     at ($(treeq0) + (jC)$) {\footnotesize $C_1(x_1)$};
%    \node (av3)        at ($(treeq0) + (j3) + (.4,1)$) {\footnotesize {\color{blue} $x_0$}};
    \node (av2)        at ($(treeq0) + (j2) + (-.3,1)$) {\footnotesize {\color{blue} $x_3$}};
    \node (av4)        at ($(treeq0) + (j4) + (.2,1)$) {\footnotesize {\color{blue} $x_3$}};
    \node (av5)        at ($(treeq0) + (j5) + (.6,1)$) {\footnotesize {\color{blue} $x_4$}};
    \node (avC)        at ($(treeq0) + (jC) + (-.6,1)$) {\footnotesize {\color{blue} $x_1$}};
%    \node (treeq6)     at ($(treeq0) + (j6)$) {\footnotesize $y_5x_5$};
%    \node (treeqc0)   at ($(treeq0) + (c0)$) {\footnotesize join tree of $q_2$ };
    
    \draw  (treeqC) -> (treeq2) -> (treeq3) -> (treeq4) ->(treeq5);
%    \draw (treeq3) -> (treeq1);
    \end{tikzpicture}
    \caption{Join tree of $q_2$. The predecessor variables of each atom
      are shown on the incoming edge of the atom.
   }
    \label{fig:ex-q2-join-tree}
\end{figure}
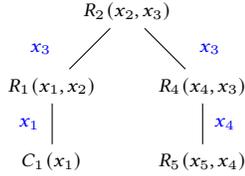

\paragraph{Lists of  progress trees.}

The last step of the preprocessing phase is to create the lists
$\mn{trees}(v,h)$ of progress trees for each atom $v$ in $q_2$ and
each predecessor map $h$ for $v$. Recall that by the latter we mean a
function $h:\bar z \rightarrow \mn{adom}(D_2) \setminus N$ whose range
is a guarded set in~$D_2$, and where $\bar z$ are the predecessor
variables in $v$. For brevity, we represent $h$ in the form
$z_1\cdots z_n \mapsto c_1\cdots c_n$ when $z_1,\dots,z_n$ are the
variables in $\bar z$ and $h(z_i)=c_i$ for $1 \leq i \leq n$; this
becomes $\varepsilon \mapsto \varepsilon$ when $\bar z$ is the empty
tuple. For the join tree of $q_1$ with marked predecessor variables, see~Figure~\ref{fig:ex-q2-join-tree}.

Also recall that a progress tree is a pair $(q,g)$ with CQ $q$ a
subtree of $q_2$ and $g$ a function from $\mn{var}(q)$ to
$(\mn{adom}(D_2)\setminus N) \cup \{ \ast \}$ that must satisfy
Conditions (1)-(4) given in Section~\ref{sect:LPAsingleWildcardUpper}.
We represent the function $g$ in the same way as predecessor maps.
Examples of progress trees include
$$
  (R_1(x_1,x_2), x_1x_1x_2 \mapsto ba)
$$
and
$$
(R_2(x_2,x_3) \wedge R_1(x_1,x_2) \wedge C_1(x_1) ,x_1x_2x_3 \mapsto
{\ast}{\ast}a).
$$
The reader is invited to verify that the relevant Conditions~(1)-(4)
are all satisfied for these progress trees. Intuitively, the second
progress tree $(q,g)$ above describes an `excursion' of the part $q$
of $q_2$ into the `null part' of $D_2$. This excursion consists of
mapping $x_1$ to $n_2$, $x_2$ to $n_1$, and
$x_3$ to $a$.

Let us review two non-examples progress trees, starting
with
%
% $$(x_0x_2x_3 \mapsto c_0{\ast}{\ast},
% R_2(x_0,x_2,x_3)).$$ This is not a progress tree as the predecessor
% variables $x_2$ in its root atom (which is of course also the only
% atom) is mapped to $\ast$, and thus Condition~(1) is violated. As
% another example,
%
$$
(R_4(x_4,x_3)\wedge R_5(x_5,x_4),
x_3x_4x_5 \mapsto abc)
$$
which is not a progress tree as the predecessor variable
$x_4$ of atom $R_5(x_5,x_4)$ is mapped to
`$\ast$' and thus Condition~(2) is violated.
Next consider
$$
(R_4(x_4,x_3) \wedge R_5(x_5,x_4),
x_3x_4x_5 \mapsto a{\ast}c)
$$
which is not a progress tree because there is no guarded set in $D_2$ that
contains $a$ and $c$, and thus Condition~(4) is violated.

\smallskip

We now give all the lists $\mn{trees}(v,h)$ that are computed in the
preprocessing phase.  For brevity, we represent progress trees $(q,g)$
as the CQ $q$ in which every variable $x$ was replaced with
$g(x)$. List items are separated by `;'. The lists are:
\begin{itemize}
%    \item atom $v=P(x_0)$, no predecessor variables {
%        \begin{itemize}
%            \item $\mn{trees}(v,\varepsilon \mapsto \varepsilon) = [ P(c_0)]$
%        \end{itemize}
%    }
    \item atom $v=R_1(x_1, x_2)$ with a predecessor variable $x_2$ {
        \begin{itemize}
%            \item for $x_0 x_2 \mapsto c_0 a $: $\mn{list} = [\{R_1(c_0,b,a)\}, \{R_1(c_0,d,a)\}]$
            \item $\mn{trees}(v, x_2 \mapsto a) = [R_1(b,a); R_1(d,a)]$
            \item $\mn{trees}(v, x_2 \mapsto b) = [R_1(\ast, b) \wedge C_1(\ast)]$
            \item $\mn{trees}(v, x_2 \mapsto c) = []$
            \item $\mn{trees}(v, x_2 \mapsto d) = [R_1(\ast, d) \wedge C_1(\ast)]$
            \item $\mn{trees}(v, x_2 \mapsto e) = []$
        \end{itemize}
    }
    \item atom $v=R_2(x_2, x_3)$ no predecessor variables {
        \begin{itemize}
            \item $\mn{trees}(v,\emptyset \mapsto \emptyset) = [ R_2(b,a); R_2(d,a);$\\
            $R_2(\ast,a) \wedge R_1(\ast,\ast)  \wedge C_1(\ast);
            R_2(a,\ast) \wedge R_4(a,\ast)
            ]$
        \end{itemize}
    }
    \item atom $v=R_4(x_4, x_3)$ with a predecessor variable $x_3$ {
    \begin{itemize}
        \item $\mn{trees}(v,x_3 \mapsto \beta) = [] $ for $\beta \in \{b,c,d,e\}$
        \item $\mn{trees}(v,x_3 \mapsto c_0 a) = [R_4(b,a); R_4(d,a);$\\
                                        $R_4(\ast,a) \wedge R_5(\ast,\ast)]$
    \end{itemize}
}
    \item atom $v=R_5(x_5, x_4)$ with a predecessor variable $x_4$ {
    \begin{itemize}
        \item $\mn{trees}(v,x_2 \mapsto c) = []$
        \item $\mn{trees}(v,x_2 \mapsto a) = [R_5(b,a); R_5(d,a); R_5(e,a)]$
        \item $\mn{trees}(v,x_2 \mapsto b) = [R_5(c,b); R_5(\ast, b)]$
        \item $\mn{trees}(v,x_2 \mapsto d) = [R_5(\ast, d)]$
        \item $\mn{trees}(v,x_2 \mapsto e) = []$
    \end{itemize}

    \item atom $v=C_1(x_1)$ with a predecessor variable $x_1$ {
        \begin{itemize}
            \item $\mn{trees}(v,x_1 \mapsto \beta) = [] $ for $\beta \in \{a,c,d,e\}$
            \item $\mn{trees}(v,x_1 \mapsto b) = [C_1(b)]$
        \end{itemize}
    }
}
\end{itemize}
All the remaining lists are empty.  The lists above are sorted in
database preferring order, as required, and thus we are ready for the
enumeration phase.

\paragraph{Enumeration and pruning}

In the enumeration
phase, % for minimal partial answers with a single wildcard
we traverse the join tree of $q_1$ in a depth-first
fashion, assembling a minimal partial answer to $q_1$ on $D_1$.  Once
such an answer is found, we output it and execute pruning, then
backtrack in a systematic way and re-start answer assemblage to
produce the next answer, and so on.

In our example, there are no complete answers. The first partial answer generated is
$\bar{c}^{\ast} = {\ast} b a b c$. The
answer $\bar{c}^{\ast}$ is displayed on the left-hand side of
Figure~\ref{fig:ex-answers-join-trees}, inside the join tree for
$q_1$. The blue boxes indicate the progress trees that have been used
in assembling the answer $\bar{c}^{\ast}$.

\begin{figure}
  \scalebox{1}{
    %\begin{wrapfigure}{r}{0.45\textwidth}
    \centering
    \begin{tikzpicture}
    
    \node (w1) at (-1,0){};
    \node (w2) at (3,0) {};
%    \node (w3) at (7,0) {};
    \node (caption1) at ($(w1) + (0,-4)$) {Answer ${\ast} b a b c$};
    \node (caption2) at ($(w2) + (0,-4)$) {Answer ${\ast} {\ast} a b c$};
%    \node (caption3) at ($(w3) + (0,-4)$) {Answer $c_0 {\ast} {\ast} a {\ast} c$};

    \pic { PiecePA={(w1)}{{\ast}}{b}{a}{b}{c}};
%    \pic { Cx0={(w1)}{blue}};
%    \pic { Cx0={($(w1) + (j2)$)}{blue}};
    \pic { Cx0={($(w1) + (j3)$)}{blue}};
    \pic { Cx0={($(w1) + (j4)$)}{blue}};
    \pic { Cx0={($(w1) + (j5)$)}{blue}};
    \pic { CxC1w={(w1)}{blue}};

    \pic { PiecePA={(w2)}{{\ast}}{{\ast}}{a}{b}{c}};
%    \pic { Cx0={(w2)}{blue}};
    \pic { Cx0={($(w2)+ (j4)$)}{blue}};
    \pic { Cx0={($(w2)+ (j5)$)}{blue}};
%    \pic { Cx12={(w2)}{blue}};
    \pic { CxC12w={(w2)}{blue}};
%    \pic { Cx12={(w2)}{blue}};
%    \pic { Cx0w={($(w2) + (j2)$)}{red}};
%    \pic { Cx34w={(w2)}{red}};
%    

%    \pic { PiecePA={(w3)}{{\ast}}{{\ast}}{a}{{\ast}}{c}};
%    \pic { Cx0={(w3)}{blue}};
%    \pic { Cx12={(w3)}{blue}};
%    \pic { Cx34w={(w3)}{red}};

    \draw[dashed] (1,0)--(1, -4.5);
%    \draw[dashed] (5,1)--(5, -4.5);
%
    \end{tikzpicture}
}
\caption{Three least partial answers, inside join tree of~$q_2$.}
\label{fig:ex-answers-join-trees}
\end{figure}
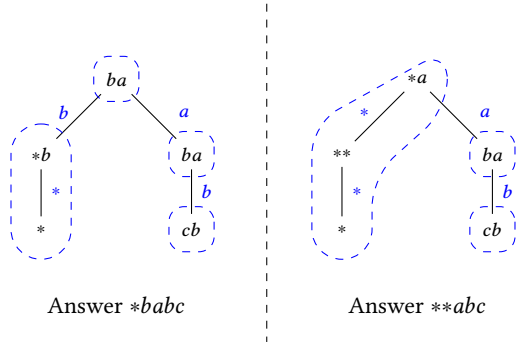

Let us now consider pruning with $\bar{c}^{\ast}$. Informally, we
consider all progress trees $(q,g)$ such that $q$ is some subtree
of $q_2$ and $g$ can be obtained by starting with $\bar x
\mapsto \bar c^\ast$, then restricting to the variables in
$\mn{var}(q)$, and then switching at least one variable from a
non-wildcard to a wildcard. One example of such a progress
tree is 
\[
R_2(\ast,a) \wedge R_1(\ast, \ast)  \wedge C_1(\ast).
\]
Pruning removes this tree from
$\mn{trees}(R_2(x_2,x_3),\emptyset \mapsto \emptyset)$.  One consequence of
this pruning that the partial answer ${\ast}{\ast}abc$ displayed
on the right of Figure~\ref{fig:ex-answers-join-trees}, which is not
a minimal partial answer, is not output in the enumeration
phase.

% After this pruning subroutine ends, the greedy algorithm will alternate between generating
% the next (minimal) partial answer and running a new pruning subroutine, until no new answers can be generated.

\ifbool{arxive}{
\paragraph{Minimal partial answers with multi-wildcards.}

We also briefly consider the enumeration algorithm for minimal partial
answers with multi-wildcards from Section~\ref{sect:enummulti},
illustrating in particular the necessity of using \emph{cones}.
Recall that a cone of a wildcard tuple $\bar{a}^{\ast}$ is the set of
all multi\=/wildcard tuples $\bar{b}^{\Wmc}$ such that the wildcard
tuple $\bar{b}^{\ast}$ obtained from $\bar{b}^{\Wmc}$ by replacing all
named wildcards by $\ast$ satisfies
$\bar{a}^{\ast} \preceq \bar{b}^{\ast}$.

Intuitively, the enumeration algorithm for multiple wildcards starts
the enumeration algorithm for a single wildcard as a black box
procedure and whenever the black box generates a minimal partial answer
$\bar{a}^{\ast}$ then it outputs the multi\=/wildcard minimal partial answers
from $\text{cone}(\bar{a}^{\ast})$, with some bookkeeping to
prevent repetition.

Consider again the CQ $q_1$ and database $D_1$.  As mentioned before,
the first generated answer is ${\ast} b a b c$ and pruning removes
from $\mn{trees}(R_2(x_2,x_3),\emptyset \mapsto \emptyset)$ the progress tree
$R_2(\ast,a) \wedge R_1(\ast, \ast)  \wedge C_1(\ast)$. Apart from preventing
the partial answer ${\ast}{\ast}abc$ to be output as noted above,
this also suppresses the partial answer ${\ast} b a b {\ast}$
(which is correct, as it is not a minimal partial answer). In contrast,
it is not hard to check that ${\ast_1} b a b {\ast_1}$ is a minimal
partial answer and thus must be output by the enumeration procedure
for multi-wildcards.

A naive version of the procedure without cones would simply look at
each minimal partial answer with a single wildcard $\bar c^\ast$ output
by the black box and then output all multi-wildcard answers
$\bar c^\Wmc$ obtained from $\bar c^\ast$ by replacing each occurrence
of the wildcard with some wildcard from \Wmc. Clearly, such a naive
version would miss the minimal partial answer
${\ast_1} b a b {\ast_1}$. However,
${\ast} b a b c \prec {\ast} b a b {\ast}$ and thus
${\ast_1} b a b {\ast_1} \in \mn{cone}({\ast} b a b c)$.  Our
more refined algorithm therefore adds ${\ast_1} b a b {\ast_1}$
to the list $L$ when processing the tuple ${\ast} b a b c$
and outputs it at the end of its run along with the other tuples in $L$.
}

\end{document}